\documentclass[onecolumn, draftcls, 11pt]{IEEEtran}

\usepackage[utf8]{inputenc} 

\usepackage{geometry} 
\usepackage{graphicx}
\usepackage{subfig}

\usepackage{booktabs} 
\usepackage{array} 
\usepackage{paralist} 
\usepackage{verbatim} 
\usepackage{subfig} 
\usepackage{amsmath}
\usepackage{amssymb}
\usepackage{amsthm}
\usepackage{amsfonts}
\usepackage{bbm}
\usepackage{color}
\usepackage{mathrsfs}  
\usepackage{bm}
\usepackage{mathtools}
\usepackage{dsfont}

\allowdisplaybreaks

\usepackage{fancyhdr} 
\newcommand{\R}{\mathbb{R}}
\newcommand{\C}{\mathbb{C}}
\newcommand{\Z}{\mathbb{Z}}
\newcommand{\N}{\mathbb{N}}
\newcommand{\Q}{\mathbb{Q}}

\newcommand{\pre}{\mrm{par}}
\newcommand{\anc}{\mrm{anc}}
\newcommand{\In}{\mrm{In}}

\newcommand{\dom}{\m{D}}
\newcommand{\sgn}{\mrm{sgn}}

\newcommand{\lvl}{\mrm{lv}}

\newcommand{\coleqq}{\vcentcolon=}

\newcommand{\rhoisom}{\stackrel{\rho}{\sim}}

\newcommand{\sigisom}{\stackrel{\sigma}{\sim}}
\newcommand{\outmap}[2]{\left\langle{#1}\right\rangle^{#2}}
\newcommand{\OnoA}[2]{\langle{#1}\rangle^{#2}}
\newcommand{\outmapT}[3]{\left\langle{#1}\right\rangle^{#2,\,#3}}
\newcommand{\OTnoA}[3]{\langle{#1}\rangle^{#2,\,#3}}

\newcommand{\vin}{v_{\hspace{1pt}\mathrm{in}}}

\newcommand{\m}{\mathcal}
\newcommand{\mrm}{\mathrm}
\newcommand{\wtd}{\widetilde}
\newcommand{\LtCla}[2]{\Sigma_{#1,#2}}
\newcommand{\Vin}{V_{\mrm{in}}}
\newcommand{\Vout}{V_{\mrm{out}}}
\renewcommand{\Re}{\mrm{Re}}
\renewcommand{\Im}{\mrm{Im}}

\theoremstyle{theorem}
\newtheorem{theorem}{Theorem}
\newtheorem{lemma}{Lemma}
\newtheorem{prop}{Proposition}

\theoremstyle{definition}
\newtheorem{definition}{Definition}

\theoremstyle{remark}

\theoremstyle{definition}

\begin{document}

\title{\huge Affine Symmetries and Neural Network Identifiability}

\author{\IEEEauthorblockN{Verner Vla\v{c}i\'{c} and Helmut B\"olcskei  \\[0.4cm]}
\IEEEauthorblockA{
Dept. of EE and Dept. of Math., ETH Zurich, Switzerland\\
 [-0.1cm] Email: vlacicv@mins.ee.ethz.ch, hboelcskei@ethz.ch}
}

\maketitle
\thispagestyle{empty}

\begin{abstract}
	We address the following question of neural network identifiability: Suppose we are given a function $f:\R^m\to\R^n$ and a nonlinearity $\rho$. Can we specify the architecture, weights, and biases of all feed-forward neural networks with respect to $\rho$ giving rise to $f$?
Existing literature on the subject suggests that the answer should be yes, provided we are only concerned with finding networks that satisfy certain ``genericity conditions''. Moreover, the identified networks are mutually related by symmetries of the nonlinearity. For instance, the $\tanh$ function is odd, and so flipping the signs of the incoming and outgoing weights of a neuron does not change the output map of the network.
The results known hitherto, however, apply either to single-layer networks, or to networks satisfying specific structural assumptions (such as full connectivity), as well as to specific nonlinearities.
In an effort to answer the identifiability question in greater generality, we consider arbitrary nonlinearities with potentially complicated affine symmetries, and we show that the symmetries can be used to find a rich set of networks giving rise to the same function $f$. The set obtained in this manner is, in fact, exhaustive (i.e., it contains all networks giving rise to $f$) unless there exists a network $\m{A}$ ``with no internal symmetries'' giving rise to the identically zero function. This result can thus be interpreted as an analog of the rank-nullity theorem for linear operators. We furthermore exhibit a class of ``$\tanh$-type'' nonlinearities (including the $\tanh$ function itself) for which such a network $\m{A}$ does not exist, thereby solving the identifiability question for these nonlinearities in full generality and settling an open problem posed by Fefferman in \cite{Fefferman1994}. Finally, we show that this class contains nonlinearities with arbitrarily complicated symmetries.
\end{abstract}

\section{Introduction}

\subsection{Background and previous work}

Deep neural network learning has become a highly successful machine learning method employed in a wide range of applications such as optical character recognition \cite{lecun:1995MNIST},
image classification \cite{Krizhevsky2012Imagenet}, speech recognition \cite{Hint2012acoustic}, and generative models \cite{GAN2014}. Neural networks are typically defined as concatenations of affine maps between finite dimensional spaces and  nonlinearities applied coordinatewise, and are often studied as mathematical objects in their own right, for instance in approximation theory \cite{Boelcskei2019,Mallat2012,Petersen2018,Wiatowski2018} and in control theory \cite{Albertini1993,SontagMTNS1993}.

In data-driven applications \cite{Bengio2016,LeCun2015} the parameters of a neural network (i.e., the coefficients of the network's affine maps) need to be learned based on training data. In many cases, however, there exist multiple networks with different parameters, or even different architectures, giving rise to the same input-output map on the training set. These networks might differ, however, in terms of their generalization performance. In fact, even if several networks with differing architectures realize the same map on the entire domain, some of them might be easier to arrive at through training than others. It is therefore of interest to understand the ways in which a given function can be parametrized as a neural network.
 Specifically, we ask the following question of identifiability: Suppose that we are given a function $f:\R^m\to\R^n$ and a nonlinearity $\rho$. Can we specify the network architecture, weights, and biases of all feed-forward neural networks with respect to $\rho$ realizing $f$?
For the special case of the $\tanh$ nonlinearity, this question was first addressed in \cite{Sussman1992} for single-layer networks, and in \cite{Fefferman1994} for multi-layer networks satisfying certain ``genericity conditions'' on the architecture, weights, and biases. The identifiability question for single-layer networks with nonlinearities satisfying the so-called ``independence property'' (corresponding to the absence of non-trivial affine symmetries according to our Definition \ref{def:sym}) was solved in \cite{Sontag1993}, whereas the recent paper \cite{Vlacic2019} reports the first known identifiability result for multi-layer networks with minimal conditions on the architecture, weights, and biases, albeit with artificial nonlinearities designed to be ``highly asymmetric''.
We also remark that the identifiability of recurrent  single-layer networks was considered in \cite{Albertini1993} and \cite{SontagMTNS1993}.

It is important to note that all aforementioned results, as well as the results in the present paper, are concerned with the identifiability of networks given knowledge of the function $f$ on its entire domain. This corresponds to characterizing the fundamental limit on nonuniqueness in neural network representation of functions. Specifically, the nonuniqueness can only be richer if we are interested in networks that realize $f$ on a proper subset of $\R^m$, such as a finite (training) sample $\{x_1,\dots, x_m\}\subset\R^m$.
Moreover, we do not address neural network reconstruction, i.e., we do not provide a procedure for constructing an instance of a network realizing a given function $f$, but rather focus on building a theory that systematically describes how the neural networks realizing $f$ relate to one another. We do this in full generality for networks with ``$\tanh$-type'' nonlinearities (including the $\tanh$ function itself), settling an open problem posed by Fefferman in  \cite{Fefferman1994}.

Recent results on neural network  reconstruction on samples can be found in \cite{Fornasier2018}, \cite{Fornasier2019} for shallow networks and in \cite{Rolnick2020} for ReLU networks of arbitrary depth.

\subsection{Affine symmetries as a template for neural network nonuniqueness}

In order to develop intuition on the identifiability of general neural networks, we follow \cite{Sussman1992} and \cite{Sontag1993} and begin by considering single-layer networks. To this end, let $\rho:\R\to\R$ be a nonlinearity, and let
\begin{equation}\label{eq:proto-sym-2nets}
\OnoA{\m{N}}{\rho}\coleqq\sum_{p=1}^{n}\lambda_p\,\rho(\omega_p\cdot\,+\,\theta_p) \;+\lambda\qquad \text{and}\qquad \OnoA{\m{N}'}{\rho}\coleqq \sum_{p=1}^{n'}\lambda_p'\,\rho(\omega_p'\cdot\,+\,\theta_p') \;+\lambda'
\end{equation}
be the maps realized by the single-layer networks $\m{N}$ and $\m{N}'$, both with nonlinearity $\rho$. Suppose that these networks realize the same function, i.e.,
\begin{equation*}
\sum_{p=1}^{n}\lambda_p\, \rho(\omega_pt+\theta_p) \;+\lambda=\sum_{p=1}^{n'}\lambda_p'\, \rho(\omega_p't+\theta_p') \;+\lambda',
\end{equation*}
for all $t\in\R$. This is equivalent to the following linear dependency relation between the constant function $\bm{1}:\R\to \R$ taking on the value 1 and affinely transformed copies of $\rho$:
\begin{equation}\label{eq:proto-sym}
\sum_{p=1}^{n}\lambda_p\,\rho(\omega_pt+\theta_p)\;-\sum_{p=1}^{n'}\lambda_p'\,\rho(\omega_p't+\theta_p')= (\lambda'-\lambda)\bm{1}(t),
\end{equation}
for all $t\in\R$.

We consider two concrete nonlinearities to demonstrate how linear dependency relations of the form \eqref{eq:proto-sym} lead to formally different networks realizing the same function. First, let $\rho=\tanh$. Then, as $\tanh(t)=-\tanh(-t)$, for all $t\in\R$, we have
\begin{equation*}
\sum_{p=1}^{n}\lambda_p\,\tanh(\omega_p\cdot\,+\theta_p)-\sum_{p=1}^{n}s_p\lambda_p\,\tanh(s_p\omega_p\cdot\,+s_p\theta_p)=0,
\end{equation*}
 for every choice of signs $s_p\in\{-1,+1\}$, $p\in\{1,\dots,n\}$, i.e., with the notation in \eqref{eq:proto-sym-2nets}, we have $\OnoA{\m{N}}{\tanh}=\OnoA{\m{N}'}{\tanh}$ with $n'=n$, $\lambda'=\lambda$, $\lambda_p'=s_p \lambda_p$, $\omega_p'=s_p\omega_p$, and $\theta_p'=s_p \theta_p$, for all $p\in\{1,\dots, n\}$. Underlying this nonuniqueness is the simple insight that $\tanh(t)=-\tanh(-t)$ can be rewritten as $\tanh(t)+\tanh(-t)=0$, which, in turn, can be interpreted as a single-layer network with two neurons, mapping every input to output 0.
 
  For a more intricate example, consider the clipped rectified linear unit (CReLU) nonlinearity given by $\rho_{c}(t)=\min\{1,\max\{0,t\}\}$, and note that
\begin{equation}\label{eq:intro-symm-1} 
\rho_{c}\,(t)-{\textstyle \frac{1}{2}}\rho_{c}\,(2t)-{\textstyle \frac{1}{2}}\rho_{c}\,(2t-1)=0, \quad\text{for all }t\in\R,
\end{equation}
 corresponds to a single-layer network with three neurons mapping every input to output 0. This can be rewritten as $\rho_{c}\,(t)={\textstyle \frac{1}{2}}\rho_{c}\,(2t)+{\textstyle \frac{1}{2}}\rho_{c}\,(2t-1)$ and applied recursively to yield
\begin{equation*} 
 \OnoA{\m{N}^n}{\rho_{c}}\coleqq \sum_{p=1}^{n} 2^{-p}\rho_{c}(2^p \cdot\, -1) \;+2^{-n}\rho_{c}(2^{n}\cdot)=\rho_{c},
\end{equation*}
for all $n\in\N$. In other words, we have effectively used the three-neuron network \eqref{eq:intro-symm-1} to repeatedly replace single nodes with pairs of nodes without changing the function realized by the network, thereby constructing an infinite collection of different networks, all satisfying $\OnoA{\m{N}^n}{\rho_{c}}=\rho_{c}$.

In summary, we see that, at least for single-layer networks, non-uniqueness in the realization of a function arises from affine symmetries of the nonlinearity, where the symmetries are none other than single-layer networks mapping every input to output 0. Namely, these ``zero networks'' can be used as templates for modifying the structure of (more complex) networks without affecting the function they realize. This motivates the following definition.

\begin{definition}[Nonlinearity and affine symmetry]\label{def:sym}
A \emph{nonlinearity} is a continuous function $\rho:\R\to\R$ such that $\rho\neq \{t\mapsto at +b:t\in\R\}$, for all $a,b\in\R$.
Let $\rho:\R\to\R$ be a nonlinearity and $\m{I}$ a nonempty finite index set. An \emph{affine symmetry} of $\rho$ is a collection of real numbers of the form $\left(\zeta,\{(\alpha_s,\beta_s,\gamma_s)\}_{s\in\m{I}}\right)$ such that,
\begin{enumerate}[(i)]
\item for all $t\in\R$,
\begin{equation}\label{eq:symm-generic}
\sum_{s\in\m{I}} \alpha_s \rho(\beta_{s} t + \gamma_{s})=\zeta\,\bm{1}(t),
\end{equation}
and
\item there does not exist a proper subset $\m{I}'$ of $\m{I}$ such that $\{ \rho(\beta_{s}\cdot\, +\,  \gamma_{s}):s\in\m{I}'\}\cup\{\bm{1}\}$ is a linearly dependent set of functions from $\R$ to $\R$.
\end{enumerate}
\end{definition}
\noindent Item (ii) in Definition \ref{def:sym} is a minimality condition, ensuring that only ``atomic'' symmetries qualify under the formal definition.

Note that every nonlinearity $\rho$ satisfies $\rho+(-\rho)=0$, and hence possesses at least the ``trivial affine symmetries'' $\big(0,\{(\alpha,\beta,\gamma),(-\alpha,\beta,\gamma)\}\big)$, for $\alpha,\beta\in\R\setminus\{0\}$ and $\gamma\in\R$.
We remark that Definition \ref{def:sym} is more general than what is needed to cover our examples above, as $\zeta$ in \eqref{eq:symm-generic} is allowed to be an arbitrary real number, whereas we had $\zeta=0$ in both of our examples. One can, of course, seek to build a theory encompassing even more general symmetries, e.g. those for which the right-hand side of \eqref{eq:symm-generic} is itself an affine function $t\mapsto \zeta_1 + \zeta_2\,  t$ (which, in the context of $\rho$-modification introduced later, could then be absorbed into the next layer of the network). This is, however, outside the scope of the present paper.

\subsection{Formalizing the identifiability question}

Our aim is to generalize the aforementioned correspondence between neural network non-uniqueness and the affine symmetries of the underlying nonlinearity $\rho$ to multi-layer networks of arbitrary architecture. Moreover, we wish to do so in a canonical fashion, i.e., without regard to the ``fine properties'' of $\rho$ beyond its affine symmetries. Specifically, we will derive conditions under which the set of networks giving rise to a fixed $f$ and derived from the affine symmetries of $\rho$ through ``symmetry modification'' is exhaustive (i.e., it contains all networks giving rise to $f$). These conditions are formally characterized by our null-net theorems (Theorem \ref{thm:intro-NN-G} and Theorem \ref{thm:intro-NN-L}).
The concept of symmetry modification will be introduced in the following sections, and corresponds to using $\tanh(t)=-\tanh(-t)$ to flip the signs of weights and biases in the network (in the case when $\rho=\tanh$) or using $\rho_{c}\,(t)={ \frac{1}{2}}\rho_{c}\,(2t)+{ \frac{1}{2}}\rho_{c}\,(2t-1)$ to replace single nodes with pairs of nodes in the network (in the case $\rho=\rho_c$).

In order to streamline the extension of the discussion in the previous subsection to multi-layer networks and to facilitate the comparison of our results with previous work, it will be opportune to immediately introduce neural networks in their full generality, i.e., as ``computational graphs''. To this end, we recall the definition of a directed acyclic graph, as well as several associated concepts that will be needed later.

\begin{definition}[Directed acyclic graph, parent and ancestor set, input nodes, and node level]\label{def:DAG}\hfil
\begin{itemize}[--]
\item A directed graph is an ordered pair $G=(V,E)$ where $V$ is a nonempty finite set of nodes and $E\subset V\times V\setminus \{(v,v):v\in V\}$ is a set of directed edges. We interpret an edge $(v,\wtd{v})$ as an arrow connecting the nodes $v$ and $\wtd{v}$ and pointing at $\wtd{v}$.  
\item A directed cycle of a directed graph $G$ is a set $\{v_1,\dots,v_k\}\subset V$ such that, for every $j\in\{1,\dots,k\}$,  $(v_j,v_{j+1})\in E$, where we set $v_{k+1}\coleqq v_1$.
\item A directed graph $G$ is said to be a directed acyclic graph (DAG) if it has no directed cycles.
\end{itemize}
Let $G=(V,E)$ be a DAG. 
\begin{itemize}[--]
\item We define the parent set of a node by $\pre (\wtd{v})=\{v\in V :(v,\wtd{v})\in E\}$.
\item For a set $W\subset V$ we define $\pre^0(W)=W$ and $\pre^r(W)=\bigcup_{s\in W}\pre^{r-1}(\pre(s))$, for $r\geq 1$. The ancestor set of $W$ is now given by $\anc(W)=\bigcup_{r\geq 0}\pre^r(W)$.
\item We say that $v\in V$ is an input node if $\pre(v)=\varnothing$, and we write $\In(G)$ for the set of input nodes.
\item We define the level $\lvl(v)$ of a node $v\in V$ recursively as follows. If $\pre (v)=\varnothing$, we set $\lvl(v)=0$. If $\pre(v)=\{v_1,v_2,\dots,v_k\}$ and $\lvl(v_1),\lvl(v_2),\dots,\lvl(v_k)$ are defined, we set $\lvl(v)=\max\{\lvl(v_1),\lvl(v_2),\dots,\lvl(v_k)\}+1$.
\end{itemize}
\end{definition}

\noindent As the graph $G$ in Definition \ref{def:DAG}  is assumed to be acyclic, the level is well-defined for all nodes of $G$. We are now ready to introduce our general definition of a neural network.

\begin{definition}[GFNNs and LFNNs]\label{def:GFNN}
A general feed-forward neural network (GFNN) with $D$-dimensional output is an ordered septuple $\m{N}=(V,E,\Vin,\allowbreak \Vout,\Omega,\Theta,\Lambda)$, where
\begin{enumerate}[(i)]
\item $G=(V,E)$ is a 
DAG, called the architecture of $\m{N}$,
\item $\Vin=\In(G)$ is the set of input nodes of $\m{N}$,
\item $\Vout\subset V\setminus \Vin$ is the set of output nodes of $\m{N}$,
\item $\Omega=\{\omega_{\wtd{v}v}\in \R\setminus\{0\} : (v,\wtd{v})\in E\}$ is the set of weights of $\m{N}$,
\item $\Theta=\{\theta_{v}\in \R: v\in V\setminus \Vin\}$ is the set of biases of $\m{N}$, and
\item $\Lambda=\{\lambda^{(r)}\in \R:r\in \{1,\dots,D\}\}\cup\{\lambda_{w}^{(r)}\in\R:w\in \Vout,r\in\{1,\dots,D\}\}$ is the set of output scalars of $\m{N}$.
\end{enumerate}
The depth of a GFNN is defined as $L(\m{N})=\max\{\lvl(v):v\in V\}$. A layered feed-forward neural network (LFNN) is a GFNN satisfying $\lvl(\wtd{v})=\lvl(v)+1$ for all $(v,\wtd{v})\in E$.
\end{definition}

The role of the output scalars is to form $D$ affine combinations of the functions realized by the output nodes, which are then designated as the coordinates of the $D$-dimensional output function of the network. Note that this renders the definition of the function realized by a network more general than directly taking the functions realized by the output nodes to be the output of the network. Formally, we have the following.

\begin{definition}[Output maps]\label{def:GFNNrealiz}
Let $\m{N}=(V,E,\Vin,\Vout,\Omega,\Theta,\Lambda)$ be a GFNN with $D$-dimensional output, and let $\rho:\R\to \R$ be a nonlinearity.
The map realized by a node $u\in V$ under $\rho$ is the function $\outmap{u}{\rho}:\R^{\Vin}\to \R$ defined recursively as follows:
\begin{itemize}[--]
\item If $u\in \Vin$, set $\outmap{u}{\rho}\!(\bm{t})=t_u$, for all $\bm{t}=(t_{\wtd{u}})_{\wtd{u}\in \Vin}\in \R^{\Vin}$.
\item Otherwise, set $\outmap{u}{\rho}\!(\bm{t})=\rho \left(\sum_{v\in\pre (u)} \omega_{uv}\cdot \outmap{v}{\rho}(\bm{t}) +\theta_u\right)$, for all $\bm{t}\in \R^{\Vin}$.
\end{itemize}
The map realized by $\m{N}$ under $\rho $ is the function $\outmap{\m{N}}{\rho}:\R^{\Vin}\to \R^{D}$ given by
\begin{equation*}
\outmap{\m{N}}{\rho}=\left(\lambda^{(r)}\,\bm{1}+\sum_{w\in \Vout}\lambda^{(r)}_{w}\outmap{w}{\rho}\right)_{r\in\{1,\dots,D\}}.
\end{equation*}
 When dealing with several networks $\m{N}_j$ 
 we will write $\outmapT{u}{\rho}{\m{N}_j}$ for the map realized by $u$ in $\m{N}_j$, to avoid ambiguity.
\end{definition}
\noindent We will treat nodes $u\in V$ only as ``handles'', and never as variables or functions. This is relevant when dealing with multiple networks that have shared nodes, as in the example depicted in Figure \ref{fig:SharedNodes}. On the other hand, the  map $\outmap{u}{\rho}$ realized by $u$ is a function. We remark that Definitions \ref{def:GFNN} and \ref{def:GFNNrealiz} are largely analogous to \cite[Defs. 8,11]{Vlacic2019}, save for the output scalars $\Lambda$ that do not feature in \cite{Vlacic2019}.

\begin{figure}[h!]\centering
\includegraphics[height=50mm,angle=0]{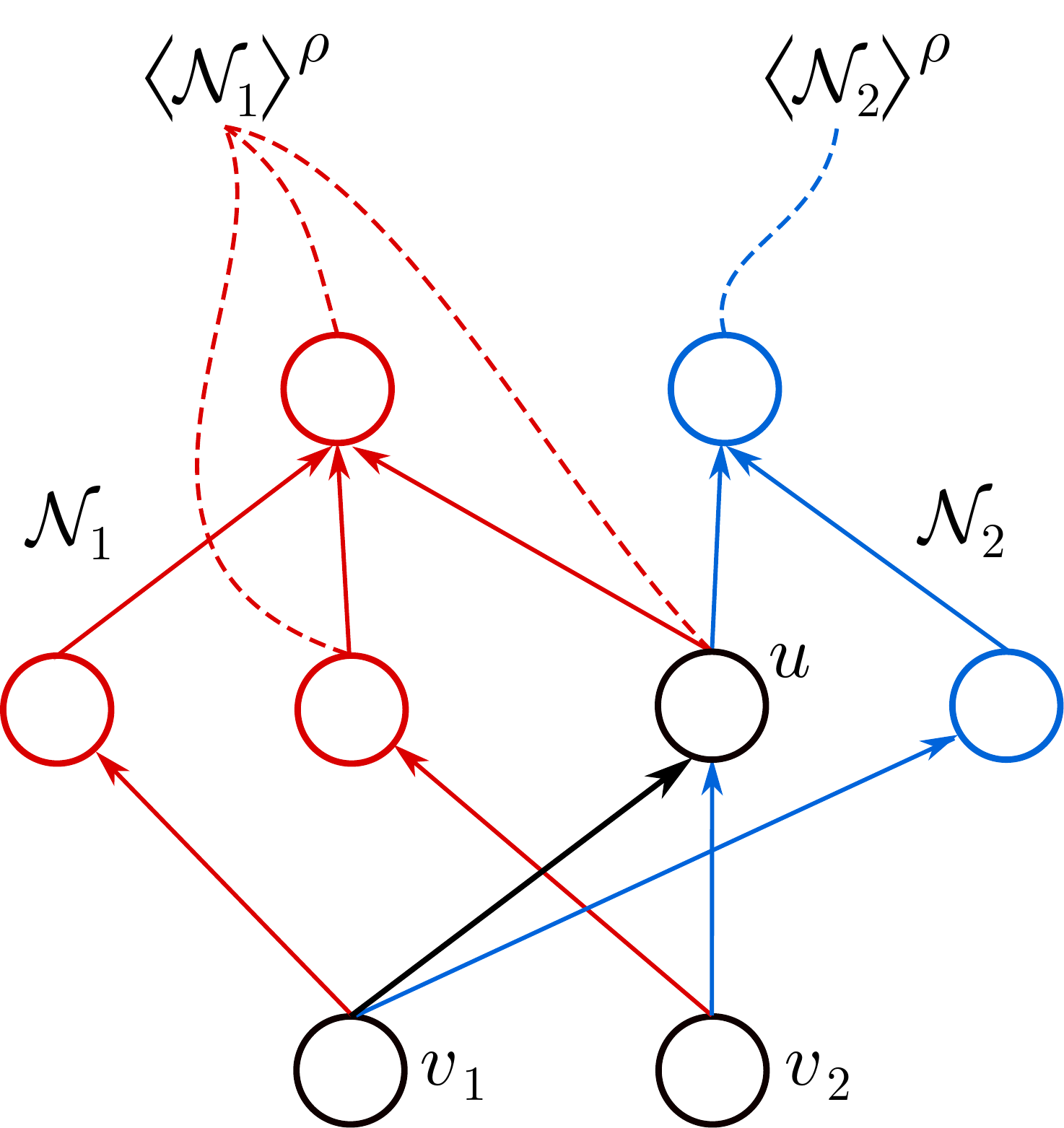}
\caption{The network $\m{N}_1$ consists of the elements in red and black, and $\m{N}_2$ consists of the elements in blue and black. The arrows represent the edges, and the dashed lines represent the output scalars. Note that $\m{N}_1$ and $\m{N}_2$ share the nodes $v_1$, $v_2$, and $u$, even though the functions $\OTnoA{u}{\rho}{\mathcal{N}_1}$ and $\OTnoA{u}{\rho}{\mathcal{N}_2}$ may be ``completely unrelated''. \label{fig:SharedNodes}}
\end{figure}

Note that LFNNs are similar to feed-forward neural networks as widely studied in the literature, namely as concatenations of affine maps between finite dimensional spaces and elementwise application of nonlinearities. Our definition of LFNNs is, however, somewhat more general, in the sense of the map of the network being allowed to depend directly on ``non-final'' nodes. An example of such an LFNN is $\m{N}_1$ in Figure \ref{fig:SharedNodes}. Further still, GFNNs are more general than LFNNs, and allow for ``skip connections'' within the network itself. For an example of a GFNN that is not layered, see Figure \ref{fig:nonlayered}.

\begin{figure}[h!]\centering
\includegraphics[height=50mm,angle=0]{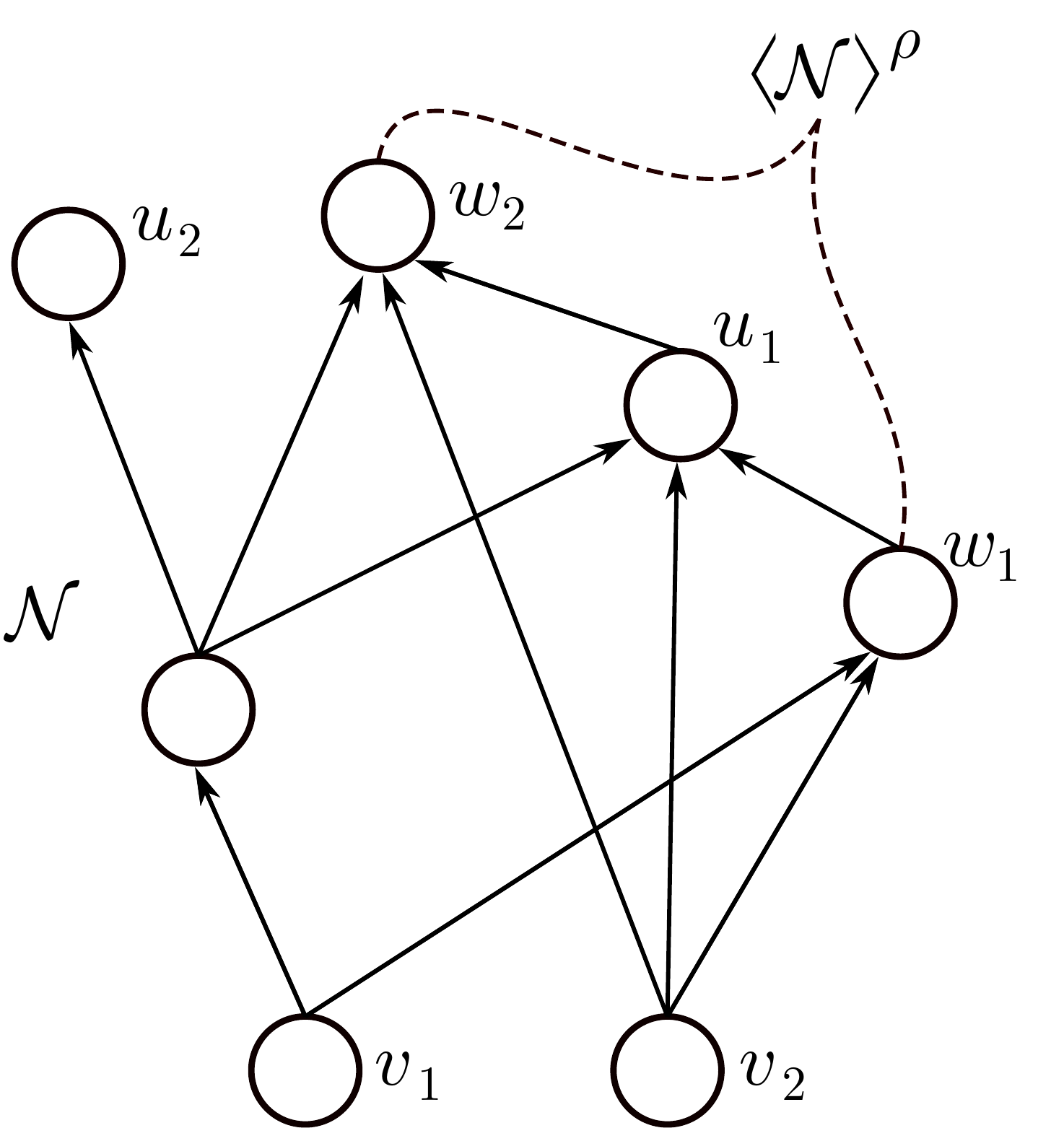}
\caption{The GFNN $\m{N}$ with one-dimensional output, input set $\{v_1,v_2\}$ and output set $\{w_1,w_2\}$. Note that $\m{N}$ is not layered as $\lvl(u_1)=2\neq 1= \lvl(v_2)+1$. This network is also degenerate due to the presence of the node $u_2$ which does not contribute to the map realized by $\m{N}$. \label{fig:nonlayered}}
\end{figure}

In order to meaningfully discuss the identifiability of GFNNs from their ouput maps, it is necessary that the networks under consideration have no spurious nodes, i.e., nodes that are ``invisible'' to the map of the network. Formally, we will require that GFNNs satisfy the following non-degeneracy property:

\begin{definition}[Non-degeneracy]\label{def:NonDeg}
We say that a GFNN $\m{N}=(V,E,\Vin,\Vout,\Omega,\Theta,\Lambda)$ with $D$-dimensional output is non-degenerate if
\begin{enumerate}[(i)]
\item $V\setminus \Vin= \anc(\Vout)\setminus \Vin$,
\item for every $w\in \Vout$, there exists an $r\in\{1,\dots,D\}$ such that $\lambda_w^{(r)}\neq 0$.
\end{enumerate}
Networks that are not non-degenerate are referred to as degenerate.
\end{definition}
\noindent Informally, a network is non-degenerate if its every non-input node ``leads up'' to at least one output node, and each output node contributes to at least one of the $D$ coordinates of the map realized by $\m{N}$. For example, the network $\m{N}$ in Figure \ref{fig:nonlayered} is degenerate. Note that non-degenerate networks are allowed to have input nodes without any outgoing edges. This is useful as we want our theory to encompass networks whose maps are constant relative to some (or all) of the inputs. An extreme but important case are the so-called \emph{trivial networks} implementing the constant zero function from $\R^{\Vin}$ to $\R^D$.

\begin{definition}[Trivial network]\label{def:TrivNet}
Let $\Vin$ be a nonempty set of nodes. We define the \emph{trivial network with $D$-dimensional output and input $\Vin$} as $\m{T}^{\, \Vin,\, D}=(\Vin,\varnothing,\Vin, \varnothing, \varnothing,\varnothing,\Lambda)$, where $\Lambda=\{\lambda^{(r)}\coleqq 0:r\in\{1,\dots,D\}\}$.
\end{definition}

\noindent Note that $\m{T}^{\, \Vin,\, D}$ is the only network with input set $\Vin$ and $D$-dimensional output of depth $0$.

We are now ready to formalize our notion of neural network identifiability. 

\begin{definition}[Identifiability]\label{def:identif}
For given $\Vin$ and $D\in\N$, let $\mathscr{N}$ be a set of non-degenerate GFNNs with $D$-dimensional output and input set $\Vin$.
Let $\rho$ be a nonlinearity, and suppose that $\sim$ is an equivalence relation on $\mathscr{N}$ such that 
\begin{equation}\label{eq:ident-compat}
\m{N}_1\sim \m{N}_2 \quad\implies\quad \OnoA{\m{N}_1}{\rho}(t)=\OnoA{\m{N}_2}{\rho}(t),\;\forall t\in \R^{\Vin}.
\end{equation}
We say that $(\mathscr{N},\rho)$ \emph{is identifiable up to }$\sim$ if, for all $\m{N}_1,\m{N}_2\in\mathscr{N}$,
\begin{equation*}
\OnoA{\m{N}_1}{\rho}(t)=\OnoA{\m{N}_2}{\rho}(t),\;\forall t\in \R^{\Vin}\quad\implies\quad  \m{N}_1\sim \m{N}_2.
\end{equation*}
\end{definition}
\noindent The equivalence relation $\sim$ thus models the ``degree of nonuniqueness'' of networks with nonlinearity $\rho$, in the sense that the relation $\sim$ partitions $\mathscr{N}$ into equivalence classes containing networks realizing the same map.  Conversely, by saying that $(\mathscr{N},\rho)$ is identifiable up to $\sim$, we mean that the equivalence class of networks realizing a given function can be inferred from the function itself.
A trivial example of such a relation is the equality relation, i.e., $\m{N}_1\sim \m{N}_2$ if and only if $\m{N}_1=\m{N}_2$. We saw in the introduction, however, that networks realizing a given function are not unique in the presence of non-trivial affine symmetries of $\rho$, and therefore in such cases $\mathscr{N}$ is not identifiable up to equality.  On the other hand, we could define an equivalence relation $\sim$ on $\mathscr{N}$ by setting $\m{N}_1\sim \m{N}_2$ if and only if $\OnoA{\m{N}_1}{\rho}=\OnoA{\m{N}_2}{\rho}$. Then $\mathscr{N}$ is, of course, identifiable up to $\sim$, but the relation $\sim$ defined in this way is not at all informative about the relationship between the structures of the networks realizing the same function. We are therefore interested in specifying the relation $\sim$ in Definition \ref{def:identif} in terms of the architecture, weights, and biases of the networks in $\mathscr{N}$ in an explicit fashion, and one would ideally like to do so for as large a class $\mathscr{N}$ of networks as possible.

To make further headway in our understanding of how $\sim$ can manifest itself for concrete nonlinearities and multi-layer networks, we again consider the case $\rho=\tanh$.
Let $\m{N}=(V,E,\Vin,\Vout,\Omega,\Theta,\allowbreak \Lambda)$ and $\m{N}'=(V',E',\Vin,\Vout',\Omega',\Theta', \Lambda')$ be non-degenerate GFNNs with $D$-dimensional output and the same input set $\Vin$. Suppose that there exist a bijection $\pi:V\to V'$ with $\pi(v)=v$, for all $v\in \Vin$, and signs $s_v\in\{-1,+1\}$, for $v\in V\setminus \Vin$, such that
\begin{itemize}[--]
\item $E'=\{(\pi(v),\pi(u)):(v,u)\in E\}$,
\item $\Vout'=\{\pi(v):v\in V\}$,
\item $\Omega'=\{\omega'_{\pi(u)\pi(v)}\coleqq s_u\omega_{uv}s_v:(v,u)\in E\}$,
\item $\Theta'=\{\theta_{\pi(v)}'\coleqq s_v\theta_v:v\in V\setminus \Vin\} $, and
\item $\Lambda'=\big\{ (\lambda^{(r)})'\coleqq \lambda^{(r)}:r\in\{1,\dots,D\}\big\}\cup \big\{(\lambda_{\pi(w)}^{(r)})'\coleqq s_{w}\lambda_{w}^{(r)}:w\in \Vout ,r\in\{1,\dots,D\} \big\}$.
\end{itemize}
We will then say that $\m{N}$ and $\m{N}'$ are \emph{isomorphic up to sign changes}, and write $\m{N}\sim_{\pm}\m{N}'$.
Owing to $\tanh(t)=-\tanh(-t)$, we have $\outmap{\m{N}}{\tanh}=\outmap{\m{N}'}{\tanh}$ whenever $\m{N}$ and $\m{N}'$ are isomorphic up to sign changes. The following question is thus natural: For which classes $\mathscr{N}$ is $(\mathscr{N},\tanh)$ identifiable up to $\sim_{\pm}$? This question was treated in the seminal paper by Fefferman \cite{Fefferman1994}, who showed that $(\mathscr{N}^{\Vin,D}_{\mrm{Feff}},\tanh)$ is identifiable up to $\sim_{\pm}$, where $\mathscr{N}^{\Vin,D}_{\mrm{Feff}}$ is the set of non-degenerate LFNNs $\m{N}=(V,E,\Vin,\Vout,\Omega,\Theta,\Lambda)$ with $D$-dimensional output and input set $\Vin$ satisfying the following structural conditions:
\begin{itemize}
\item[(F1)] $(v,u)\in E$, for all $u,v\in V$ such that $\lvl(u)=\lvl(v)+1$ (full connectivity),
\item[(F2)] $\lvl(w)=L(\m{N})$, for all $w\in \Vout$,
\item[(F3)] $\lambda^{(r)}=0$, for all $r\in\{1,\dots, D\}$, and $\Vout$ can be enumerated as $\Vout=\{w_1,\dots, w_D\}$ so that $\lambda_{w_j}^{(r)}=\delta_{jr}$, for $j,r\in\{1,\dots,D\}$, where $\delta_{jr}$ denotes the Kronecker delta,
\end{itemize}
as well as the following genericity conditions on the weights and biases:
\begin{itemize}[--]
\item[(F4)] $\theta_u\neq 0$ and $\theta_u\neq \theta_{\wtd{u}}$, for all $u,\wtd{u}\in V\setminus \Vin$ such that $u\neq \wtd{u}$ and $\lvl(u)=\lvl(\wtd{u})$, and
\item[(F5)] for all $\ell\in\{1,\dots, L(\m{N})\}$ and all $u,\wtd{u},v\in V$ so that $\lvl(v)=\ell-1$, $\lvl(u)=\lvl(\wtd{u})=\ell$, and $u\neq \wtd{u}$, we must have
\begin{equation*}
\omega_{uv}/\omega_{\wtd{u}v}\notin\left\{p/q:p,q\in\Z,\; 1\leq q\leq 100 D_\ell^2\right\},
\end{equation*}
where $D_\ell=\#\{u'\in V :\lvl(u')=\ell\}$ is the number of nodes in the $\ell$-th layer.
\end{itemize}
Fefferman's proof of the identifiability of $(\mathscr{N}^{\Vin,D}_{\mrm{Feff}},\tanh)$ up to $\sim_{\pm}$ is significant as it is the first known identification result for multi-layer networks. The proof is effected by the insight that the architecture, the weights, and the biases of a network $\m{N}\in \mathscr{N}^{\Vin,D}_{\mrm{Feff}}$ are encoded in the geometry of the singularities of the analytic continuation of $\OnoA{\m{N}}{\tanh}$.
The precise conditions (F1) -- (F5) are distilled from the proof technique so that the class of networks $\mathscr{N}^{\Vin,D}_{\mrm{Feff}}$ be as large as possible, while still guaranteeing identifiability up to $\sim_\pm$.
In the contemporary practical machine learning literature, however, a network satisfying assumptions (F1) -- (F5) would not be considered generic, as (F1) imposes a full connectivity constraint throughout the network, and (F4) implies that all biases are nonzero.
Indeed, Fefferman remarks explicitly that it would be interesting to replace (F1) -- (F5) with minimal hypotheses for layered networks. In the present paper, we address this issue and fully resolve the question of identifiability up to $\sim_\pm$ for GFNNs (and thus, in particular, for LFNNs) with the $\tanh$-nonlinearity.

The following two sections bring an informal exposition of our results leading to the resolution of neural network identifiability for the $\tanh$-nonlinearity, whereas the remainder of the paper (from Section \ref{sec:formal-NNT} onwards) is devoted to formalizing these results.

\section{A theory of identifiability based on affine symmetries}\label{sec:a-theory-of-ident}

\subsection{Canonical symmetry-induced isomorphisms and the null-net theorems}

We saw in the introduction how the symmetry $\tanh(\,\cdot\,)+ \tanh(-\, \cdot\,)=0$ of $\tanh$ leads to the equivalence relation $\sim_{\pm}$. By the same token, we will next show how the affine symmetries of a general nonlinearity $\rho$ lead to a canonical equivalence relation $\sim$ among GFNNs.
We begin by reconsidering the CReLU nonlinearity $\rho_{c}(t)=\min\{1,\max\{0,t\}\}$, both for the sake of concreteness, and because this nonlinearity, whilst of simple structure, exhibits all the phenomena we wish to address. 
We have already seen that the affine symmetry \eqref{eq:intro-symm-1} of $\rho_{c}$ leads to infinitely many distinct networks of depth 1 realizing the same map. The same symmetry can also lead to structurally different multi-layer networks realizing the same map, as illustrated by the following example. Let $\m{N}_1$, $\m{N}_2$, $\m{N}_3$, and $\m{N}_4$ be GFNNs as given schematically in Figure \ref{fig:modification}.
We then have
\begin{equation}\label{eq:examp-1-manip}
\begin{aligned}
\OnoA{\m{N}_1}{\rho_{c}}(t_1,t_2)&=\rho_{c}\big(\rho_{c}(t_1-t_2) +4 \rho_{c}(2 t_1- 2 t_2-2)\big)-{\textstyle \frac{1}{2}} \rho_{c}\big(2\rho_{c}(t_1-t_2) + 8\rho_{c}(2 t_1- 2 t_2-2)\big),\\
\OnoA{\m{N}_2}{\rho_{c}}(t_1,t_2) &=\rho_{c}\big({\textstyle \frac{1}{2}}\rho_{c}(2t_1-2t_2)+{\textstyle \frac{1}{2}}\rho_{c}(2t_1-2t_2-1) +4 \rho_{c}(2 t_1- 2 t_2-2)\big)\\
&\qquad\qquad -{\textstyle \frac{1}{2}} \rho_{c}\big(\rho_{c}(2t_1-2t_2)+\rho_{c}(2t_1-2t_2-1) + 8\rho_{c}(2 t_1- 2 t_2-2)\big),\\
\OnoA{\m{N}_3}{\rho_{c}}(t_1,t_2)&={\textstyle \frac{1}{2}}  \rho_{c}\big(\rho_{c}(2t_1-2t_2)+\rho_{c}(2t_1-2t_2-1) + 8\rho_{c}(2 t_1- 2 t_2-2)-1\big),\text{and}\\
 \OnoA{\m{N}_4}{\rho_{c}}(t_1,t_2)&={\textstyle \frac{1}{2}}  \rho_{c}\Big(\rho_{c}(2t_1-2t_2)+2 \rho_{c}\big(t_1-t_2-{\textstyle \frac{1}{2}}\big)+7 \rho_{c}(2 t_1- 2 t_2-2) -1\Big),\, \text{for $(t_1,t_2)\in\R^2$}.
\end{aligned}
\end{equation}
We now observe that $\OnoA{\m{N}_{j+1}}{\rho_{c}}=\OnoA{\m{N}_{j}}{\rho_{c}}$, for every $j\in\{1,2,3\}$, and moreover, each of these equalities can be established by performing substitutions of the affine symmetry \eqref{eq:intro-symm-1} of $\rho_{c}$ in the formal expressions \eqref{eq:examp-1-manip} of the maps $\OnoA{\m{N}_{j}}{\rho_{c}}$, $j\in\{1,2,3,4\}$.

This motivates the concept of \emph{$\rho$-modification} (to be formally introduced in Definition \ref{def:modif}) of a GFNN $\m{N}$. Suppose that an affine symmetry of $\rho$ can be used to manipulate the formal expression of $\OnoA{\m{N}}{\rho}$ as in the example above. We interpret this manipulation as a ``structural operation'' on $\mathcal{N}$ involving three distinct sets of nodes (all with a common parent set):
\begin{itemize}[--]
\item $A$, the set of nodes of $\m{N}$ to be removed,
\item $B$, the set of nodes of $\m{N}$ whose outgoing weights and output scalars are to be altered,
\item $C$, a set of newly-created nodes to be adjoined to the network.
\end{itemize}
The resulting GFNN $\m{N}'$ is called a \emph{$\rho$-modification of $\m{N}$}. We note that some of the sets $A$, $B$, and $C$ may be empty.

\begin{figure}[h!]\centering
\includegraphics[height=130mm,angle=0]{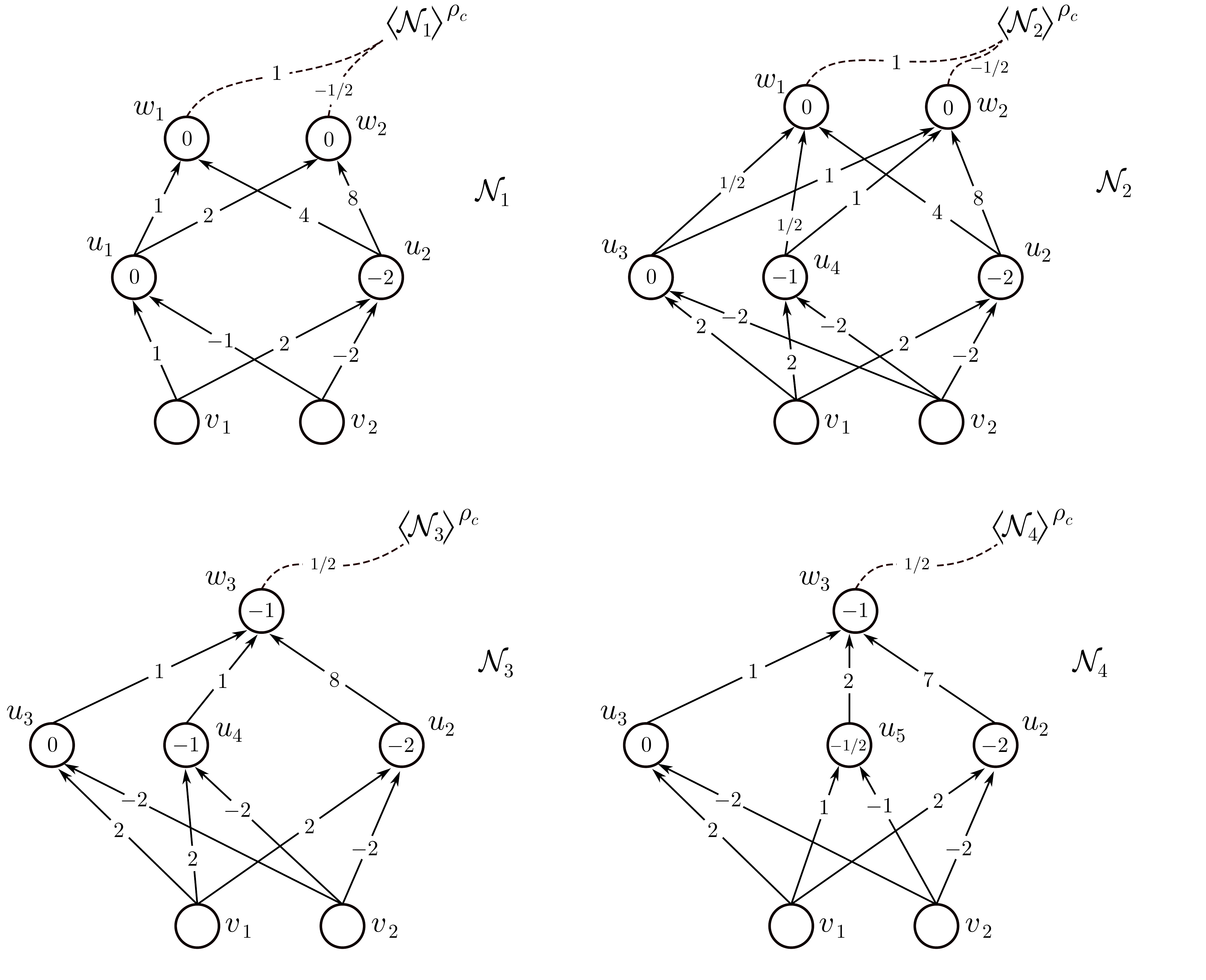}
\caption{The GFNNs $\m{N}_1$, $\m{N}_2$, $\m{N}_3$, and $\m{N}_4$. The edges are labeled by their weights, the numbers inside the nodes are their biases, and the numbers on the dashed lines are the output scalars. We have that $\m{N}_2$ is a $\rho$-modification of $\m{N}_1$ with $(A,B,C)=(\{u_1\},\varnothing,\{u_3,u_4\})$, $\m{N}_3$ is a $\rho$-modification of $\m{N}_2$ with $(A,B,C)=(\{w_1,w_2\},\varnothing,\{w_3\})$, and $\m{N}_4$ is a $\rho$-modification of $\m{N}_3$ with $(A,B,C)=(\{u_4\},\{u_2\},\{u_5\})$. \label{fig:modification}}
\end{figure}

We can thus define an equivalence relation $\rhoisom$ (to be formally introduced in Definition \ref{def:symiso}), called the \emph{$\rho$-isomorphism}, on the set $\mathscr{N}^{\Vin, D}$ of all GFNNs with $D$-dimensional output and input set $\Vin$ by letting $\m{N}\rhoisom \m{M}$ if and only if $\m{M}$ can be obtained from $\m{N}$ via a finite sequence of $\rho$-modifications. Thus, the networks $\m{N}_1$ and $\m{N}_4$ in the example above, although structurally rather different, are $\rho_c$-isomorphic.

A special case of $\rho$-modification arises if the incoming weights of several neurons $U=\{u_1,\dots,u_m\}$ of a GFNN $\m{N}$ ``line up'' with an affine symmetry of $\rho$, allowing for a $\rho$-modification with strictly fewer nodes than $\m{N}$. More precisely, suppose that a set of nodes $U$ have the same parent set $P$, and that there exist nonzero reals $\{\beta_u\}_{u\in U}$ and $\{\kappa_v\}_{v\in P}$ such that $\{\omega_{uv}\}_{v\in P}=\beta_u\{\kappa_v\}_{v\in P}$, for all $u\in U$. Assume further that $\left(\zeta,\{(\alpha_u,\beta_u,\theta_u)\}_{u\in U}\right)$ is an affine symmetry of $\rho$. Then, setting $K_P=\sum_{v\in P}\kappa_{v}\OnoA{v}{\rho}$, we have
\begin{equation}\label{eq:intro-explain-red}
\begin{aligned}
\sum_{u\in U}\alpha_u \OnoA{u}{\rho}& =\sum_{u\in U}\alpha_u\,\rho\Big(\sum_{v\in P}\omega_{uv}\OnoA{v}{\rho}\;+\theta_u \Big)=\sum_{u\in U}\alpha_u\, \rho\left(\beta_u K_P+\theta_u \right)=\zeta\,\bm{1}.
\end{aligned}
\end{equation}
Therefore, the set $\{\bm{1},\OnoA{u_1}{\rho},\dots,\OnoA{u_m}{\rho}\}$ is linearly dependent, and so $\m{N}$ admits a $\rho$-modification with $A\supset\{u_1\}$, $B= U\setminus A$, $C=\varnothing$, hence yielding a network with strictly fewer nodes than $\m{N}$. We call such a $\rho$-modification a \emph{$\rho$-reduction}. A simple example of a $\rho$-reduction is the $\tanh$-reduction of the single-layer network with the map $\tanh(\,\cdot\,)+\tanh(-\,\cdot\,)$ to the trivial network.
For a more involved example of a $\rho$-reduction, see Figure \ref{fig:reduction1}.

\begin{figure}[h!]\centering
\includegraphics[height=50mm,angle=0]{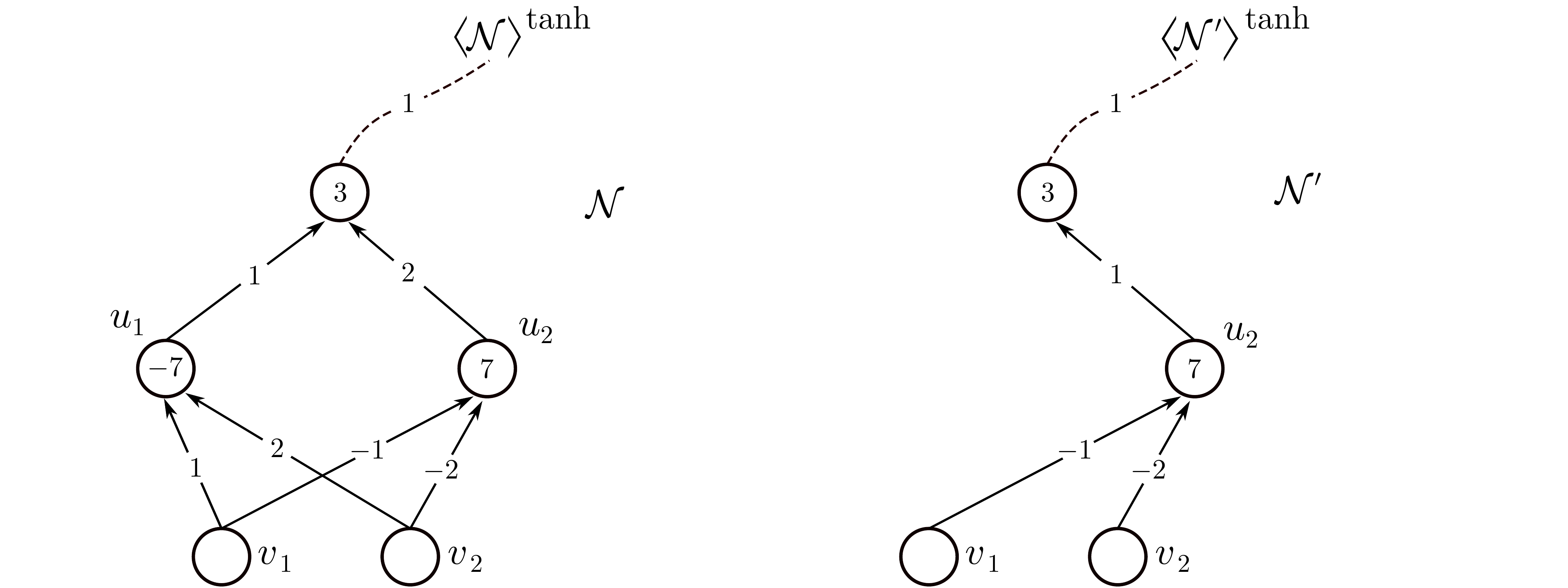}
\caption{The network $\m{N}'$ is a $\tanh$-reduction of $\m{N}$ with $(A,B,C)=(\{u_1\},\{u_2\},\varnothing)$. Concretely, we have $\OnoA{u_1}{\tanh} +\OnoA{u_2}{\tanh}=0$ (corresponding to \eqref{eq:intro-explain-red}), and thus
 $\OnoA{\m{N}}{\tanh}=\tanh\big(\OnoA{u_1}{\tanh} +2\OnoA{u_2}{\tanh} +3\big)=\tanh\big(\OnoA{u_2}{\tanh} +3\big)=\OnoA{\m{N}'}{\tanh}$, as claimed.
 \label{fig:reduction1}}
\end{figure}

A $\rho$-reduction can, in fact, yield neurons with no incoming edges. In that case, the maps of such neurons are constant, determined only by their biases, and so their values can be ``propagated through the network'' in the form of bias alteration, and the corresponding ``constant'' parts of the network can subsequently be deleted. For an example of such a $\rho$-reduction, see Figure \ref{fig:reduction2}.

\begin{figure}[h!]\centering
\includegraphics[height=50mm,angle=0]{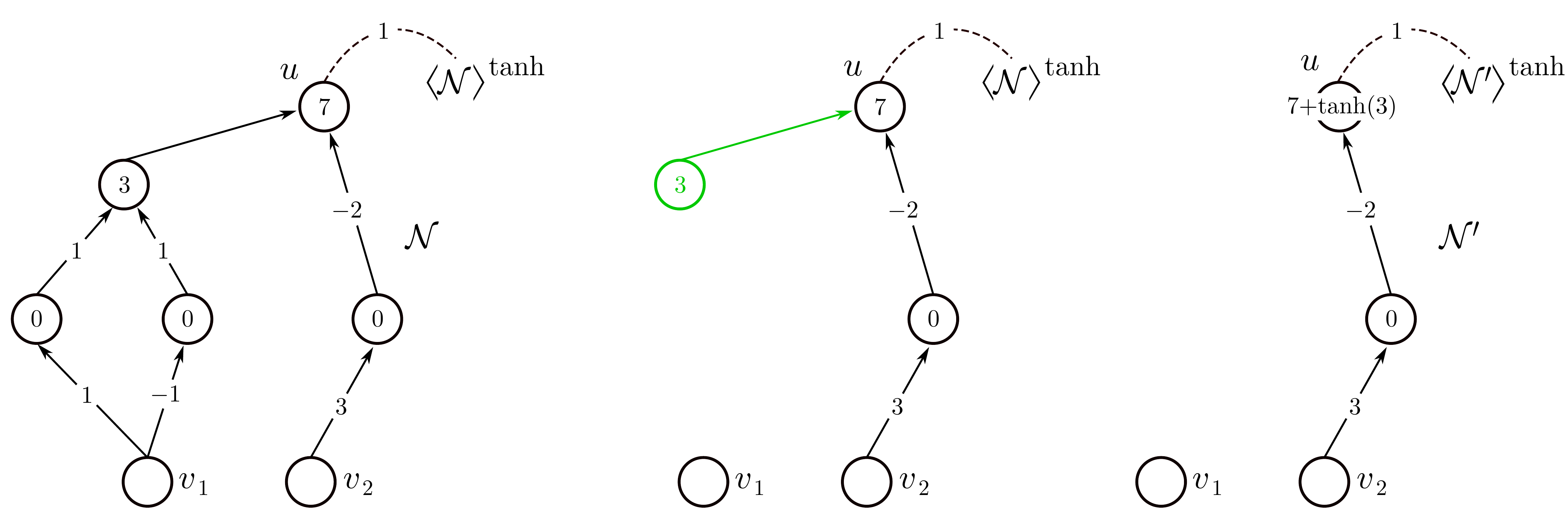}
\caption{The network $\m{N}'$ is a $\tanh$-reduction of $\m{N}$. Note that the map of the node in green is constant, taking on the value $\tanh(3)$. This value can be propagated as a bias alteration, and so the bias $7$ in the node $u$ is replaced with $7+\tanh(3)$. This example also illustrates that it is necessary to allow for input nodes without outgoing edges (the node $v_1$ in this example) in order for every network $\m{N}$ to be $\rho$-isomorphic to a regular $\m{N}'$ when reduced to ``lowest terms''. \label{fig:reduction2}}
\end{figure}

\begin{definition}
We will say that a GFNN is \emph{irreducible} if it does not admit a $\rho$-reduction, and if it is both irreducible and non-degenerate, we will say that it is \emph{regular}. 
\end{definition}

We remark that trivial networks are vacuously regular. Note that every GFNN $\m{N}$ can be reduced to ``lowest terms'' via a sequence of $\rho$-reductions, i.e., there exists a regular $\m{N}'$ such that $\m{N}\rhoisom \m{N}'$. Hence, in order to establish whether the equality $\OnoA{\m{N}}{\rho}=\OnoA{\m{M}}{\rho}$ implies $\m{N}\rhoisom \m{M}$, for $\m{N},\m{M}\in \mathscr{N}^{\Vin, \, D}$, it suffices to find regular $\m{N}'$ and $\m{M}'$ such that $\m{N}\rhoisom \m{N}'$ and $\m{M}\rhoisom \m{M}'$, and ascertain whether $\OnoA{\m{N}'}{\rho}=\OnoA{\m{M}'}{\rho}$ implies $\m{N}'\rhoisom \m{M}'$.

Therefore, in order to settle the question of identifiability up to $\rhoisom$ for all non-degenerate networks, it suffices to consider the classes $\mathscr{N}^{\Vin,\,D}_{\mrm{G}}$ and $\mathscr{N}^{\Vin,\,D}_{\mrm{L}}$ of all regular GFNNs, respectively regular LFNNs, with $D$-dimensional output and input set $\Vin$. Our first results relate the identifiability of regular networks to the following \emph{null-net condition}.

\begin{definition}[Null-net condition]\label{def:nnn-cond-intro}
Let $\rho$ be a nonlinearity and $\Vin$ a nonempty set of nodes. We say that $\rho$ satisfies the general (respectively layered) null-net condition on $\Vin$ if the only network $\m{A}\in \mathscr{N}^{\Vin,\, 1}_{\mrm{G}}$ (respectively $\m{A}\in \mathscr{N}^{\Vin,\,1}_{\mrm{L}}$) satisfying $\OnoA{\m{A}}{\rho}=0$ is the trivial network $\m{T}^{\,\Vin,\,1 }$.
\end{definition}

Definition \ref{def:nnn-cond-intro} addresses only networks with one-dimensional output, as one can easily construct identically zero networks with multi-dimensional output from identically-zero networks with one-dimensional output and vice versa.

\begin{theorem}[Null-net theorem for GFNNs]\label{thm:intro-NN-G}
Let $\rho$ be a nonlinearity. Then the class $(\mathscr{N}^{\Vin,D}_{\mrm{G}},\rho)$ of all regular GFNNs with $D$-dimensional output and input set $\Vin$ is identifiable up to $\rhoisom$ if and only if $\rho$ satisfies the general null-net condition on $\Vin$. 
\end{theorem}

\begin{theorem}[Null-net theorem for LFNNs]\label{thm:intro-NN-L}
Let $\rho$ be a nonlinearity. Then the class $(\mathscr{N}^{\Vin,D}_{\mrm{L}},\rho)$  of all regular LFNNs with $D$-dimensional output and input set $\Vin$ is identifiable up to $\rhoisom$ if and only if $\rho$ satisfies the layered null-net condition on $\Vin$. 
\end{theorem}

\noindent Theorems \ref{thm:intro-NN-G} and \ref{thm:intro-NN-L} can be seen as nonlinear analogs of the rank-nullity theorem for the ``output realization'' map $\m{E}$ taking elements of the quotient set $\mathscr{N}^{\Vin,D}_{\mrm{G}}/\rhoisom$ to functions from $\R^{\Vin}$ to $\R^{D}$ via $\m{E}([\m{N}])=\OnoA{\m{N}}{\rho}$, where $[\m{N}]=\{\m{N}'\in\mathscr{N}^{\Vin,D}_{\mrm{G}}:\m{N}'\rhoisom \m{N}\}$ denotes the equivalence class of $\m{N}$. Namely,  Theorems \ref{thm:intro-NN-G} and \ref{thm:intro-NN-L} state that the solution to the equation $\m{E}([\m{N}])=f$ is unique for every $f\in \mrm{Img}(\m{E})$ if and only if the ``null-set'' of $\m{E}$, i.e., the set of solutions $[\m{A}]$ to the equation $\m{E}([\m{A}])=0$, is trivial.

\subsection{Absence of the null-net condition for the ReLU and other piecewise linear nonlinearities}

The null-net condition does not hold for various piecewise linear nonlinearities like the ReLU, the leaky ReLU, the absolute value function, or the clipped ReLU. Concretely, let $\rho_1(t)=\max\{at,t\}$, where $a\in(0,1)$, $\rho_{2}=|\cdot|$, and either $\rho_3(t)=\max\{0,t\}$ or $\rho_3=\rho_c$. Then the $\rho_j$-regular networks $\m{A}_j$ with one-dimensional output and input set $\Vin=\{v_1\}$ as depicted in Figure \ref{fig:half-line-affine} are non-trivial, and yet satisfy $\OnoA{\m{A}_j}{\rho_j}=0$, for $j\in\{1,2,3\}$. These examples can easily be extended to input sets $\Vin$ of arbitrary cardinality.

\begin{figure}[h!]\centering
\includegraphics[height=60mm,angle=0]{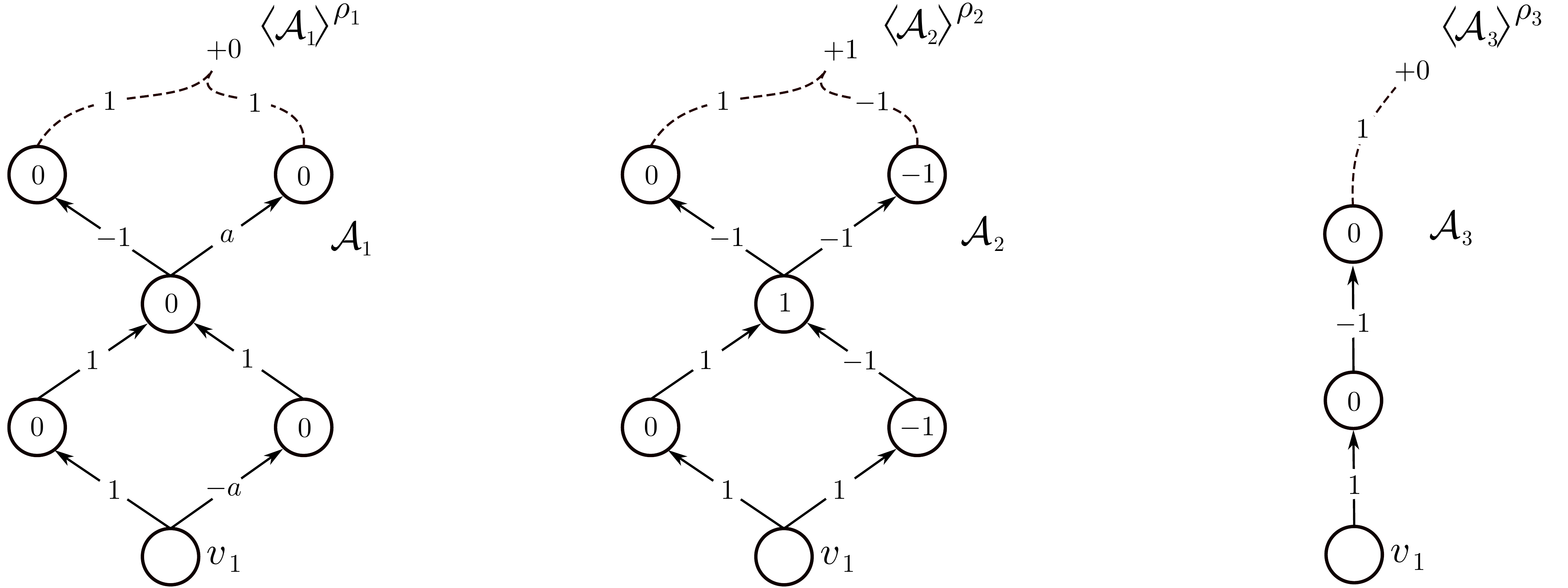}
\caption{The networks $\m{A}_j$ are $\rho_j$-regular and satisfy $\OnoA{\m{A}_j}{\rho_j}=0$, for $j\in\{1,2,3\}$. \label{fig:half-line-affine}}
\end{figure}

For these nonlinearities there exist non-$\rho$-isomorphic networks realizing the same function, indicating that the identifiability of networks with such nonlinearities is necessarily more involved. In particular, ``non-affine'' symmetries of the nonlinearity would have to be taken into account when characterizing the equivalence relation $\rho$ that is supposed to fully capture the non-uniqueness of networks realizing a given function (where, by analogy with viewing affine symmetries as single-layer zero-output networks, ``non-affine'' symmetries would correspond to multi-layer zero-output networks such as $\m{A}_1$, $\m{A}_2$, and $\m{A}_3$ in Figure \ref{fig:half-line-affine}).

\section{Identifiability for the $\tanh$ and other meromorphic nonlinearities}\label{sec:ident-for-tanh}

\subsection{Single-layer networks with the $\tanh$-nonlinearity and the simple alignment condition}

Even though both the identifiability of $(\mathscr{N}^{\Vin,D}_{\mrm{G}},\rho)$ and the null-net condition are statements quantified over all regular GFNNs (or LFNNs), and in particular over networks of arbitrarily complicated architecture, Theorems \ref{thm:intro-NN-G} and \ref{thm:intro-NN-L} allow us to shift the original question of identifiability of regular networks to a different realm where the problem will be easier to tackle by leveraging the ``fine properties'' of the nonlinearity.
Therefore, our goal will henceforth be to establish suitable sufficient conditions on nonlinearities guaranteeing that the null-net condition holds on all input sets $\Vin$.

In order to motivate our results and techniques, we demonstrate informally how the null-net condition is established for the $\tanh$ nonlinearity on a singleton input set $\{\vin\}$, and indicate in the relevant places how this argument extends to more general meromorphic nonlinearities. As the maps realized by networks with 1-dimensional output and input set $\{\vin\}$ are single-valued functions of one variable, and are defined in terms of repeated compositions of the meromorphic function $\tanh$ and affine combinations, they can be analytically continued to their natural domains in $\C$ and can therefore be studied in the context of complex analysis. This approach was pioneered by Fefferman in \cite{Fefferman1994}. 

Before continuing, we will need a concrete description of irreducibility for the $\tanh$ nonlinearity:

\begin{lemma}\label{lem:intro-tanh-irred}
A GFNN $\m{N}=(V,E,\Vin,\allowbreak \Vout,\Omega,\Theta,\Lambda)$ is irreducible with respect to the $\tanh$ nonlinearity if and only if there do not exist nodes $u_1,u_2\in V\setminus \Vin$, $u_1\neq u_2$, and an $s\in\{-1,1\}$ such that 
$P\coleqq \pre(u_1)=\pre(u_2)$,
 $\{\omega_{u_1v}\}_{v\in P}=s\{\omega_{u_2v}\}_{v\in P}$, and $\theta_{u_1}=s\,\theta_{u_2}$.
\end{lemma}

\noindent This result is a direct consequence of the following lemma providing an exhaustive characterization of the affine symmetries of $\tanh$.

\begin{lemma}[Sussman, {\cite[Lemma 1]{Sussman1992}}]\label{lem:intro-tanh-signs}
Every affine symmetry of $\tanh$ is either $\left(0,\!\{(\alpha,\beta,\gamma),\!(-\alpha,\!\beta,\!\gamma)\}\right)$ or $\left(0,\!\{(\alpha,\beta,\gamma),(\alpha,\!-\beta,\!-\gamma)\}\right)$, for some $\alpha,\beta\in\R\setminus\{0\}$ and $\gamma\in\R$.
\end{lemma}

\noindent Concretely, this says that the only affine symmetries of $\tanh$ are the ``trivial'' and the ``odd'' symmetries. As a result, $\tanh$-modification of a regular network corresponds to either leaving the network intact (if substituting the trivial symmetry), or flipping the signs of the bias and the incoming and outgoing weights of a single neuron (if substituting the odd symmetry).

Going back to establishing the null-net condition for $\tanh$ on the input set $\{\vin\}$, we first consider the single-layer case. Concretely, let $\m{N}$ be a regular GFNN with 1-dimensional output, input set $\{\vin\}$, and $L(\m{N})=1$. Enumerating the non-input nodes of $\m{N}$ as $\{u_1,\dots, u_{D_1}\}$, we have
\begin{equation*}
\OnoA{\m{N}}{\tanh}(t)=\lambda^{(1)}\;+ \sum_{j=1}^{D_1} \lambda^{(1)}_{u_j}\,\tanh(\omega_{u_j\vin}t\,+\, \theta_{u_j}), \quad \text{for }t\in \R,
\end{equation*}
where $\lambda^{(1)}_{u_j}\neq 0$, for all $j\in\{1,\dots, D_1\}$, as $\m{N}$ is non-degenerate. We aim to show that $\OnoA{\m{N}}{\tanh}$ cannot be identically zero. Then, as $\OnoA{\m{N}}{\tanh}$ can be analytically continued to a meromorphic function on $\C$, it suffices to show that its set of poles $P\subset \bigcup_{j=1}^{D_1} \omega_{u_j\vin}^{-1}\! \left(-\,\theta_{u_j} + i \pi\big(\Z +\frac{1}{2}\big) \right)$ is nonempty. To this end, let $P_j=\omega_{u_j\vin}^{-1}\! \left(-\,\theta_{u_j} + i \pi\big(\Z +\frac{1}{2}\big)\right)$ be the set of poles of $\tanh(\omega_{u_j\vin}\cdot \,+\, \theta_{u_j})$, for $j\in\{1,\dots, D_1\}$, and consider the set
\begin{equation*}
\m{J}=\{j\in\{1,\dots, D_1\}: \; \omega_{u_j\vin}^{-1}\! \theta_{u_j}= \omega_{u_1\vin}^{-1}\! \theta_{u_1},\;  \omega_{u_j\vin}/\omega_{u_1\vin}\in\Q  \}
\end{equation*}
of indices $j$ for which the functions $\tanh(\omega_{u_j\vin}\cdot \,+\, \theta_{u_j})$ and $\tanh(\omega_{u_1\vin}\cdot \,+\, \theta_{u_1})$ have common poles.
Now, assume by way of contradiction that $P\cap \bigcup_{j\in \m{J}}P_j=\varnothing$, and set 
\begin{equation*}
\beta=\max_{k\in \m{J}} |\omega_{u_{k}\vin}|\quad\text{and}\quad \m{J}_{\mrm{max}}=\{j\in\m{J}: |\omega_{u_j\vin}|= \beta\}.
\end{equation*}
Then $\#(\m{J}_{\mrm{max}})\geq 2$, as $\m{J}_{\mrm{max}}=\{j^*\}$ being a singleton would imply that $\OnoA{\m{N}}{\tanh}$ has a pole at $\beta^{-1}\big(\!-\theta_{u_{j^*}} + \frac{i \pi}{2}\big)\in P_{j^*}$, contradicting the assumption $P\cap \bigcup_{j\in \m{J}}P_j=\varnothing$. We hence deduce that there exist distinct $j_1,j_2\in\m{J}_{\mrm{max}}$. Then $ |\omega_{u_{j_1}\vin}|=|\omega_{u_{j_2}\vin}|$ and $\omega_{u_{j_1}\vin}^{-1}\! \theta_{u_{j_1}}= \omega_{u_{j_2}\vin}^{-1}\!\theta_{j_2}$, which, by Lemma \ref{lem:intro-tanh-irred},  stands in contradiction to the irreduciblity of $\m{N}$.

This establishes that $\OnoA{\m{N}}{\tanh}$ has a pole $p\in\bigcup_{j\in \m{J}} P_j$, which suffices to conclude that $\OnoA{\m{N}}{\tanh}$ cannot be identically zero. Before proceeding to the multi-layer case, it will be opportune to continue the argument above and prove a stronger statement, namely that the set $P$ of poles  of $\OnoA{\m{N}}{\tanh}$ is unbounded.
To this end, write $\OnoA{\m{N}}{\tanh}= \lambda^{(1)} +  f_1 + f_2$, where
\begin{equation*}
f_1\coleqq \sum_{j\in \m{J}} \lambda^{(1)}_{u_j}\,\tanh(\omega_{u_j\vin}\cdot \,+\, \theta_{u_j}) \quad\text{and}\quad  f_2\coleqq \sum_{j\in \{1,\dots,D_1\}\setminus \m{J} } \lambda^{(1)}_{u_j}\,\tanh(\omega_{u_j\vin}\cdot \,+\, \theta_{u_j}).
\end{equation*}
Note that the sets of poles of $f_1$ and $f_2$ are disjoint (as $P_j\cap P_k=\varnothing$, for all $j\in \m{J}$ and $k\in  \{1,\dots,D_1\}\setminus \m{J}$), and hence $p$ must be a pole of $f_1$.
What is more, as $\omega_{u_j\vin}/\omega_{u_1\vin}\in\Q$, for all $j\in \m{J}$, there exists a $T\in\R$ such that $\omega_{u_j\vin}T/\pi \in \Z$, for all $j\in\m{J}$, and so $f_1$ is $iT$-periodic, further implying that $p+iTk$ is a pole of $f_1$, for every $k\in\Z$. Therefore,
$P\supset \{ p+iTk: k\in\Z\}$, and so $P$ is unbounded. This argument leads to the following alignment condition for the $\tanh$ nonlinearity.

\begin{definition}[Simple alignment condition]\label{def:SAC}
Let $\sigma$ be a meromorphic nonlinearity on $\C$. We say that $\sigma$ satisfies the \emph{simple alignment condition} (SAC) if the following implication holds for all finite sets of triples $\{(\alpha_s,\beta_s,\gamma_s)\}_{s\in\m{I}}\subset \R \times\R\times \R$:
\begin{equation*}
\text{the set of poles of  $f\coleqq \sum_{s\in\m{I}}\alpha_s\,\sigma(\beta_s\cdot\,+\,\gamma_s)$ is bounded} \quad\implies\quad f\text{ is constant on }\C.
\end{equation*}
\end{definition}
 
\subsection{Multi-layer networks with the $\tanh$-nonlinearity and the composite alignment condition}
 
We are now ready to proceed to the multi-layer case of our argument establishing the null-net condition for $\tanh$ on $\{\vin\}$. More specifically, we will show how the ``nonemptiness of the pole set'' property can be extended to multi-layer networks by induction on depth. This will then immediately imply that the maps of these networks cannot be identically zero, establishing the null-net condition for $\tanh$ on the singleton input set $\{\vin\}$.
Our discussion will reveal a sufficient condition (the \emph{composite alignment condition}) for this inductive argument to generalize to arbitrary meromorphic nonlinearities with simple poles only, which, together with the SAC, will allow us to establish the null-net condition for meromorphic nonlinearities more general than $\tanh$.

It will be of interest to consider the maximal domain in $\C$ to which the map $\OnoA{\m{N}}{\tanh}$ of a non-trivial regular GFNN  $\m{N}$ can be analytically continued.
Even though for a general holomorphic function there may not exist a unique maximal set to which it can be analytically continued (consider, for instance, the function $z\mapsto \sqrt{1+z^2}$), this is the case for holomorphic functions defined on a domain with countable complement in $\C$ (a property the map $\OnoA{\m{N}}{\tanh}$ will be shown to possess). We thus have the following definition.
\begin{definition}[Natural domain]\label{def:nat-dom}
Suppose $f:\dom\to \C$ is a holomorphic function on a domain with countable complement in $\C$. The natural domain of $f$ is the unique maximal set $\dom_f\supset \dom$ with respect to set inclusion to which $f$ can be analytically continued.
\end{definition}
\noindent The existence of a unique maximal set $\dom_f$ in Definition \ref{def:nat-dom} is formally justified by \cite[Lemma III.A.1]{Fefferman1994}.

Now, let $\m{N}=(V,E,\{\vin\},\Vout,\Omega,\Theta,\Lambda)$ be a non-trivial regular GFNN with 1-dimensional output of depth $L(\m{N})\geq 2$, and, for every non-trivial regular GFNN $\m{N}'$ with 1-dimensional output, input set $\{\vin\}$, and depth $L(\m{N}')< L(\m{N})$, assume that
\begin{enumerate}[--]
\item $\OnoA{\m{N}'}{\tanh}$ can be analytically continued to a domain with countable complement in $\C$ and
\item the set of simple poles of $\OnoA{\m{N}'}{\tanh}$ is nonempty.
\end{enumerate}
We aim to show that the set of simple poles of $\OnoA{\m{N}}{\tanh}$ is nonempty under these assumptions. To this end, first note that we can write
\begin{equation}\label{eq:intro-simplified-1}
\OnoA{\m{N}}{\tanh}(z)=f(z) + \sum_{w\in \Vout^{>1}} \lambda^{(1)}_w\,  \tanh\big(\OnoA{\m{N}_w}{\tanh}(z)\big),
\end{equation}
where $\m{N}_w$, for $w\in \Vout^{>1}\coleqq\{w\in \Vout:\lvl(w)>1\}$, are non-trivial regular GFNNs with input set $\{\vin\}$ and depth $L(\m{N}_w)< L(\m{N})$, and $f:\dom_f\to \C$ is a meromorphic function given by
\begin{equation*}
f(z)=\lambda^{(1)}\;+ \sum_{\substack{w\in \Vout \\ \lvl(w)=1 }} \lambda^{(1)}_{w}\,\tanh(\omega_{w\vin}z\,+\, \theta_{w}).
\end{equation*}
One can show that \eqref{eq:intro-simplified-1} holds for $z$ in an open set with countable complement in $\C$ (see Lemma \ref{lem:nat-dom}), and so the natural domain $\dom_{\OnoA{\m{N}}{\tanh}}$ of $\OnoA{\m{N}}{\tanh}$ is well-defined.
Write  $P_w$ for the set of poles of $\OnoA{\m{N}_{w}}{\tanh}$, for $w\in\Vout^{>1}$. Now, fix a $w^*\in \Vout^{>1}$ and a $p\in P_{w^*}$, and set $\Vout^*=\{w\in \Vout^{>1}: p\in P_w\}$. We make the following assumption:
\begin{equation}\label{eq:non-ess-cancel}
\text{$\OnoA{\m{N}_w}{\tanh}$ is analytic at $p$, for all $w\in \Vout^{>1}\setminus \Vout^*$.}
\end{equation}
Next, note that, for $w\in \Vout^*$, as $p$ is a simple pole of $\OnoA{\m{N}_w}{\tanh}$, we can write
\begin{equation}\label{eq:intro-simplified-2}
\OnoA{\m{N}_w}{\tanh}(z)=\frac{\beta_{w}}{z-p} + \gamma_{w} + \epsilon_{w}(z-p),
\end{equation}
for $z$ in an open neighborhood of $p$, where $\beta_{w}\in \C\setminus\{0\}$, $\gamma_{w}\in \C$, and $\epsilon_{w}:\dom_{\epsilon_{w}}\to \C$ is a function holomorphic on a domain $\dom_{\epsilon_{w}}$ with countable complement in $\C$ and such that $\epsilon_{w}(0)=0$. Using \eqref{eq:intro-simplified-2} in \eqref{eq:intro-simplified-1} and performing the variable substitution $z'=\frac{1}{z-p}$ then yields
\begin{equation}\label{eq:intro-simplified-3}
\OnoA{\m{N}}{\tanh}\left(p+1/z'\right)=(f+g)\left(p+1/z' \right)+ \sum_{w\in \Vout^*} \lambda^{(1)}_w\,  \tanh\left(\beta_{w} \,z' +\gamma_w +\epsilon_{w}(1/z') \right),
\end{equation}
for all $z'\in \C$ of sufficiently large modulus, where
\begin{equation*}
g(z)=\sum_{w\in \Vout^{>1}\setminus \Vout^*} \lambda^{(1)}_w\,  \tanh\big(\OnoA{\m{N}_w}{\tanh}(z)\big)
\end{equation*}
is analytic on a punctured neighborhood of $p$ owing to the assumption \eqref{eq:non-ess-cancel}. Then, according to \eqref{eq:intro-simplified-3}, $p$ will be a cluster point of simple poles of $\OnoA{\m{N}}{\tanh}$, unless the set of poles of
\begin{equation}\label{eq:intro-simplified-4}
z'\; \longmapsto \;  \sum_{w\in \Vout^*} \lambda^{(1)}_w\,  \tanh\left(\beta_{w} \, z' +\gamma_w +\epsilon_{w}(1/z') \right)
\end{equation}
is bounded. Therefore, if we can guarantee that
\begin{enumerate}[(i)]
\item there exists a $p\in P_{w^*}$ satisfying \eqref{eq:non-ess-cancel}, and 
\item the set of poles of the function \eqref{eq:intro-simplified-4} is unbounded,
\end{enumerate}
then we will be able to conclude that the set of simple poles of $\OnoA{\m{N}}{\tanh}$ is nonempty, as desired.
Item (i) can be established by more careful bookkeeping of the clusters of poles already formed in $\OnoA{\m{N}_w}{\tanh}$, for $w\in \Vout^{>1}$, whereas (ii) will be a consequence of the {composite alignment condition} introduced next. 

\begin{definition}[Asymptotic bias compensator]\label{def:ABC}
An \emph{asymptotic bias compensator} (ABC) is a holomorphic function $\epsilon:\dom_{\epsilon}\to\C$ such that $\C\setminus\dom_{\epsilon}$ is closed and countable, $0\in \dom_{\epsilon}$, and $\epsilon(0)=0$.
\end{definition}

\begin{definition}[Composite alignment condition]\label{def:CAC}
Let $\sigma$ be a meromorphic nonlinearity on $\C$ with infinitely many simple poles and no poles of higher order. We say that $\sigma$ satisfies the \emph{composite alignment condition} (CAC) if the following implication holds for all nonempty finite sets of triples $\{(\alpha_s,\beta_s,\gamma_s)\}_{s\in\m{I}}\subset \C\times(\C\setminus \{0\})\times \C $ and all sets $\{\epsilon_s\}_{s\in\m{I}}$ of ABCs:
\begin{equation}\label{eq:cac-impl} 
\begin{aligned}
&\text{the set of poles of  $\sum_{s\in\m{I}}\alpha_s\,\sigma\left(\beta_s \cdot\,+\, \gamma_s+\epsilon_s(1/\cdot)\right)$ is bounded}\\
\implies \quad  &\exists\, \text{ nonempty }\m{I}'\subset\m{I} \quad \text{ s.t. }\quad  \beta_{s_1}^{-1}\epsilon_{s_1}=\beta_{s_2}^{-1}\,\epsilon_{s_2},\, \forall s_1,s_2\in\m{I}', \text{ and}\\
&  \text{the set of poles of $\sum_{s\in\m{I}'}\alpha_s\,\sigma\left(\beta_s \cdot\,+\,\gamma_s\right)$ is bounded.}
\end{aligned}
\end{equation}
\end{definition}

To see why item (ii) above follows from the CAC, assume by way of contradiction that the set of poles of the function \eqref{eq:intro-simplified-4} is bounded. Then, by the CAC, there exists a nonempty $U\subset \Vout^*$ such that 
$\beta_{w_1}^{-1}\epsilon_{s_1}=\beta_{w_2}^{-1}\epsilon_{w_2}$, for all $w\in U$, and the set of poles of 
\begin{equation}\label{eq:intro-simplified-5}
f_U\coleqq \sum_{w\in U}  \lambda^{(1)}_w\, \sigma\left(\beta_w \cdot \,+\,\gamma_w \right)
\end{equation}
is bounded. This together with \eqref{eq:intro-simplified-2} implies that
\begin{equation*}
\beta_{w_1}^{-1 }\OnoA{\m{N}_{w_1}}{\tanh}-\beta_{w_2}^{-1} \OnoA{\m{N}_{w_2}}{\tanh}=\left(\beta_{w_1}^{-1}\gamma_{w_1}-\beta_{w_2}^{-1}\gamma_{w_1}\right)\bm{1}
\end{equation*}
is constant, for all $w_1,w_2\in U$. Now, unless 
\begin{equation}\label{eq:intro-simplified-6}
\begin{aligned}
Y\coleqq \pre_{\m{N}}(w_1) &= \pre_{\m{N}}(w_2)\quad\text{and }\\
\beta_{w_1}^{-1}\{\omega_{w_1u}\}_{u\in Y}&= \beta_{w_2}^{-1}\{\omega_{w_2u}\}_{u\in Y},
\end{aligned}
\end{equation}
 for all $w_1,w_2\in U$, it would be possible to find distinct $w_1',w_2'\in U$ and construct a non-trivial regular GFNN $\m{N}'$ with 1-dimensional output, input set $\{\vin\}$, and depth $L(\m{N}')< L(\m{N})$ such that $\OnoA{\m{N}'}{\tanh}= \beta_{w_1'}^{-1 }\OnoA{\m{N}_{w_1'}}{\tanh}-\beta_{w_2'}^{-1} \OnoA{\m{N}_{w_2'}}{\tanh}$ is constant, which would contradict the assumption that the set of simple poles of $\OnoA{\m{N}'}{\tanh}$ is nonempty. Therefore, \eqref{eq:intro-simplified-6} must hold, which  will further imply the existence of a $\vartheta\in \R$ and a $c\in\C$ such that $\beta_w e^{-i\vartheta}\in\R$, for all $w\in U$, and
 \begin{equation}\label{eq:intro-simplified-7}
\sum_{w\in U} \lambda^{(1)}_w\,  \sigma\big(\beta_{w} e^{-i\vartheta} \, \cdot  +\,\theta_w \big)=f_U(e^{-i\vartheta}\,\cdot\, + c).
\end{equation}
As the set of poles of $f_U$ is bounded and $\beta_w e^{-i\vartheta}\in\R$, for all $w\in U$, the SAC for $\sigma$ now implies that the function \eqref{eq:intro-simplified-7} must be constant. However, this and \eqref{eq:intro-simplified-6} together contradict the irreducibility of $\m{N}$, establishing that the set of poles of  \eqref{eq:intro-simplified-4} must be unbounded.

Finally, it remains to justify why $\tanh$ satisfies the CAC. To this end, we first need to define and analyze several concepts related to densities of  subsets of $\C$. These will be used to characterize the geometric relationship between the poles of the summand functions in \eqref{eq:cac-impl}. 

\begin{definition}\label{def:dens}[Line, arithmetic sequence, and density]
\begin{enumerate}[(i)]
\item A \emph{line} in $\C$ is a set of the form $\ell=\{x+ty:t\in\R\}$, where $x\in\C$ and $y\in\C\setminus\{0\}$.
\item An \emph{arithmetic sequence} in $\C$ is a set of the form $\Pi=\{x+ky:k\in\Z\}$, where $x\in\C$ and $y\in\C\setminus\{0\}$.
\item For an arbitrary set $F\subset \C$, a discrete set $P\subset\C$, and $\varepsilon>0$, we set
\begin{equation*}
\Delta_\varepsilon(F,P)=\limsup_{N\to\infty}\frac{1}{2N}\, \#\{p\in P: |p|\leq N,\;\exists\, q\in F \text{ s.t. } |p-q|\leq \varepsilon\},
\end{equation*}
and we define the \emph{asymptotic density of $P$ along $F$} by
\begin{equation*}
\Delta(F,P)=\lim_{\varepsilon\to 0}\Delta_\varepsilon(F,P)=\inf_{\varepsilon>0}\Delta_\varepsilon(F,P).
\end{equation*}
\end{enumerate}
\end{definition}
\noindent  Note that the limit as $\varepsilon\to 0$ in the previous definition always exists, as $\Delta_\varepsilon(F,P)$ is an increasing function of $\varepsilon$.  Furthermore, as the limit superior is subadditive, so is the asymptotic density, specifically,
\begin{equation*}
\Delta(F,P_1\cup P_2)\leq\Delta(F,P_1)+\Delta(F,P_2), 
\end{equation*}
for $F\subset \C$ and discrete $P_1,P_2\subset\C$.

Now, assume that the antecedent of \eqref{eq:cac-impl} is satisfied with $\sigma=\tanh$, and let  $\wtd{P}_s$ denote the set of poles of $z\mapsto \tanh\left(\beta_s z+\, \gamma_s+\epsilon_s(1/z)\right)$, for $s\in\m{I}$. In order to specify the subset $\m{I}'\subset\m{I}$ for which we will prove the consequent of \eqref{eq:cac-impl}, we  first observe the following:
\begin{itemize}[--]
\item There exists an $R>0$ such that every element of $\bigcup_{s\in \m{I}} \wtd{P}_{s}$ is contained in both $\wtd{P}_{s_1}$  and $\wtd{P}_{s_2}$, for some distinct $s_1,s_2\in \m{I}$,
\item for every $s\in\m{I}$, the set $\wtd{P}_s$ is asymptotic to the arithmetic sequence $\Pi_s\coleqq \beta_s^{-1}\! \left(-\,\gamma_s + i \pi\big(\Z +\frac{1}{2}\big) \right) $,  in the sense that, for every $\varepsilon>0$, there exists an $A>0$ such that every $p'\in \wtd{P}_s$ with $|p'|>A$ is within $\varepsilon$ of $\Pi_s$ and every $p\in \Pi_s$ with $|p|>A$ is within $\varepsilon$ of $\wtd{P}_s$, and
\item for every $s\in\m{I}$, the density of $\wtd{P}_s$ along the line $\ell=\{\beta_s^{-1}(-\gamma_s +i t):t\in\R\}$ is strictly positive, i.e., we have $\Delta(\ell,\wtd{P}_s)>0$.
\end{itemize}
This motivates defining an undirected graph $\m G=(\m I, \m E)$ on $\m{I}$, with $\m{E}$ given by
\begin{equation*}
\m{E}\coleqq\{(s_1,s_2)\in \m I \times \m I: s_1\neq s_2 \text{ and }\exists \text{ line }\ell \text{ in }\C\text{ s.t. } \Delta(\ell,\wtd{P}_{s_1}\cap \wtd{P}_{s_2})>0\}.
\end{equation*}
Informally, the condition $\Delta(\ell,\wtd{P}_{s_1}\cap \wtd{P}_{s_2})>0$, for $(s_1,s_2)\in \m{E}$, imposes sufficient ``geometrical rigidity'' on the points of $\wtd{P}_{s_1}$  and $\wtd{P}_{s_2}$ in order for $\beta_{s_1}^{-1}\epsilon_{s_1}=\beta_{s_2}^{-1}\epsilon_{s_2}$ to hold, whereas, for $(s_1,s_2)\notin \m{E}$, we have $\Delta(\ell,\wtd{P}_{s_1}\cap \wtd{P}_{s_2})=0$ for every line $\ell$ in $\C$, and so  $\wtd{P}_{s_1}$  and $\wtd{P}_{s_2}$ do not ``get in the way'' of one another.
This reasoning will allow us to show that the consequent of  \eqref{eq:cac-impl} holds for every connected component of $\m{G}$. To this end, we fix an arbitrary connected component $\m{I}'$ of $\m{G}$ and $s_1,s_2\in\m{I}'$ such that $(s_1,s_2)\in\m{E}$. Then, as 
$\Delta(\ell,\wtd{P}_{s_1}\cap \wtd{P}_{s_2})>0$, there exists a sequence of poles $\{p_n\}_{n=1}^\infty\subset \wtd{P}_{s_1}\cap\wtd{P}_{s_2}$ diverging to infinity, i.e.,
\begin{equation*}
\beta_{s_1} p_n+ \gamma_{s_1}+\epsilon_{s_1}(1/p_n),\quad  \beta_{s_2} p_n+ \gamma_{s_2}+\epsilon_{s_2}(1/p_n)\quad \in i \pi\left(\Z +{\textstyle \frac{1}{2}}\right),
\end{equation*}
for all $n\in\N$, further implying that 
$
\left(\beta_{s_1}^{-1}\epsilon_{s_1}-\beta_{s_2}^{-1}\epsilon_{s_2}\right)(1/p_n)\in \Pi_{s_1}-\Pi_{s_2},
$
for all $n\in\N$, and $\left(\beta_{s_1}^{-1}\epsilon_{s_1}-\beta_{s_2}^{-1}\epsilon_{s_2}\right)(1/p_n)\to 0$ as $n\to\infty$. On the other hand, one can show
\begin{equation*}
\Delta(\Pi_{s_1},\Pi_{s_2})\geq \Delta(\ell,\wtd{P}_{s_1}\cap \wtd{P}_{s_2})>0,
\end{equation*}
which, by a special case of Weyl's equidistribution theorem \cite[Cor. 2.A.12]{Fefferman1994}, implies that $\beta_{s_1}/\beta_{s_2}\in\Q$, further implying that $ \Pi_{s_1}-\Pi_{s_2}$ is uniformly discrete. Therefore, we must have $(\beta_{s_1}^{-1}\epsilon_{s_1}-\beta_{s_2}^{-1}\epsilon_{s_2})(1/p_n)=0$, for all sufficiently large $n$, and thus, as $\beta_{s_1}^{-1}\epsilon_{s_1}-\beta_{s_2}^{-1}\epsilon_{s_2}$ is analytic on a neighborhood of 0 and $1/p_n\to 0$ as $n\to\infty$, it follows by the identity theorem that $\beta_{s_1}^{-1}\epsilon_{s_1}-\beta_{s_2}^{-1}\epsilon_{s_2}=0$. Hence, as $s_1$  and $s_2$ were arbitrary and $\m{I}'$ is connected, we must have $ \beta_{s_1}^{-1}\epsilon_{s_1}=\beta_{s_2}^{-1}\epsilon_{s_2}$, for all $s_1,s_2\in\m{I}'$.
It remains to show that the set ${P}_{\m{I}'}$ of poles of 
$
f_{\m{I}'}\coleqq \sum_{s\in\m{I}'}\alpha_s \tanh\left(\beta_s \cdot\,+\,\gamma_s\right)
$
 is bounded. We will, in fact, prove a stronger statement, namely that ${P}_{\m{I}'}$ is empty. To this end, suppose by way of contradiction that the set ${P}_{\m{I}'}$ is nonempty. Then, by an argument analogous to the discussion of the single-layer case (specifically, using $\beta_{s_1}/\beta_{s_2}\in\Q$, for all $s_1,s_2\in \m{I}'$), there must exist a line $\ell$ in $\C$ such that $\Delta(\ell,P_{\m{I}'})>0$. Next, letting $\xi=\beta_{s'}^{-1}\epsilon_{s'}$ for an arbitrary $s'\in\m{I}'$, we have $\epsilon_s=\beta_s \xi$, for all $s\in\m{I}'$, and thus the asymptotic density of the poles of
\begin{equation*}
\sum_{s\in\m{I}'}\alpha_s\tanh\left(\beta_s \cdot+\gamma_s+\epsilon_s(1/\cdot)\right)= \sum_{s\in\m{I}'}\alpha_s\tanh\big(\beta_s (\cdot+\xi(1/\cdot))+\gamma_s\big)=f_{\m{I}'}\big(\cdot\,+\,\xi(1/\cdot)\big)
\end{equation*}
along $\ell$ is equal to $\Delta(\ell,P_{\m{I}'})$, since $\xi(1/z)\to 0$ as $|z|\to\infty$. Now, using the subadditivity property of the asymptotic density, we find that the set $\wtd{P}_{\m{I}}$ of poles of $\sum_{s\in\m{I}}\alpha_s\tanh\left(\beta_s \cdot\,+\, \gamma_s+\epsilon_s(1/\cdot)\right)$ must satisfy
\begin{equation*}
\Delta(\ell,\wtd{P}_{\m{I}})\geq \Delta(\ell,{P}_{\m{I}'}) -\sum_{s_1\in \m{I}'}\sum_{s_2\in\m{I}\setminus\m{I}'}  \underbrace{\Delta(\ell,\wtd{P}_{s_1}\cap \wtd{P}_{s_2})}_{=0}>0,
\end{equation*}
which contradicts our assumption that $\wtd{P}_{\m{I}}$ be bounded. This shows that ${P}_{\m{I}'}=\varnothing$, thereby proving the CAC for $\tanh$ and concluding our informal argument establishing the null-net property for $\tanh$ on $\{\vin\}$.

\subsection{General meromorphic nonlinearities and arbitrary input sets}

We will later formalize the discussion in the previous two subsections, proving the following result for meromorphic nonlinearities more general than $\tanh$.

\begin{prop}\label{prop:intro-vin-NNC}
Let $\sigma$ be a meromorphic nonlinearity on $\C$ with infinitely many simple poles and no poles of higher order. Suppose that $\sigma(\R)\subset \R$, and that $\sigma$ satisfies both the SAC and the CAC. Then, for every non-trivial regular GFNN $\m{N}$ with 1-dimensional output and a singleton input set $\{\vin\}$, the map $\OnoA{\m{N}}{\sigma}$ can be analytically continued to a domain with countable complement in $\C$, and its set of poles is nonempty. In particular, $\sigma$ satisfies the general (and therefore also the layered) null-net condition on $\{\vin\}$. 
\end{prop}

The final step is to establish the null-net property on input sets $\Vin=\{v^0_1,\dots,v^0_{D_0}\}$ of arbitrary size $D_0$. As the argument for $\tanh$ is identical to that for more general meromorphic nonlinearities satisfying the alignment conditions, we proceed by assuming that $\sigma$ is a meromorphic nonlinearity on $\C$ satisfying the SAC and the CAC, but is otherwise arbitrary (we will shortly discuss such nonlinearities that are not $\tanh$).
We argue by contradiction, i.e., we assume the existence of a non-trivial regular GFNN $\m{M}$ with  input set $\Vin$ and a one-dimensional output identically equal to zero. Next, we use the \emph{input anchoring} procedure, which is a method for constructing a non-trivial network $\m{M}_a$ derived from $\m{M}$ in a manner that preserves the zero-output property while reducing the cardinality of the input set. This is achieved by selecting an input node of $\m{M}$, say $v_{D_0}^0$, and a real number $a\in\R$ that is then assigned to that node as a fixed value and propagated through the network in the form of bias alteration. The parts of $\m{M}$ whose contributions are rendered constant in the process are then deleted. The so-constructed network $\m{M}_{a}$ has a smaller input set $\Vin\setminus\{v^0_{D_0}\}$ and by construction satisfies
\begin{equation*}
\OnoA{\m{M}_a}{\sigma}(t_{v^0_1} \,, \dots ,\, t_{v^0_{D_0-1}})=\OnoA{\m{M}}{\sigma}(t_{v^0_1}\,,\dots,\, t_{v^0_{D_0-1}}, a)=0,
\end{equation*}
for all $(t_{v^0_1},\dots,t_{v^0_{D_0-1}})\in \R^{\Vin\setminus\{v^0_{D_0}\}}$. We will later show that a value of $a$ can be selected so that the network $\m{M}_a$ is regular. The procedure can now be repeated, successively eliminating the input nodes until only one remains. We are thus left with a non-trivial regular GFNN with a singleton input set and one-dimensional output identically equal to zero. This constitutes a contradiction to the null-net property for $\sigma$ on singleton input sets, thereby establishing the null-net property on arbitrary input sets. The input anchoring procedure is illustrated in Figure \ref{fig:anchoring-g-intro}.
Formalizing this argument will allow us to prove the following theorem.

\begin{theorem}\label{thm:intro-Zab-NNC}
Let $\sigma$ be a meromorphic nonlinearity on $\C$ with infinitely many simple poles and no poles of higher order. Suppose that $\sigma(\R)\subset \R$, and that $\sigma$ satisfies both the SAC and the CAC. Then $\sigma$ satisfies the general (and therefore also the layered) null-net condition on $\Vin$, for every finite set $\Vin$.
\end{theorem}

\begin{figure}[h!]\centering
\includegraphics[height=65mm,angle=0]{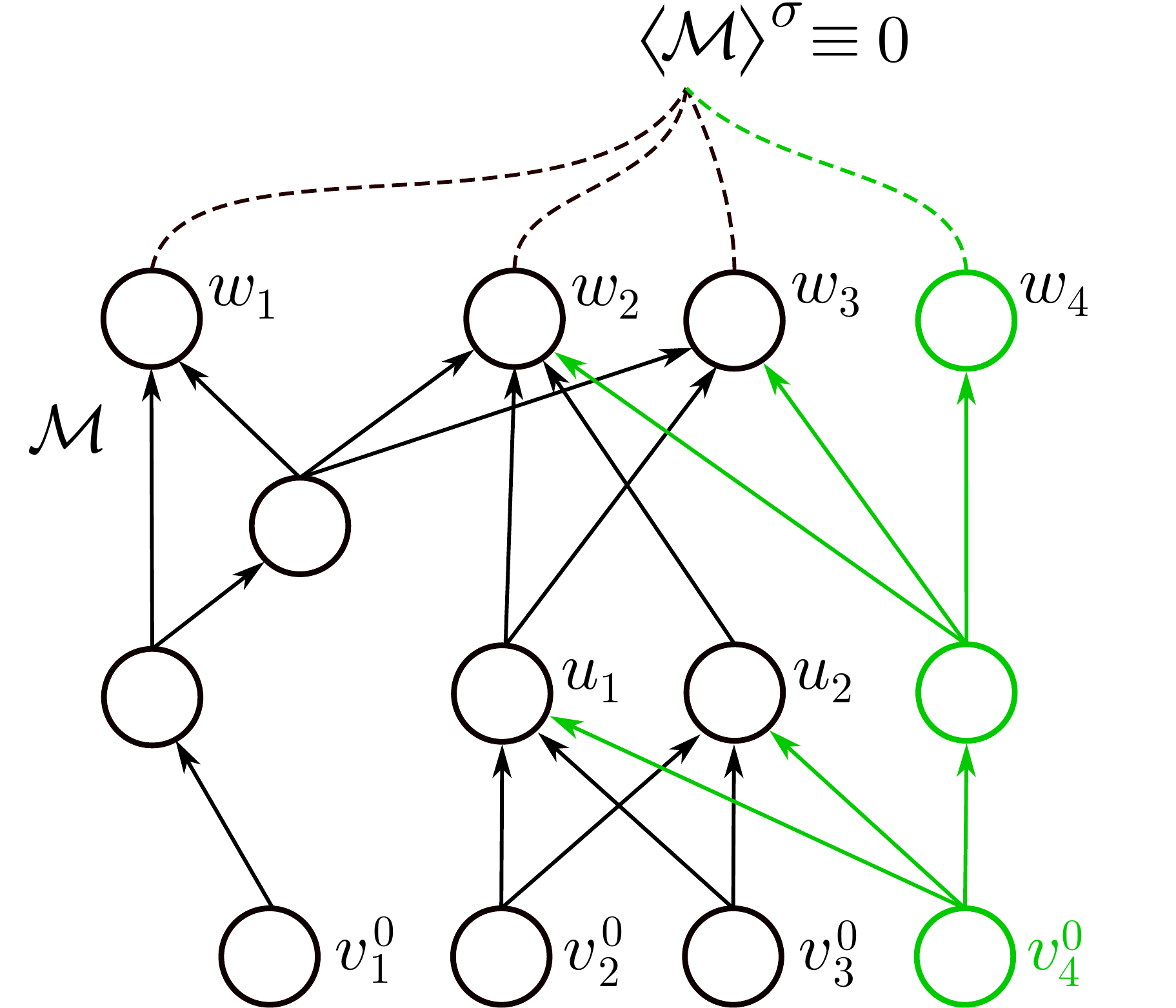}
\caption{A concrete example of anchoring the input at $v_4^0$ of a network $\m{M}$ to a real number $a$. The nodes of $\m{M}$ (in green) that are connected to $v_4^0$, but not to any of the remaining inputs $v_1^0,v_2^0,v_3^0$, are removed, while the rest of $\m{M}$ constitutes $\m{M}_a$. The anchored value $a$ is propagated through the removed parts of $\m{M}$, resulting in bias alteration at the nodes $u_1,u_2,w_2$, and $w_3$.
\label{fig:anchoring-g-intro}}
\end{figure}

\subsection{The class $\LtCla{a}{b}$ of nonlinearities}

The SAC and the CAC are admittedly rather technical conditions. However, unlike the null-net condition, which is a ``recursive''  statement about $\sigma$ (i.e., a statement about repeated compositions of affine functions and $\sigma$), the alignment conditions are statements about linear combinations of functions. The significance of Theorem \ref{thm:intro-Zab-NNC} thus lies in bridging the conceptual gap between the identifiability of single-layer networks and the identifiability of multi-layer networks, at least for meromorphic nonlinearities with simple poles only. In the present paper, we verify the SAC and the CAC for the class $\LtCla{a}{b}$ of ``$\tanh$-type'' nonlinearities introduced next.

\begin{definition}
Let $a,b>0$. The class $\LtCla{a}{b}$ consists of meromorphic functions $\sigma$ of the form
\begin{equation}\label{eq:Zab-series}
\sigma=C+\sum_{k\in\Z }c_k \big[\sgn(k)+\, \tanh\!\big(\pi b^{-1}(\,\cdot - ka)\big)\big],
\end{equation}
where $C\in\C$, and $\{c_k\}_{k\in\Z}$ is a sequence of complex numbers such that $\sup_{k\in\Z}|c_k|e^{-\pi a' |k|/b}<\infty$, for some $a'\in(0,a)$, and at least one $c_k$ is nonzero.
\end{definition}

\noindent The purpose of the $\sgn(k)$ term in \eqref{eq:Zab-series} is to make the series locally uniformly convergent even for sequences $\{c_k\}_{k\in\Z}$ growing exponentially with $k$ or $-k$.

\begin{theorem}\label{thm:intro-Zab-AC}
Let $a,b>0$ and let $\sigma\in\LtCla{a}{b}$. Then $\sigma$ satisfies the SAC and the CAC.
\end{theorem}

\noindent The proof of Theorem \ref{thm:intro-Zab-AC} is a generalization of the arguments presented above establishing the SAC and the CAC for the $\tanh$ nonlinearity. Specifically, it relies on the $ib$-periodicity of the nonlinearities in $\LtCla{a}{b}$ and the lattice geometry of their poles. As the proof involves the application of various ``point density'' techniques (such as the Kronecker-Weyl equidistribution theorem) to the poles of functions of the form $\sigma (\beta \cdot\, +\, \gamma + \epsilon(1/\cdot))$ (where $\epsilon$ is an ABC), Theorem \ref{thm:intro-Zab-AC} can be seen as a far-reaching refinement of the ``Deconstruction Lemma'' in \cite{Fefferman1994}.
We finally remark that our techniques can be adapted to prove the SAC and the CAC for nonlinearities of the form $\sigma(z)=r(e^z)$, where $r$ is a bounded non-constant real rational function with only simple poles.

The implications of Theorems \ref{thm:intro-NN-G}, \ref{thm:intro-NN-L}, \ref{thm:intro-Zab-NNC}, and \ref{thm:intro-Zab-AC} can now be summarized as follows:

\begin{theorem}\label{thm:intro-summary}
Let $a,b>0$, $\sigma\in\LtCla{a}{b}$, $D\in\N$, and let $\Vin$ be a nonempty finite set. Then $(\mathscr{N}^{\Vin,D}_{\mrm{G}},\sigma)$ and $(\mathscr{N}^{\Vin,D}_{\mrm{L}},\sigma)$ are identifiable up to $\sigisom$.
\end{theorem}

\noindent In particular, as Lemma \ref{lem:intro-tanh-signs} implies that the $\tanh$-isomorphism is none other than the relation $\sim_{\pm}$,  and $\tanh\in \LtCla{1}{\pi}$, Theorem \ref{thm:intro-summary} specializes to the following result.

\begin{prop}\label{prop:intro-tanh-ident}
Let $\Vin$ be a nonempty finite set and $D\in\N$. Then $(\mathscr{N}^{\Vin,D}_{\mrm{G}},\tanh)$ and $(\mathscr{N}^{\Vin,D}_{\mrm{L}},\tanh)$ are identifiable up to $\sim_{\pm}$.
\end{prop}

\noindent We remark that the characterization of irreducibility for the $\tanh$ nonlinearity according to Lemma \ref{lem:intro-tanh-irred} directly generalizes the concept of irreducibility in \cite{Sussman1992}, and is analogous to the no-clones condition in \cite{Vlacic2019}.

\subsection{Nonlinearities in $\LtCla{a}{b}$ with exotic affine symmetries}

Note that, given an arbitrary $\zeta\in \R$ and a finite set of real numbers $\{(\alpha_s,\beta_s,\gamma_s)\}_{s\in \m{I}}$, it is not clear whether there exists a nonlinearity with the affine symmetry $\left(\zeta,\{(\alpha_s,\beta_s,\gamma_s)\}_{s\in \m{I}}\right)$. It is likewise unclear if such a nonlinearity exists that additionally satisfies the null-net condition. Even though the existence of such nonlinearities would be desirable to justify the generality of the theory of $\rho$-modification and $\rho$-isomorphism presented in Section \ref{sec:a-theory-of-ident}, this is likely a difficult open problem. We are, however, able to offer a partial solution by showing that the class $\LtCla{a}{b}$ contains nonlinearities with (infinitely many) distinct affine symmetries that are more involved than the trivial and odd symmetries of the $\tanh$ function.

\begin{prop}\label{prop:intro-exotic-symm}
Let $\{\alpha_k\}_{k=0}^n$ be arbitrary nonzero real numbers with $n\geq 1$. Then there exist $b>0$, $\zeta\in\R$, and a $\sigma\in\LtCla{1}{b}$ such that $\left(\zeta,\{(\alpha_k,1,k)\}_{k=0}^n\right)$ is an affine symmetry of $\sigma$.
\end{prop}

 \subsection{Organization of the remainder of the paper}

We conclude this section by laying out the organization of the remainder of the paper. In Section \ref{sec:formal-NNT}, we formalize the concepts of $\rho$-modification and $\rho$-isomorphism and prove Theorems \ref{thm:intro-NN-G} and \ref{thm:intro-NN-L}. In Section \ref{sec:Clustering}, we analyze the pole structure of network maps with a meromorphic nonlinearity satisfying the SAC and the CAC, providing a formal proof of (a strengthened version of) Proposition \ref{prop:intro-vin-NNC}. In Section \ref{sec:inp-anch}, we introduce the procedure of input anchoring, allowing us to prove Theorem \ref{thm:intro-Zab-NNC}, and in Section VII, we analyze the fine properties of $\LtCla{a}{b}$--nonlinearities, allowing us, in turn, to prove Theorem \ref{thm:intro-Zab-AC}. Finally, the Appendix contains the proofs of various ancillary results that are either simple, standard, or based on ideas already seen in the main body of the paper.

\section{The $\rho$-isomorphism and the null-net theorems}\label{sec:formal-NNT}

\subsection{Irreducibility, regularity, $\rho$-modification, and the $\rho$-isomorphism}

We begin this chapter by formalizing the concepts of irreducibility and regularity, already introduced informally in Section \ref{sec:a-theory-of-ident}.

\begin{definition}[Irreducibility]\label{def:reduc}
Let $\m{N}=(V,E,\Vin,\Vout,\Omega,\Theta,\Lambda)$ be a GFNN with $D$-dimensional output, and let $\rho:\R\to\R$ be a nonlinearity. Let $U\subset V$ be a nonempty set of nodes, and suppose the following hold:
\begin{enumerate}[(i)]
\item the nodes in $U$ have a common parent set $P\subset V$, i.e., $\pre(u)=P$, for all $u\in U$,
\item  there exist sets of nonzero real numbers $\{\kappa_{v}\}_{v\in P}$ and $\{\beta_{u}\}_{u\in U}$ such that $\{\omega_{uv}\}_{v\in P}=\beta_u \{\kappa_v\}_{v\in P}$, for all $u\in U$, and
\item there exist a $\zeta\in \R$ and nonzero real numbers $\{\alpha_u\}_{u\in U}$ such that $\left(\zeta,\{(\alpha_u,\beta_u,\theta_u)\}_{u\in U}\right)$ is an affine symmetry of $\rho$.
\end{enumerate}
We then say that $\m{N}$ is $\rho\,$--\emph{reducible}. Whenever we wish to specify the set $U$ causing the reducibility, we will say that $\m{N}$ is \emph{$(\rho,U)$--reducible}. Finally, a GFNN that is not reducible will be called \emph{irreducible}.
\end{definition}

\begin{definition}[Regularity]
We say that a GFNN is \emph{regular} if it is irreducible and non-degenerate according to Definition \ref{def:NonDeg}.
The set of all regular GFNNs (repectively regular LFNNs) with $D$-dimensional output and input set $\Vin$ is denoted by $\mathscr{N}^{\Vin,D}_{\mrm{G}}$ (respectively $\mathscr{N}^{\Vin,D}_{\mrm{L}}$).
\end{definition}

We now formalize symmetry modification, already introduced informally in Section \ref{sec:a-theory-of-ident}. Before providing the definition, we motivate the concept by describing how an affine symmetry can be used to replace a single node in the network by newly-created nodes. Thus, let $\m{N}$ be a GFNN, and let $u^*$ be a non-input node of $\m{N}$ to be replaced. Write $P=\pre(u^*)$, and suppose that $B\subset V\setminus\{u^*\}$ is a set of nodes with parent set $P$ and such that there exist nonzero real numbers $\{\kappa_v\}_{v\in P}$ and $\{\beta_u\}_{u\in \{u^*\}\cup B}$ satisfying $\{\omega_{uv}\}_{v\in P}=\beta_u \{\kappa_v\}_{v\in P}$, for all $u\in \{u^*\}\cup B$. Suppose furthermore that the nonlinearity $\rho$ has an affine symmetry $\left(\zeta,\{(\alpha_u,\beta_u,\theta_u)\}_{u\in \{u^*\}\cup B}\cup  \{(\alpha'_p,\beta'_p,\gamma'_p)\}_{p=1}^n \right)$. 
Now, if $w$ is a node of $\m{N}$ with $u^*\in\pre(w)$, then, writing $K_P(t)=\sum_{v\in P}\kappa_v\OnoA{v}{\rho}(t)$, we have
\begin{equation}\label{eq:modif-intro-1}
\begin{aligned}
\OTnoA{w}{\rho}{\m{N}}(t)&=\rho\Bigg(\sum_{u\in \pre(w)}\omega_{wu}\OnoA{u}{\rho}(t)\;+\theta_w\Bigg)\\
&=\rho\Bigg(\sum_{u\in \{u^*\} \cup (B\,\cap \, \pre(w))}\hspace*{-5mm} \omega_{wu}\,\rho\big(\beta_u\,K_P(t)+\theta_u \big)\quad+\sum_{u\in \pre(w)\setminus (\{u^*\}\cup B)}\hspace*{-3mm} \omega_{wu}\OnoA{u}{\rho}(t) \quad + \theta_w\Bigg)\\
&=\rho\Bigg(\sum_{u\in B\,\cap\, \pre(w)}\hspace*{-2mm}  \left(\omega_{wu}-\!{\frac{\alpha_u \omega_{wu^*}}{\alpha_{u^*}}}\!\right)\! \rho\big(\beta_u K_P(t)+\theta_u \big) \; +\hspace*{-3mm}\sum_{u\in  B\setminus \pre(w) }\hspace*{-6mm}  -{\frac{\alpha_u \omega_{wu^*}}{\alpha_{u^*}}}\!\,\rho\big(\beta_u K_P(t)+\theta_u \big) \\
&\qquad+\sum_{p=1}^n{ -\frac{\alpha_p' \omega_{wu^*}}{\alpha_{u^*}} }\rho\big(\beta_p' K_P(t)+\gamma_p' \big)\; +\hspace*{-0mm} \sum_{u\in \pre(w)\setminus (\{u^*\}\cup B) }\hspace*{-3mm} \omega_{wu}\OnoA{u}{\rho}(t) \quad +\theta_w+{ \frac{\zeta\omega_{wu^*}}{\alpha_{u^*}}} \Bigg),
\end{aligned}
\end{equation}
for $t\in\R$. Therefore, $\m{N}$ can be modified without changing the map $\OTnoA{w}{\rho}{\m{N}}$ by removing the node $u^*$, replacing the weights $\omega_{wu}$ by $\omega_{wu}-{\textstyle\frac{\alpha_u \omega_{wu^*}}{\alpha_{u^*}}}$, for $u\in B\cap \pre(w)$, creating new edges $(u,w)$ with weights $-{\textstyle\frac{\alpha_u \omega_{wu^*}}{\alpha_{u^*}}}$, for $u\in  B\setminus \pre(w)$, adjoining $n$ new nodes $\{u_1',\dots, u_n'\}$ with biases $\gamma_p'$, incoming edges $(v,u_p')$ with weights $\beta_p'\kappa_v$, for $v\in P$, and outgoing edges $(u_p',w)$ with weights $-\frac{\alpha_p' \omega_{wu^*}}{\alpha_{u^*}}$, and finally replacing the bias $\theta_w$ by $\theta_w+\frac{\zeta\omega_{wu^*}}{\alpha_{u^*}}$.
Moreover, as the node $w$ with $u^*\in\pre(w)$ was arbitrary, this modification can be performed for all such $w$ simultaneously, therefore resulting in another network whose map is identical to $\OnoA{\m{N}}{\rho}$.
In this example only the node $u^*$ was removed. However, multiple nodes (the set $A$ in the next definition) can be removed at once in a similar manner, provided a suitable affine symmetry exists. We thus have the following formal definition:

\begin{definition}[$\rho$-modification]\label{def:modif}
Let $\m{N}=(V,E,\Vin,\Vout,\Omega,\Theta,\Lambda)$ be an irreducible GFNN with $D$-dimensional output, and let $\rho:\R\to\R$ be a nonlinearity. 
Let $A,B\subset V$, $A\neq \varnothing$, be disjoint sets of non-input nodes with a common parent set $P\subset V$, and let  $W=\{w\in V : \pre(w)\cap A\neq \varnothing\}$. Suppose the following are satisfied:
\begin{enumerate}[(i)]
\item there exists an affine symmetry
$
\left(\zeta,\{(\alpha_u,\beta_u,\theta_u)\}_{u\in A\cup B}\cup \{(\alpha'_p,\beta'_p,\gamma'_p)\}_{p=1}^n \right)
$ of $\rho$ with $n\geq 1$,
\item there exists a set of nonzero real numbers $\{\kappa_v\}_{v\in P}$ such that $\{\omega_{uv}\}_{v\in P}=\beta_u\{\kappa_v\}_{v\in P}$, for all $u\in A\cup B$,
\item $A\subset \pre(w)$, for all $w\in W$, and there exist nonzero real numbers $\{\nu_w\}_{w\in W}$ such that $\{\omega_{wu}\}_{u\in A}=\nu_w\{\alpha_u\}_{u\in A}$, for all $w\in W$,
\item either $A\cap \Vout=\varnothing$, or $A\subset \Vout$ and there exist real numbers $\{\mu_r\}_{r=1}^D$ such that $\{\lambda_{u}^{(r)}\}_{u\in A}=\mu_r\{\alpha_u\}_{u\in A}$, for all $r\in \{1,\dots, D\}$.
\end{enumerate}
We define a GFNN $\m{N}'=(V',E',\Vin,\Vout',\Omega',\Theta',\Lambda')$ by modifying $\m{N}$ as follows:
\begin{itemize}[--]
\item The nodes in $A$ and their incoming and outgoing edges are deleted, and a set $C=\{u_1',\dots,u_n'\}$ of $n$ new nodes (disjoint from $V$) is adjoined to the existing set of nodes $V$.
\item For $v\in P$ and $u_p'\in C$, an edge $(v,u_p')$ is created and assigned weight $\beta_p'\kappa_v$, and the node $u_p$ is assigned bias $\gamma_p'$.
\item For $w\in W$ and $u_p'\in C$, an edge $(u_p',w)$ is created and assigned weight $-\alpha_p'\nu_w$, and the bias $\theta_{w}$ is replaced by $\theta_{w}+\zeta\nu_w$.
\item  For $w\in W$ and $u\in B$
\begin{itemize}[-]
\item if $(u,w)\notin E$, an edge $(u,w)$ with weight $-\alpha_u \nu_{w}$ is created; otherwise
\item if $\omega_{wu}-\alpha_u \nu_{w}\neq0$, the weight $\omega_{wu}$ is replaced by $\omega_{wu}-\alpha_u \nu_{w}$, and 
\item if $\omega_{wu}-\alpha_u \nu_{w}= 0$, the edge $(u,w)$ is deleted.
\end{itemize}
\item If $A\cap \Vout=\varnothing$, then set $\Vout'=\Vout$ and $\Lambda'=\Lambda$, completing the construction.
\item If $A\subset \Vout$, then, for every $r\in\{1,\dots,D\}$,
\begin{itemize}[-]
\item the output scalar $\lambda^{(r)}$ is replaced by $\lambda^{(r)}+\zeta\mu_r$,
\item for $u_p'\in C$, new output scalars $\lambda_{u_p'}^{(r)}=-\alpha_p'\mu_r$ are created, and
\item for $u\in B\setminus \Vout$, new output scalars $\lambda_{u}^{(r)}=-\alpha_u\mu_r$ are created.
\item The set $B_{\mrm{out}}^{(r)}=\{u\in B\cap \Vout: \lambda_u^{(r)}-\alpha_u \mu_r\neq 0\}$ is defined, and, for $u\in B_{\mrm{out}}^{(r)}$, the output scalars $\lambda_u^{(r)}$ are replaced by $\lambda_u^{(r)}-\alpha_u\mu_r$.
\item Set $\Vout'=(\Vout\setminus (A\cup B))\; \cup \; \bigcup_{r=1}^D B_{\mrm{out}}^{(r)} \;\cup\;  C$ completing the construction.
\end{itemize}
\end{itemize}
We say that the so-constructed network $\m{N}'$ is a \emph{$(\rho\,;A,B,C)$--modification of $\m{N}$}. Whenever it is not necessary to explicitly specify the sets $A$, $B$, and $C$ involved in the modification, we will simply say that $\m{N}'$ is a \emph{$\rho$-modification} of $\m{N}$. A $\rho$-modification that is a regular network is called a \emph{regular $\rho$-modification}.
\end{definition}

The directed graph underlying the network $\m{N}'$ in Definition \ref{def:modif} is acyclic, as required by Definition \ref{def:GFNN}. Indeed, the nodes in the sets $A$ and $B$ all have the same parent set in $\m{N}$, and therefore have the same level, say $\ell_0$, whereas the nodes in $W$ have level at least $\ell_0+1$ in $\m{N}$. The nodes in $C$ will then have level $\ell_0$ in $\m{N}'$, and the nodes $w\in W$ will satisfy $\lvl_{\m{N}'}(w)=\lvl_{\m{N}}(w)$. A concrete example of a $\rho$-modification is shown schematically in Figure \ref{fig:ModificationFormal}.

\begin{figure}[h!]\centering
\includegraphics[height=56mm,angle=0]{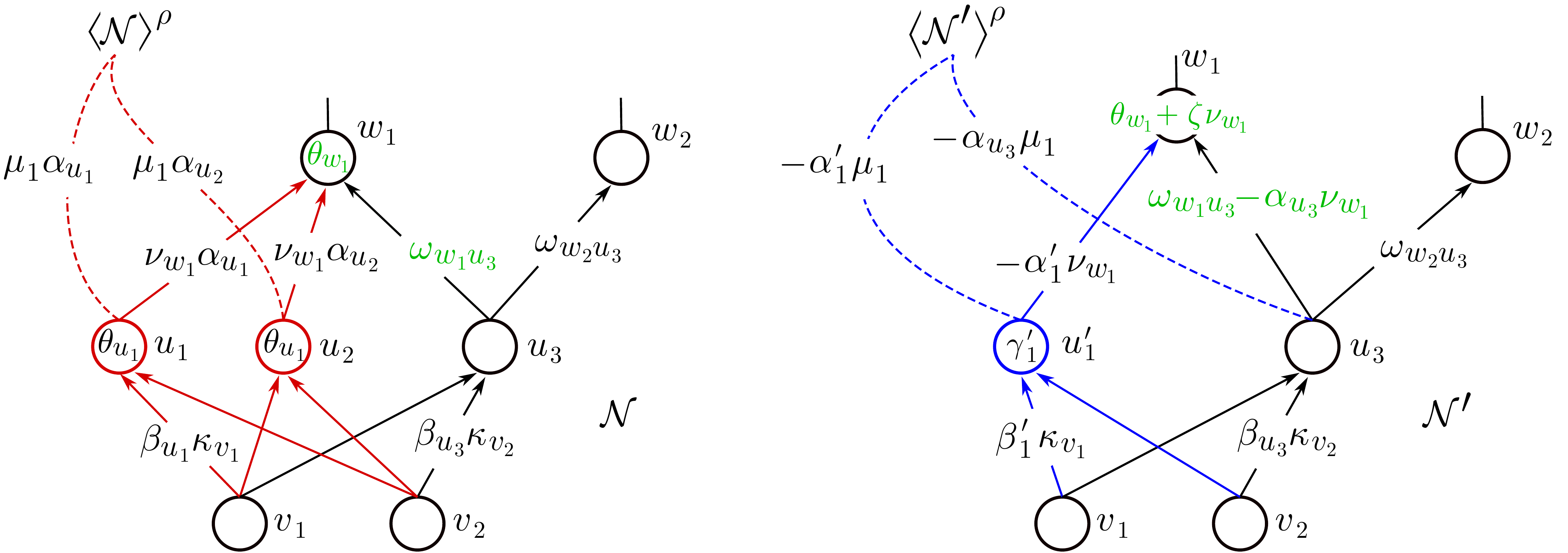}
\caption{Consider a network $\m{N}$ with nodes $P=\{v_1,v_2\}$, $A=\{u_1,u_2\}$, $B=\{u_3\}$, and $W=\{w_1\}$ as shown on the left. Suppose that $\rho$ has an affine symmetry $\big(\zeta,\{(\alpha_{u_1},\beta_{u_1},\theta_{u_1}), (\alpha_{u_2},\beta_{u_2},\theta_{u_2}), (\alpha_{u_3},\beta_{u_3},\theta_{u_3}), \allowbreak  (\alpha_{u_1'},\beta_{u_1'},\theta_{u_1'}) \} \big)$ and that there exist sets $\{\kappa_{v_1}, \kappa_{v_2}\}$, $\{\nu_{w_1}\}$, and $\{\mu_1\}$ so that Conditions (i) -- (iv) in Definition \ref{def:modif} hold. Then $\m{N}$ admits a $(\rho\,;A,B,C)$--modification $\m{N}'$ (right), where $C=\{u_1'\}$. The parts of $\m{N}$ in red are removed in the process, the parts in blue are added, and those in green are altered. Note that, if $\omega_{w_1u_3}-\alpha_{u_3}\nu_{w_1}= 0$, then the edge $(u_3,w_1)$  would have to be deleted as well. \label{fig:ModificationFormal}}
\end{figure}

Note that the set $B$ in Definition \ref{def:modif} is allowed to be empty, but the sets $A$ and $C$ must be nonempty. In particular, Definition \ref{def:modif} does not encompass $\rho$-reduction, in contrast to the informal definition of $\rho$-modification provided in Section \ref{sec:a-theory-of-ident}. This is in order to avoid the scenario described in Figure \ref{fig:reduction2} that necessitates further alteration to obtain a network without ``constant parts''. Moreover, restricting the number of possibilities in which $\rho$-modification can be carried out renders the claims of Theorems \ref{thm:intro-NN-G} and \ref{thm:intro-NN-L} stronger.

The following proposition summarizes the properties of GFNNs that are readily seen to be preserved under $\rho$-modification.

\begin{prop}\label{prop:modif-basic-prop}
Let $\m{N}_1=(V^1,E^1,\Vin,\Vout^1,\Omega^1,\Theta^1,\Lambda^1)$ be a GFNN with $D$-dimensional output, let $\rho$ be a nonlinearity, and let $\m{N}_2=(V^2,E^2,\Vin,\Vout^2,\Omega^2,\Theta^2,\Lambda^2)$ be a $\rho$-modification of $\m{N}_1$. Then,
\begin{enumerate}[(i)]
\item if $\m{N}_1$ is layered, then $\m{N}_2$ is also layered,
\item $\m{N}_1$ is a $\rho$-modification of $\m{N}_2$, and
\item $\outmap{\m{N}_1}{\rho}=\outmap{\m{N}_2}{\rho}$.
\end{enumerate}
\end{prop}

These properties naturally lead to the following definition of isomorphism up to $\rho$-modification. For example, the networks $\m{N}_1$, $\m{N}_2$, $\m{N}_3$, and $\m{N}_4$ in Figure \ref{fig:modification} are $\rho_c$-isomorphic.

\begin{definition}[$\rho$-isomorphism]\label{def:symiso}
Let $\m{N}$ and $\m{M}$ be regular GFNNs with $D$-dimensional output and the same input set, and let $\rho:\R\to\R$ be a nonlinearity. We say that $\m{N}$ is $\rho$-isomorphic to $\m{M}$, and write $\m{N}\rhoisom\m{M}$, if there exists a finite sequence $\m{N}_1,\dots, \m{N}_n$ of regular GFNNs with $D$-dimensional output and the same input set such that $\m{N}_1=\m{N}$, $\m{N}_n=\m{M}$, and, for $j\in\{1,\dots,n-1\}$, $\m{N}_{j+1}$ is a regular $\rho$-modification of $\m{N}_j$.
\end{definition}

\begin{prop}\label{prop:rhoisom-equiv}
The binary relation $\rhoisom$ is an equivalence relation on both $\mathscr{N}^{\Vin,D}_{\mrm{G}}$ and $\mathscr{N}^{\Vin,D}_{\mrm{L}}$, and if $\m{N}\rhoisom\m{M}$, then $\outmap{\m{N}}{\rho}=\outmap{\m{M}}{\rho}$.
\end{prop}
\begin{proof}
By item (i) of Proposition \ref{prop:modif-basic-prop}, $\rhoisom$ is a relation on both $\mathscr{N}^{\Vin,D}_{\mrm{G}}$ and $\mathscr{N}^{\Vin,D}_{\mrm{L}}$.
Reflexivity and transitivity follow immediately from Definition \ref{def:symiso}. To establish symmetry, let $\m{N}$ and $\m{M}$ be regular GFNNs with $D$-dimensional output, and let $\rho:\R\to\R$ be a nonlinearity. Suppose that $\m{N}\rhoisom\m{M}$ and let $\m{N}_1,\dots, \m{N}_n$, $n\geq 1$, be a sequence of regular GFNNs as in Definition \ref{def:symiso}. Then, by item (ii) in Proposition \ref{prop:modif-basic-prop}, we know that $\m{N}_{j}$ is a $\rho$-modification of $\m{N}_{j+1}$, for all $j\in\{1,\dots,n-1\}$, and thus $\m{M}=\m{N}_{n},\m{N}_{n-1},\dots, \m{N}_1=\m{N}$ is a sequence establishing $\m{M}\rhoisom\m{N}$, and thereby symmetry of the relation $\rhoisom$. Moreover, we have
\begin{equation*}
\OnoA{\m{N}}{\rho}=\OnoA{\m{N}_1}{\rho}=\OnoA{\m{N}_{2}}{\rho}=\dots=\OnoA{\m{N}_n}{\rho}=\OnoA{\m{M}}{\rho},
\end{equation*}
as desired.
\end{proof}

We note that trivial networks $\m{T}^{\,\Vin,\, D}$ do not admit any $\rho$-modifications (simply as they do not have any non-input nodes), and therefore the only network that is $\rho$-isomorphic to $\m{T}^{\,\Vin,\, D}$ is $\m{T}^{\,\Vin,\, D}$ itself.

\subsection{Subnetworks and proofs of the null-net theorems}

The following proposition is the cornerstone of the null-net theorems.

\begin{prop}\label{prop:nonisom->Enullnet}
Let $\m{N}_1$ and $\m{N}_2$ be regular GFNNs, both with $D$-dimensional output and the same input set $\Vin$, and let $\rho$ be a nonlinearity. Suppose that $\m{N}_1$ and $\m{N}_2$ are not $\rho$-isomorphic and $\outmap{\m{N}_1}{\rho}=\outmap{\m{N}_2}{\rho}$. Then  there exists a non-trivial regular GFNN $\m{A}$ (layered if $\m{N}_1$ and $\m{N}_2$ are layered) with one-dimensional output and input set $\Vin$ such that $\outmap{\m{A}}{\rho}= 0$.
\end{prop}

The proof of Proposition \ref{prop:nonisom->Enullnet} relies crucially on being able to perform  $\rho$-modification in a manner that preserves regularity. Unfortunately, neither irreducibility nor non-degeneracy are generally preserved under $\rho$-modification. 
The following proposition, however, tells us that, for every $\rho$-modification of a regular GFNN, there exists an alternative (but related) $\rho$-modification that preserves regularity, which will suffice for the purpose of proving Proposition \ref{prop:nonisom->Enullnet}.

\begin{prop}\label{prop:modif-regul}
Let $\m{N}$ be a regular GFNN with $D$-dimensional output, let $\rho:\R\to\R$ be a nonlinearity, and let $A_0,B_0$ be disjoint sets of nodes of $\m{N}$ with a common parent set $P$ such that $\m{N}$ admits a $(\rho\,;A_0,B_0,C_0)$--modification. Then there exist disjoint sets $A\supset A_0$ and $B$ of nodes with common parent set $P$, and a $C\subset C_0$, such that $\m{N}$ admits a regular $(\rho\,;A,B,C)$--modification.
\end{prop}

The proof of Proposition \ref{prop:modif-regul} proceeds via the next two lemmas (proved in the Appendix) that treat the irreducibility and non-degeneracy aspects of regularity separately.
To motivate the first lemma, we note that $\rho$-modification can be seen as a process whereby certain nodes $A$ are removed from a GFNN by replacing their maps with a combination of the maps of nodes $B$ already present in the GFNN, as well as several ``nascent'' nodes $C$. However, if we add too many nascent nodes $C$ at once, we might provoke reducibility in the resulting network. This situation is illustrated in Figure \ref{fig:ModificationLemma3}.
Our lemma thus shows that irreducibility can be preserved by ``modifying frugally'', i.e., by adding the least possible number of nodes $C$ that facilitates $\rho$-modification:

\begin{lemma}\label{lem:modif-irred}
Let $\m{N}$ be an irreducible GFNN with $D$-dimensional output, let $\rho:\R\to\R$ be a nonlinearity, and let $A_0,B_0$ be disjoint sets of nodes of $\m{N}$ with a common parent set $P$ such that $\m{N}$ admits a $(\rho\,;A_0,B_0,C_0)$--modification. Let $C\subset C_0$ be a set of least possible cardinality so that there exist disjoint sets $A\supset A_0$ and $B$ of nodes of $\m{N}$ with a common parent set $P$ such that $\m{N}$ admits a $(\rho\,;A,B,C)$--modification $\m{N}'$. Then $\m{N}'$ is irreducible.
\end{lemma}

\begin{figure}[h!]\centering
\includegraphics[width=\textwidth,angle=0]{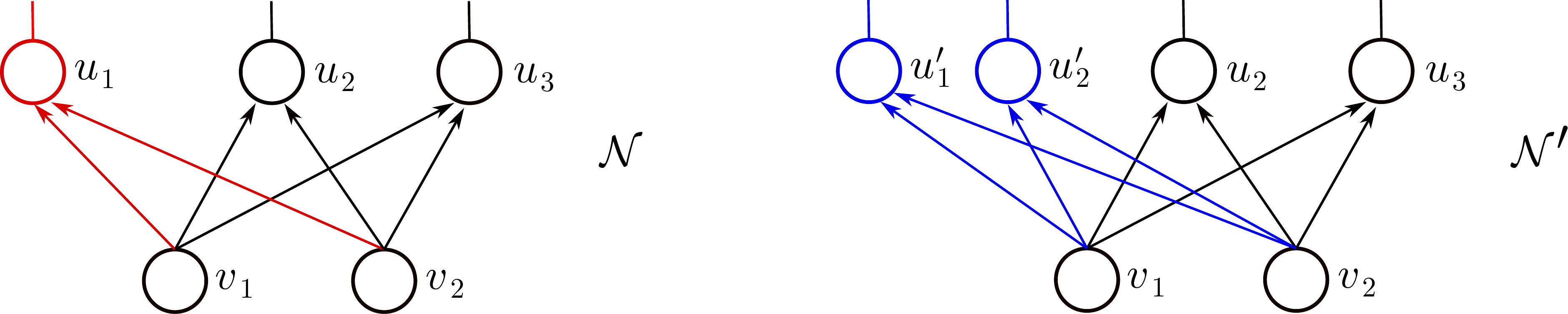}
\caption{Consider a regular network $\m{N}$ as shown on the left, and suppose that it admits a $(\rho\,;A,B,C)$--modification $\m{N}'$ (shown on the right), where $A=\{u_1\}$, $B=\{u_2\}$, and $C=\{u_1', u_2'\}$. The network $\m{N}'$ might then be $(\rho,U)$--reducible, for some $U\subset B\cup C \cup \{u_3\}$ containing at least one element of $C$.
\label{fig:ModificationLemma3}}
\end{figure}

To motivate the second lemma, note that the $(\rho\,; A,B,C)$--modification of a non-degenerate network $\m{N}$ is degenerate precisely if there exists a node $u^*\in B$ that loses all its outgoing edges in the process, and, if $u^*$ is an output node of $\m{N}$, all its output scalars are set to zero. Degeneracy can thus be avoided by performing an alternative $\rho$-modification that, in addition to the nodes in $A$, removes such problematic nodes as well.

\begin{lemma}\label{lem:modif-nondeg}
Let $\m{N}$ be a non-degenerate GFNN with $D$-dimensional output, let $\rho:\R\to\R$ be a nonlinearity, and let $A,B$ be disjoint sets of nodes of $\m{N}$ with a common parent set such that $\m{N}$ admits a $(\rho\,;A,B,C)$--modification.
Then there exists a set $B^*\subset B$ such that $\m{N}$ admits a non-degenerate $(\rho\,;A\cup B^*,B\setminus B^*,C)$--modification $\m{N}'$.
\end{lemma}

\noindent We are now ready to prove Proposition \ref{prop:modif-regul}.

\begin{proof}[Proof of Proposition \ref{prop:modif-regul}]
Let $C\subset C_0$ be a subset of minimal cardinality such that $\m{N}$ admits a $(\rho\;,A',B',C)$--modification, for some disjoint sets $A'\supset A_0$ and $B'$ of nodes of $\m{N}$ with parent set $P$. Now, as $\m{N}$ is regular and hence non-degenerate, we have by Lemma \ref{lem:modif-nondeg} that there exists a $B^*\subset B'$ such that $\m{N}$ admits a non-degenerate $(\rho\,;A'\cup B^*,B'\setminus B^*,C)$--modification $\m{N}'$. As $\m{N}$ is irreducible, it follows by Lemma \ref{lem:modif-irred} that $\m{N}'$ is irreducible, and thus $\m{N}'$ is the desired regular $\rho$-modification of $\m{N}$.
\end{proof}

In order to prove Proposition \ref{prop:nonisom->Enullnet}, we will also need the following definition of a subnetwork of a GFNN:

\begin{definition}[Subnetwork]\label{def:GFNNsubnet}
Let $\m{N}=(V,E,\Vin,\Vout,\Omega,\Theta,\Lambda)$ be a GFNN with $D$-dimensional output.
A \emph{subnetwork of $\m{N}$}
 is a GFNN
$\m{N}'=(V',E',\Vin',\Vout',\Omega',\Theta',\Lambda')$ with $D'$-dimensional output such that there exists a set $S\subset V$ so that
\begin{enumerate}[(i)]
\item $V'=\anc_{\m{N}}(S)$,
\item $E'=\{(v,\wtd{v})\in E: v,\wtd{v}\in V'\}$,
\item $\Vin'\supset \Vin\cap V'$,
\item $\Omega'=\{\omega_{\wtd{v}v}:(v,\wtd{v})\in E'\}$,
\item $\Theta'=\{\theta_{v}:v\in V'\}$.
\end{enumerate}
Whenever we wish to specify explicitly the set $S$ giving rise to $\m{N}'$, we will say that $\m{N}'$ is a \emph{subnetwork of $\m{N}$ generated by $S$.}
\end{definition}
\noindent Note that subnetworks generated by a set $S$ are not unique. They become unique, though, if we also specify their input and output sets $\Vin'$ and $\Vout'$, and their set of output scalars $\Lambda'$.

\begin{proof}[Proof of Proposition \ref{prop:nonisom->Enullnet}]
Let $\m{N}_1=(V^1,E^1,\Vin,\Vout^1,\Omega^1,\Theta^1,\Lambda^1)$ and $\m{N}_2 =(V^2,E^2,\Vin,\Vout^2,\allowbreak\Omega^2,\Theta^2,\Lambda^2)$ be as in the proposition statement, and let $\mathscr{M}$ be the set of regular GFNNs $\m{M}$ with the following properties
\begin{itemize}[--]
\item the input set of $\m{M}$ is $\Vin$,
\item $\m{M}$ is a subnetwork of $\m{N}_1$, and
\item $\m{N}_2$ is $\rho$-isomorphic to some regular GFNN containing $\m{M}$ as a subnetwork.
\end{itemize}
We introduce a partial order $\prec$ on $\mathscr{M}$ by setting $\m{M}'\prec\m{M}$ if and only if $\m{M}'$ is a subnetwork of $\m{M}$. Now, let $\m{M}_1$ be a maximal element of $\mathscr{M}$ with respect to $\prec$, and let $\m{N}_3=(V^{3},E^{3},\Vin,\Vout^{3},\Omega^{3},\Theta^{3},\allowbreak \Lambda^{3})$ be a regular GFNN such that $\m{N}_2\rhoisom\m{N}_3$ and $\m{M}_1$ is a subnetwork of $\m{N}_3$.

Note that both $\m{N}_1$ and $\m{N}_3$ contain $\m{M}_1$ as a subnetwork. In particular, the set of nodes of $\m{M}_1$ is given by ${V}^{\m{M}_1}=V^1\cap V^{3}$.
Furthermore, as $\m{N}_2\rhoisom\m{N}_3$, we have by Proposition \ref{prop:modif-basic-prop}:
\begin{equation}\label{eq:morphProp-1}
\begin{aligned}
0&=\big(\outmap{\m{N}_1}{\rho}\big)_r- \big(\outmap{\m{N}_{2}}{\rho}\big)_r=\big(\outmap{\m{N}_1}{\rho}\big)_r-\big(\outmap{\m{N}_3}{\rho}\big)_r\\
&=\lambda^{(r),1}-\lambda^{(r),3} +\sum_{w\in \Vout^1\setminus {\Vout^{3}}} \lambda^{(r),1}_w\outmapT{w}{\rho}{\m{N}_1}\\
&\quad+\sum_{w\in \Vout^1\cap {\Vout^{3}}}\left(\lambda^{(r),1}_w-\lambda^{(r),3}_w \right)\outmapT{w}{\rho}{\m{N}_1}-\sum_{w\in {\Vout^{3}}\setminus \Vout^1} \lambda^{(r),3}_w\outmapT{w}{\rho}{\m{N}_3},
\end{aligned}
\end{equation}
for all $r\in\{1,\dots, D\}$.
We now show the following:

\textit{Claim: there exist an $r\in\{1,\dots,D\}$ and a $w\in \Vout^1\cup {\Vout^{3}}$ such that at least one of the following three statements holds:
\begin{equation}
\begin{aligned}
w&\in \Vout^1\setminus {\Vout^{3}} && \text{ and } && \lambda^{(r),1}_w\neq 0,  \\
 w&\in \Vout^1\cap {\Vout^{3}} && \text{ and } && \lambda^{(r),1}_w-\lambda^{(r),3}_w\neq 0,   \label{eq:morphProp-4}\\
w&\in {\Vout^{3}}\setminus \Vout^1 && \text{ and } && \lambda^{(r),3}_w\neq 0. 
\end{aligned}
\end{equation}
}

\noindent\textit{Proof of Claim.} Suppose by way of contradiction that this is not the case, i.e.,  we have 
\begin{equation*}
\begin{aligned}
\lambda^{(r),1}_w&= 0, &&\text{ for all }w\in \Vout^1\setminus {\Vout^{3}},\\
 \lambda^{(r),1}_w-\lambda^{(r),3}_w&= 0, &&\text{ for all }w\in \Vout^1\cap {\Vout^{3}},\text{ and }\\
\lambda^{(r),3}_w&=0, &&\text{ for all }w\in {\Vout^{3}}\setminus \Vout^1,
\end{aligned}
\end{equation*}
for all $r$. Then, as $\m{N}_1$ is non-degenerate, Property (ii) in Definition \ref{def:NonDeg} implies that $\Vout^1\setminus {\Vout^{3}}=\varnothing$, i.e., $\Vout^1\subset {\Vout^{3}}$. Similarly, as $\m{N}_3$ is non-degenerate, we have $\Vout^{3}\setminus {\Vout^1}=\varnothing$, and thus $\Vout^1={\Vout^{3}}\subset{V}^{\m{M}_1}$. But then we have $V^1\setminus \Vin=V^{\m{M}_1}\setminus \Vin=V^{3}\setminus \Vin$, again by non-degeneracy of $\m{N}_1$ and $\m{N}_3$. Next, as $\lambda^{(r),1}_w=\lambda^{(r),3}_w$, for all $w\in \Vout^1=\Vout^{3}$, it follows from \eqref{eq:morphProp-1} that $\lambda^{(r),1}=\lambda^{(r),3}$, for all $r\in\{1,\dots,D\}$. Thus $\m{N}_1=\m{N}_3$, contradicting the assumption that $\m{N}_1$ and $\m{N}_2$ are not $\rho$-isomorphic. This establishes the Claim.

 Now, for $r\in\{1,\dots,D\}$, set 
\begin{equation*}
S^r=\{w\in \Vout^1\cup \Vout^{3}:\text{one of the statements in \eqref{eq:morphProp-4} holds}\}.
\end{equation*}
Furthermore, for $r\in\{1,\dots,D\}$ and $w\in S^r$, let 
\begin{equation*}
{\lambda}_{w}^{(r)}=
\begin{cases}
\lambda^{(r),1}_w, &\text{if }w\in \Vout^1\setminus {\Vout^{3}}\\
\lambda^{(r),1}_w-\lambda^{(r),3}_w, &\text{if }w\in \Vout^1\cap {\Vout^{3}}\\
\lambda^{(r),3}_w, &\text{if }w\in \Vout^{3}\setminus {\Vout^{1}}\\
\end{cases},
\end{equation*}
and set ${\Lambda}^{(r)}=\{{\lambda}^{(r)}\coleqq 0\}\cup\{\lambda_{w}^{(r)}:w\in S^r\}$ and ${\Lambda}=\bigcup_{r=1}^D {\Lambda}^{(r)}$. By the Claim we know that there exists an $r^*$ such that $S^{r^*}\neq\varnothing$. Moreover, as $S^{r^{*}}\subset \Vout^1\cup \Vout^{3}$, we have $S^{r^*}\cap \Vin=\varnothing$.

Now, define the ``combined'' network $\wtd{\m{A}}=(V^1\cup V^{3}, E^1\cup E^{3}, \Vin, \Vout^1\cup \Vout^{3}, \Theta^1\cup\Theta^{3},\Omega^1\cup\Omega^{3},{\Lambda})$. Finally, let $\m{A}$ be the subnetwork of $\wtd{\m{A}}$ generated by $S^{r^*}$, with input set $\Vin$, output set $S^{r^*}$, and output scalars ${\Lambda}^{(r^*)}$.
Then $\m{A}$ is non-degenerate by construction, $\m{A}$ is not the trivial network  $\m{T}^{\,\Vin,\, 1}$ as $S^{r^*}\cap \Vin=\varnothing$, and by \eqref{eq:morphProp-1} we have
$
\outmap{\m{A}}{\rho}=
0
$. Moreover, if $\m{N}_1$ and $\m{N}_2$ are layered, then $\m{N}_3$ is layered as it is $\rho$-isomorphic to $\m{N}_2$, and hence $\m{A}$ is layered as well.

It remains to show that $\m{A}$ is irreducible. As $\m{A}$ is a subnetwork of $\wtd{\m{A}}$, it suffices to show that $\wtd{\m{A}}$ is irreducible. Assume by way of contradiction that $\wtd{\m{A}}$ is $(\rho,U)$--reducible for some $U\subset V^1\cup V^{3}$. As $\m{N}_1$ and $\m{N}_3$ are both irreducible, we must have $U\not\subset V^1$ and $U\not\subset V^{3}$. In particular, we have $U\cap V^{3}\neq\varnothing$, $ U\setminus V^{3}\neq \varnothing$, and the common parent set $P$ of the nodes $U$ is contained in $V^{\m{M}_1}$.
By definition of reducibility, there exist sets of nonzero real numbers $\{\kappa_{v}\}_{v\in P}$ and $\{\beta_u\}_{u\in U}$ such that $\{\omega_{uv}\}_{v\in P}=\beta_u\{\kappa_v\}_{v\in P}$, for all $u\in U$, as well as a $\zeta\in\R$ and nonzero real numbers $\{\alpha_u\}_{u\in U}$ such that $(\zeta,\{(\alpha_u,\beta_u,\theta_u)_{u\in U}\})$ is an affine symmetry of $\rho$. Now, by definition of affine symmetry,
\begin{equation}\label{eq:morphProp-2}
\sum_{u\in U}\alpha_u \,\rho(\beta_{u}t+\theta_u)=\zeta\,\bm{1}(t),\quad\text{for all }t\in\R.
\end{equation}
Fix an arbitrary node $u^*\in U\cap V^{3}$ and let $B_0\coleqq (U\cap V^{3})\setminus\{u^*\}$. Then \eqref{eq:morphProp-2} can be rearranged to get
\begin{equation*}\label{eq:morphProp-3}
\left[\alpha_{u^*}\rho(\beta_{u^*}t+\theta_{u^*})+\!\sum_{u\in B_0 }\alpha_u \,\rho(\beta_{u}t+\theta_u)\right]+\!\sum_{u\in U\setminus V^{3}}\alpha_u \,\rho(\beta_{u}t+\theta_u)=\zeta\,\bm{1}(t) ,\quad t\in\R.
\end{equation*}
It follows that $\m{N}_3$ admits a $\left(\rho\,; \{u^*\},B_0,U\setminus V^{3}\right)$--modification. Now, by Proposition \ref{prop:modif-regul}, there exist disjoint sets $A\supset\{u^*\}$ and $B$ of nodes of $\m{N}_3$ with parent set $P$ and a $C\subset U\setminus V^{3}\subset V^{1}\setminus V^{3}$ such that $\m{N}_3$ admits a regular $\left(\rho\,; A,B,C\right)$--modification $\m{N}_4$. In particular $C\neq\varnothing$, and thus an arbitrary subnetwork $\m{M}_2$ of $\m{N}_1$ generated by $C\cup V^{\m{M}_1}$ is an element of $\mathscr{M}$ with $\m{M}_1\prec\m{M}_2$ and $\m{M}_1\neq \m{M}_2$, contradicting the maximality of $\m{M}_1$. This establishes that $\wtd{\m{A}}$ is irreducible and concludes the proof.
\end{proof}

\begin{definition}[General null-net condition]
Let $\rho$ be a nonlinearity and $\Vin$ a nonempty finite set. We say that $\rho$ satisfies the \emph{general null-net condition on $\Vin$} if the only regular GFNN $\m{A}$ with one-dimensional output and input set $\Vin$ such that $\OnoA{\m{A}}{\rho}=0$ on $\R^{\Vin}$ is the trivial network $\m{T}^{\,\Vin,\,1}$.
\end{definition}

\begin{theorem}[Null-net theorem for GFNNs]\label{thm:NN-G}
Let $\mathscr{N}^{\Vin,D}_{\mrm{G}}$ be the set of all regular GFNNs with input set $\Vin$ and $D$-dimensional output, and let $\rho$ be a nonlinearity. Then $(\mathscr{N}^{\Vin,D}_{\mrm{G}},\rho)$ is identifiable up to $\rhoisom$ if and only if $\rho$ satisfies the general null-net condition on $\Vin$.
\end{theorem}

\noindent The general null-net condition and Theorem \ref{thm:NN-G} have corresponding versions for layered networks:

\begin{definition}[Layered null-net condition]
Let $\rho$ be a nonlinearity and $\Vin$ a nonempty finite set. We say that $\rho$ satisfies the \emph{layered null-net condition on $\Vin$} if the only regular LFNN $\m{A}$ with one-dimensional output and input set $\Vin$ such that $\OnoA{\m{A}}{\rho}=0$ on $\R^{\Vin}$ is the trivial network $\m{T}^{\,\Vin,\,1}$.
\end{definition}

\begin{theorem}[Null-net theorem for LFNNs]\label{thm:NN-L}
Let $\mathscr{N}^{\Vin,D}_{\mrm{L}}$ be the set of all regular LFNNs with input set $\Vin$ and $D$-dimensional output, and let $\rho$ be a nonlinearity. Then $(\mathscr{N}^{\Vin,D}_{\mrm{L}},\rho)$ is identifiable up to $\rhoisom$ if and only if $\rho$ satisfies the layered null-net condition on $\Vin$.
\end{theorem}

As the proofs of Theorems  \ref{thm:NN-G} and \ref{thm:NN-L} are completely analogous, we present them jointly. The proof is a straightforward consequence of Proposition \ref{prop:nonisom->Enullnet}.

\begin{proof}[Proof of Theorems \ref{thm:NN-G} and \ref{thm:NN-L}]
Proposition \ref{prop:rhoisom-equiv} implies directly that $\rhoisom$ satisfies \eqref{eq:ident-compat} for both $(\mathscr{N}^{\Vin,D}_{\mrm{G}},\rho)$ and $(\mathscr{N}^{\Vin,D}_{\mrm{L}},\rho)$. Next, suppose that $(\mathscr{N}^{\Vin,D}_{\mrm{G}},\rho)$ (respectively $(\mathscr{N}^{\Vin,D}_{\mrm{L}},\rho)$) is not identifiable up to $\rhoisom$, and let $\m{N}_1,\m{N}_2\in \mathscr{N}^{\Vin,D}_{\mrm{G}}$ (respectively $\m{N}_1,\m{N}_2\in \mathscr{N}^{\Vin,D}_{\mrm{L}}$) be non-$\rho$-isomorphic and such that $\outmap{\m{N}_1}{\rho}=\outmap{\m{N}_2}{\rho}$. Then, by Proposition \ref{prop:nonisom->Enullnet}, there exists a non-trivial regular GFNN (respectively LFNN) $\m{A}$ with one-dimensional output and input set $\Vin$ such that $\outmap{\m{A}}{\rho}=0$. Therefore, $\rho$ fails the general (respectively layered) null-net condition on $\Vin$.

Conversely, suppose that $\rho$ does not satisfy the general (respectively layered) null-net condition on $\Vin$, and let $\m{A}$ be a non-trivial regular GFNN (respectively LFNN) with one-dimensional output such that $\outmap{\m{A}}{\rho}=0$. Then the networks $\m{T}^{\,\Vin,\, 1}$ and $\m{A}$ are regular GFNNs (respectively LFNNs) satisfying $\outmap{\m{T}^{\,\Vin,\, 1}}{\rho}=0=\outmap{\m{A}}{\rho}$, and are not $\rho$-isomorphic (simply as the only network that is $\rho$-isomorphic to $\m{T}^{\,\Vin,\, D}$ is $\m{T}^{\,\Vin,\, D}$ itself). Hence, $(\mathscr{N}^{\Vin,D}_{\mrm{G}},\rho)$ (respectively $(\mathscr{N}^{\Vin,D}_{\mrm{L}},\rho)$) is not identifiable up to $\rhoisom$, completing the proof.

\end{proof}

\section{Pole clustering for single-input network maps with a meromorphic nonlinearity satisfying the SAC and the CAC} \label{sec:Clustering}

Throughout this section we fix a meromorphic nonlinearity $\sigma$ such that
\begin{itemize}[--]
\item $\sigma(\R)\subset\R$,
\item $\sigma$ has infinitely many simple poles and no poles of higher order, and
\item $\sigma$ satisfies the SAC and the CAC. 
\end{itemize}
In this section we formally establish that the map of every non-trivial regular GFNN $\m{N}$ with 1-dimensional output and a singleton input set can be analytically continued to a domain with countable complement in $\C$, and that the set of simple poles of $\OnoA{\m{N}}{\sigma}$ is nonempty. We will, in fact, prove a much stronger result about the structure of the singularities of $\OnoA{\m{N}}{\sigma}$. In order to state this result, we need the concept of clustering depth introduced next.

We write $D(a,r)=\{z\in\C:|z-a|\leq r\}$ and $D^\circ(a,r)=\{z\in\C:|z-a|< r\}$ respectively for the closed and open disk in $\C$ of radius $r\geq 0$ centered at $a\in\C$. 

\begin{definition}[Cluster sets and clustering depth]\label{def:cluster}
Let $E\subset\C$ be a set and let $z\in \C$ be a point.
\begin{enumerate}[(i)]
\item For a nonnegative integer $k$ we define the \emph{$k^{\text{th}}$ cluster set $\m{C}^k(E)$ of $E$} inductively as follows:
	\begin{itemize}[--]
	\item We set $\m{C}^0(E)=E$, and
	\item for $k\geq 1$, we let $\m{C}^k(E)$ be the set of cluster points of $\m{C}^{k-1}(E)$. 
	\end{itemize}
\item We define the \emph{clustering depth $L_{\m{C}}(E)$ of $E$} as the least $k$ for which $\m{C}^k(E)=\varnothing$, if such a $k$ exists, and otherwise we set $L_{\m{C}}(E)=\infty$.
\item We define the \emph{clustering depth of $E$ at $z$} by
\begin{equation*}
L_{\m{C}}(E,z)\coleqq \lim_{\varepsilon\to 0} L_{\m{C}}(E\cap D^\circ(z,\varepsilon))=\inf_{\varepsilon> 0} L_{\m{C}}(E\cap D^\circ(z,\varepsilon)).
\end{equation*}
\end{enumerate}
\end{definition}

\noindent Note that the limit as $\varepsilon\to 0$ in the previous definition always exists, as $L_{\m{C}}(E\cap D^\circ(z,\varepsilon))$ is an increasing function of $\varepsilon$. The following lemma lists some of the properties of cluster sets and clustering depth.

\begin{lemma}\label{lem:cluster-prop}
Let $E,F\subset \C$ be sets, let $z\in\C$ be a point, and let $k$ be a nonnegative integer. Then
\begin{enumerate}[(i)]
\item $\m{C}^1(E)$ is closed,
\item the closure $\overline{E}$ of $E$ satisfies $\overline{E}=E\cup \m{C}^1(E)$,
\item if $L_{\m{C}}(E,z)\geq 1$, then either $z\in E$ or $z$ is a cluster point of $E$,
\item if $E\subset F$, then $\m{C}^k(E)\subset \m{C}^k(F)$,
\item $\m{C}^k(E\cup F)=\m{C}^k(E)\cup \m{C}^k(F)$, and
\item $L_{\m{C}}(E\cup F)=\max\{ L_{\m{C}}(E), L_{\m{C}}(F)\} $.
\end{enumerate}
\end{lemma}

 We are now ready to state the main result of this section, which strengthens Proposition \ref{prop:intro-vin-NNC}. 

\begin{prop}\label{prop:pole-structure}
Let $\m{N}=(V,E,\{\vin\},\Vout,\Omega,\Theta,\Lambda)$ be a non-trivial regular GFNN with 1-dimensional output and a singleton input set $\{\vin\}$. Then
\begin{enumerate}[(i)]
\item $\OnoA{\m{N}}{\sigma}$ can be analytically continued to a domain with countable complement in $\C$,
\item writing $ \dom_{\OnoA{\m{N}}{\sigma}} $ for the natural domain of $\OnoA{\m{N}}{\sigma}$ and $P_{\m{N}}\subset\C\setminus \dom_{\OnoA{\m{N}}{\sigma}} $ for its set of simple poles, we have $\C\setminus \dom_{\OnoA{\m{N}}{\sigma}}=\overline{P_{\m{N}}}$, and
\item $L_{\m{C}}( \overline{P_{\m{N}}} )= L(\m{N})$.
\end{enumerate}
\end{prop}

\noindent Note that this result immediately implies Proposition \ref{prop:intro-vin-NNC} since the depth of a non-trivial GFNN $\m{N}$ is at least one, and hence $L_{\m{C}}( \overline{P_{\m{N}}} )= L(\m{N})\geq 1$ implies that $P_{\m{N}}\neq\varnothing$.
We remark that statement (ii) of Proposition \ref{prop:pole-structure} is equivalent to the assertion that every essential singularity of $\OnoA{\m{N}}{\sigma}$ be the limit of a sequence of its simple poles.
The proof of Proposition \ref{prop:pole-structure} uses the following auxiliary results, whose proofs can be found in the Appendix.

\begin{lemma}\label{lem:nat-dom}
Let $f:\dom_f\to\C$ be a non-constant holomorphic function on its natural domain $\dom_f$ and suppose that $\C\setminus \dom_f$ is countable. Furthermore, let $g:\dom_g \to \C$ be a meromorphic function on $\C$ with a nonempty set of poles $P$. Then $g\circ f$ can be analytically continued to $\dom\coleqq \{z\in \dom_f: f(z)\in\C\setminus P\}$, and $\dom$ has countable complement in $\C$.
\end{lemma}

\begin{lemma}\label{lem:LD->reduc}
Let $\rho:\R\to\R$ be a nonlinearity, and let $\m{J}$ be a finite index set. Suppose that $\{(\alpha_s,\beta_s,\gamma_s)\}_{s\in\m{J}}$ are triples of real numbers such that $\sum_{s\in\m{J}}\alpha_s\rho(\beta_s\cdot\,+\, \gamma_s)$ is constant. Assume that $j^*\in \m{J}$ is such that $\alpha_{j^*}\neq 0$. Then there exist a set $\m{I}\subset \m{J}$ such that $j^*\in\m{I}$, and real $\{\wtd{\alpha}_s\}_{s\in \m{I}}$ such that $\wtd{\alpha}_{j^*}\neq 0$ and $\left(\zeta,\{(\wtd{\alpha}_s,\beta_s,\gamma_s)\}_{s\in\m{I}}\right)$ is an affine symmetry of $\rho$, for some $\zeta\in\R$.
\end{lemma}

\begin{proof}[Proof of Proposition \ref{prop:pole-structure}]
The proof follows the argument outlined in Section \ref{sec:ident-for-tanh}. We proceed by induction on $L(\m{N})$. To establish the base case, we assume that $L(\m{N})=1$,
and enumerate the nodes $V\setminus\{\vin\}$ as $\{u_1,\dots, u_{D_1}\}$. Now, as $\m{N}$ is non-degenerate, we have $\Vout=\{u_1,\dots, u_{D_1}\}$, and so we can write
\begin{equation*}
\OnoA{\m{N}}{\sigma}(z)=\lambda^{(1)}\;+ \sum_{j=1}^{D_1} \lambda^{(1)}_{u_j}\,\sigma(\omega_{u_j\vin}z\,+\, \theta_{u_j}), \quad \text{for }z\in \dom_{\OnoA{\m{N}}{\sigma}},
\end{equation*}
where $\lambda^{(1)}_{u_j}\neq 0$, for all $j\in\{1,\dots, D_1\}$.
Therefore, $\OnoA{\m{N}}{\sigma}$ is meromorphic on $\C$, and so statements (i) and (ii) hold immediately. To show statement (iii), note that $P_{\m{N}}$ is discrete (simply as $\OnoA{\m{N}}{\sigma}$ is meromorphic), and so $ \overline{P_{\m{N}}}= {P_{\m{N}}}$ and $L_{\m{C}}( {P_{\m{N}}} )\leq 1$. It therefore suffices to show that $P_\m{N}$ is nonempty, as we will then have $L_{\m{C}}( {P_{\m{N}}})\geq 1$.  Suppose by way of contradiction that $P_{\m{N}}$ is empty. Then, in particular, $P_{\m{N}}$ is bounded, and so the SAC for $\sigma$ implies that $\OnoA{\m{N}}{\sigma}$ is constant. Thus, $\sum_{j=1}^{D_1} \lambda^{(1)}_{u_j}\,\sigma(\omega_{u_j\vin}\cdot\,+\, \theta_{u_j})$ is constant, and hence, by Lemma \ref{lem:LD->reduc}, there exist a nonempty set $U\subset \{u_1,\dots, u_{D_1}\}$ and real numbers $\zeta$ and $\{\alpha_{u}\}_{u\in U}$ such that
$\left(\zeta,\{(\alpha_u,\omega_{u \vin},\theta_u)\}_{u\in U}\right)$ is an affine symmetry of $\sigma$. This implies that $\m{N}$ is $(\sigma,U)$--reducible, which stands in contradiction to the regularity of $\m{N}$, and thus establishes that $P_{\m{N}}$ is nonempty.

We proceed to the induction step. Suppose that $L(\m{N})\geq 2$ and assume that the claim of the proposition holds for all non-trivial regular GFNNs $\m{N}'$ with 1-dimensional output, input set $\{\vin\}$, and depth $L(\m{N}')<L(\m{N})$.
We can now write
\begin{equation}\label{eq:pole-clustering(-2)}
\OnoA{\m{N}}{\sigma}(z)=f(z) + \sum_{w\in \Vout^{>1}} \lambda^{(1)}_w\,  \sigma \big( \OnoA{\m{N}_w}{\sigma}(z)\big),\quad\text{for }z\in\dom_{\OnoA{\m{N}}{\sigma}},
\end{equation}
where $\m{N}_w$, for $w\in \Vout^{>1}\coleqq\{w\in \Vout:\lvl(w)>1\}$, is the subnetwork of $\m{N}$ generated by $\pre_{\m{N}}(w)$ with output set $\pre_{\m{N}}(w)$ and output scalars $\{\lambda^{(1),\, \m{N}_w} \coleqq \theta_{w}\}\cup\{\lambda^{(1),\, \m{N}_w}_v\coleqq \omega_{wv} : v\in\pre_{\m{N}}(w) \}$, and $f:\dom_f\to \C$, given by
\begin{equation*}
f(z)=\lambda^{(1)}\;+ \sum_{\substack{w\in \Vout \\ \lvl(w)=1 }} \lambda^{(1)}_{w}\,\sigma(\omega_{w\vin}z\,+\, \theta_{w}),
\end{equation*}
is a meromorphic function with simple poles only. Note that the $\m{N}_w$ are non-trivial regular GFNNs of depth $L(\m{N}_w)< L(\m{N})$.

For statement (i), we first observe that, for  $w\in  \Vout^{>1} $, the induction hypothesis for $\m{N}_w$ implies that $\OnoA{\m{N}_w}{\sigma}$ is non-constant and can be analytically continued to a domain with countable complement in $\C$. Thus, by Lemma \ref{lem:nat-dom}, we have that $\sigma \circ \OnoA{\m{N}_w}{\sigma}$ also analytically continues to a domain with countable complement in $\C$, and, in particular, its natural domain $ \dom_{\sigma \circ \OnoA{\m{N}_w}{\sigma}}$ is well-defined.
Next, note that $\OnoA{\m{N}}{\sigma}$ can be analytically continued to the set
\begin{equation*}
\dom\coleqq  \dom_f \cap \bigcap_{w\in  \Vout^{>1}} \dom_{\sigma \circ \OnoA{\m{N}_w}{\sigma}}.
\end{equation*} 
Then, as $f$ is meromorphic and $\C\setminus  \dom_{\sigma \circ \OnoA{\m{N}_w}{\sigma}}$ is countable for every $w\in  \Vout^{>1}$, we have that $\C\setminus\dom$ is countable, establishing statement (i) for $\m{N}$.  (Note that the natural domain $ \dom_{ \OnoA{\m{N}}{\sigma}}$ can be a strict superset of $\dom$, e.g., if there is a point in $\C$ that is a simple pole of $\sigma \circ \OnoA{\m{N}_{w_1}}{\sigma}$ and $\sigma \circ \OnoA{\m{N}_{w_2}}{\sigma}$ for distinct $w_1$ and $w_2$, their residues could be such that the pole disappears in the linear combination \eqref{eq:pole-clustering(-2)}).

For statement (ii), we begin by noting that, as $\C\setminus  \dom_{ \OnoA{\m{N}}{\sigma}} $ is countable, every element of $\C$ is a point of analyticity, a pole, or an essential singularity of ${\OnoA{\m{N}}{\sigma}}$, and we can thus write $\C\setminus \dom_{\OnoA{\m{N}}{\sigma}}= P_{\m{N}}\cup E_{\m{N}}$, where $P_{\m{N}}$ is the set of simple poles of  $\OnoA{\m{N}}{\sigma}$ and $E_{\m{N}}$ is the set of its essential singularities and poles of higher order.
Now, as  $\m{C}^1(P_{\m{N}})\subset E_{\m{N}}$, in order to complete the proof of statement (ii) for $\m{N}$, it suffices to establish that $E_{\m{N}}\subset \m{C}^1(P_{\m{N}})$. To this end, note that the induction hypothesis for $\m{N}_w$ implies that we can write
$
  \C\setminus\dom_{\OnoA{\m{N}_w}{\sigma}}=P_{w}\cup E_{w}
$,
where $\dom_{\OnoA{\m{N}_w}{\sigma}}$ is the natural domain of $\OnoA{\m{N}_w}{\sigma}$, $P_{w}$ is its set of simple poles, and $E_w=\overline{P_{w}}\setminus {P_{w}}=\m{C}^1({P_{w}})$ is the set of its essential singularities, for $w\in  \Vout^{>1} $.

Then, recalling \eqref{eq:pole-clustering(-2)} and the fact that $f$ and $\sigma$ are meromorphic with simple poles only, we have
\begin{equation*}
\C \bigm\backslash  \bigcup_{w\in  \Vout^{>1}} \overline{P_{w}}=\bigcap_{w\in  \Vout^{>1}}  \dom_{\OnoA{\m{N}_w}{\sigma}}\subset \C\setminus E_{\m{N}},
\end{equation*}
and thus $E_{\m{N}}\subset   \bigcup_{w\in  \Vout^{>1}} \overline{P_{w}}$. It will therefore be enough to show that
\begin{equation}\label{eq:pole-clustering(-1)}
\bigcup_{w\in  \Vout^{>1}} \overline{P_{w}}= \m{C}^1(P_{\m{N}}).
\end{equation}
To this end, first note that we immediately have $\m{C}^1(P_{\m{N}})\subset E_{\m{N}}\subset   \bigcup_{w\in  \Vout^{>1}} \overline{P_{w}}$.
For the reverse inclusion, we let $p\in \bigcup_{w\in  \Vout^{>1}} \overline{P_{w}}$, and distinguish between the cases $p\notin  \bigcup_{w\in \Vout^{>1} } E_w$ and $p\in \bigcup_{w\in \Vout^{>1} } E_w$.

\textit{The case $p\notin  \bigcup_{w\in \Vout^{>1} } E_w$. } 
Fix an arbitrary $w^*\in  \Vout^{>1}$ such that $p\in \overline{P_{w^*}}$ and set $\Vout^*=\{w\in \Vout^{>1}: p\in P_{w}\}$. Note that $w^*\in \Vout^*$, simply as $P_{w^*}=\overline{P_{w^*}}\setminus {E_{w^*}}$, by the induction hypothesis for $\m{N}_{w^*}$, and so $\Vout^*$ is nonempty. Now, for $w\in \Vout^*$, as $p$ is a simple pole of $\OnoA{\m{N}_w}{\sigma}$, we can write
\begin{equation}\label{eq:pole-clustering-3}
\OnoA{\m{N}_w}{\sigma}(z)=\frac{\beta_{w}}{z-p} + \gamma_{w} + \epsilon_{w}(z-p),
\end{equation}
for $z$ in an open neighborhood of $p$, where $\beta_{w}\in \C\setminus\{0\}$, $\gamma_{w}\in \C$, and $\epsilon_{w}:\dom_{\epsilon_{w}}\to \C$ is an ABC. Then, using \eqref{eq:pole-clustering-3} in \eqref{eq:pole-clustering(-2)} and performing the variable substitution $z'=\frac{1}{z-p}$ yields
\begin{equation}\label{eq:pole-clustering-4}
\OnoA{\m{N}}{\sigma}\left(p+1/z'\right)=(f+g)\left(p+1/z' \right)+ \sum_{w\in \Vout^*} \lambda^{(1)}_w\,  \sigma\left(\beta_{w} \,z' +\gamma_w +\epsilon_{w}(1/z') \right),
\end{equation}
for all $z'\in \C$ of sufficiently large modulus, where
\begin{equation*}
g(z)=\sum_{w\in \Vout^{>1}\setminus \Vout^*} \lambda^{(1)}_w\,  \sigma\big(\OnoA{\m{N}_w}{\sigma}(z)\big).
\end{equation*}
Now, due to the case assumption $p\notin  \bigcup_{w\in \Vout^{>1} } E_w$, we have that $\OnoA{\m{N}_w}{\sigma}$ is analytic at $p$, for all $w\in \Vout^{>1}\setminus \Vout^*$, and so $g$ is analytic on a punctured neighborhood of $p$. Thus, according to \eqref{eq:pole-clustering-4}, we will have $p\in \m{C}^1(P_{\m{N}})$, unless the set of poles of
\begin{equation}\label{eq:pole-clustering-5}
z'\; \longmapsto \;  \sum_{w\in \Vout^*} \lambda^{(1)}_w\,  \sigma\left(\beta_{w} \, z' +\gamma_w +\epsilon_{w}(1/z') \right)
\end{equation}
is bounded. Suppose by way of contradiction that the set of poles of \eqref{eq:pole-clustering-5} is bounded. Then, by the CAC for $\sigma$, there  exists a nonempty $U\subset \Vout^*$ such that $\beta_{w_1}^{-1}\epsilon_{w_1}=\beta_{w_2}^{-1}\,\epsilon_{w_2}$,   for all  $w_1,w_2\in U$, and the set of poles of
\begin{equation} \label{eq:pole-clustering-6}
f_U\coleqq \sum_{w\in U} \lambda^{(1)}_w\,  \sigma\left(\beta_{w} \, \cdot  +\gamma_w \right)
\end{equation}
is bounded. This and \eqref{eq:pole-clustering-3} together imply that
\begin{equation} \label{eq:pole-clustering-7}
\beta_{w_1}^{-1 }\OnoA{\m{N}_{w_1}}{\sigma}-\beta_{w_2}^{-1} \OnoA{\m{N}_{w_2}}{\sigma}=\left(\beta_{w_1}^{-1}\gamma_{w_1}-\beta_{w_2}^{-1}\gamma_{w_1}\right)\bm{1}
\end{equation}
is constant, for all $w_1,w_2\in U$.
We next establish the following claim.

\vspace{2mm}
\textit{Claim 1: Writing $Y_w= \pre_{\m{N}}(w) $, for $w\in U$, we have
\begin{equation}\label{eq:pole-clustering-8}
\begin{aligned}
Y\coleqq  Y_{w_1}&=Y_{w_2} \quad\text{and }\\
\beta_{w_1} ^{-1}\{\omega_{w_1v}\}_{v\in Y}&=\beta_{w_2}^{-1}\{\omega_{w_2v}\}_{v\in Y},
\end{aligned}
\end{equation} 
for all $w_1,w_2\in U$, and there exists a $\vartheta\in\R$ such that $\beta_we^{-i\vartheta}\in\R$, for all $w\in U$.
}
\vspace{2mm}

\noindent\textit{Proof of Claim 1.}  We argue by contradiction, so suppose that the claim is false. Then there exist distinct $w_1,w_2\in U$ such that either $Y_{w_1}\neq Y_{w_2}$ or
\begin{equation*}
 Y\coleqq Y_{w_1}=Y_{w_2}\quad\text{and }\quad \beta_{w_1}^{-1}\{\omega_{w_1v}\}_{v\in Y }\neq \beta_{w_2}^{-1}\{\omega_{w_2v}\}_{v\in  Y}.
 \end{equation*} 
Next, let
\begin{equation*}
\begin{aligned}
Z_{1}^{\, \Re}&=\{v\in Y_{w_1}\setminus Y_{w_2}: \Re( \beta_{w_1}^{-1} \omega_{w_1v})\neq 0 \},\\
Z_{2}^{\, \Re}&=\{v\in Y_{w_1}\cap  Y_{w_2}: \Re( \beta_{w_1}^{-1} \omega_{w_1v}- \beta_{w_2}^{-1} \omega_{w_2v})\neq 0 \}, \\
Z_{3}^{\, \Re}&=\{v\in Y_{w_2}\setminus Y_{w_1}: \Re( \beta_{w_2}^{-1} \omega_{w_2v})\neq 0 \},\quad \text{and }\\
S^{\, \Re}&=Z_{1}^{\, \Re}\cup Z_{2}^{\, \Re}\cup Z_{3}^{\, \Re},
\end{aligned}
\end{equation*}
and define the sets $Z_{1}^{\, \Im}$, $Z_{2}^{\, \Im}$, $Z_{3}^{\, \Im}$, and $S^{\, \Im}$ analogously. Then, by our assumption, at least one of $S^{\, \Re}$ and $S^{\, \Im}$ must be nonempty. Suppose for now that $S^{\, \Re}\neq \varnothing$. Next, set
\begin{equation*}
{\lambda}_{v}^{(1),\,\m{N}'}=
\begin{cases}
\beta_{w_1}^{-1}\omega_{w_1v}, &\text{if }v\in Z_1^{\,\Re}\\
\beta_{w_1}^{-1}\omega_{w_1v}-\beta_{w_2}^{-1}\omega_{w_2v}, &\text{if }v\in Z_2^{\,\Re}\\
-\beta_{w_2}^{-1}\omega_{w_2v}, &\text{if }v\in Z_3^{\,\Re}
\end{cases},\quad \text{ for }v\in S^{\,\Re},
\end{equation*}
and
\begin{equation}
{\Lambda}^{\m{N}'}=\{{\lambda}^{(1),\,\m{N}'}\coleqq \Re( \beta_{w_1}^{-1}\theta_{w_1}-  \beta_{w_2}^{-1}\theta_{w_2})  \}\cup\{\lambda_{v}^{(1),\,\m{N}'}:w\in S^{\,\Im}\},
\end{equation}
and define $\m{N}'=(V^{\m{N}'},E^{\m{N}'},\{\vin\},S^{\Re},\Omega^{\m{N}'},\Theta^{\m{N}'},\Lambda^{\m{N}'})$ to be the subnetwork of $\m{N}$ with one-dimensional output generated by $S$. Then $\m{N}'$ is a regular GFNN of depth $L(\m{N}')<L(\m{N})$, and, as $\lvl_{\m{N}}(w)>1$, for $w\in U\subset \Vout^{>1}$, we have that $\m{N}'$ is non-trivial. It hence follows by the induction hypothesis for $\m{N}'$ that the set $P_{\m{N}'}$ of poles of $\m{N}'$ satisfies $L_{\m{C}}(\overline {P_{\m{N}'}})=L(\m{N}')\geq 1$. In particular, we have $P_{\m{N}'}\neq \varnothing$. On the other hand,
 \begin{equation}\label{eq:pole-clustering-9}
\begin{aligned}
\OnoA{\m{N}'}{\sigma} &=\Re \left(\beta_{w_1}^{-1}\theta_{w_1}-\beta_{w_2}^{-1}\theta_{w_2} \right) +\sum_{v\in Z_1^{\, \Re} } \Re\left( \beta_{w_1}^{-1} \omega_{w_1v} \right) \outmapT{v}{\sigma}{\m{N}}\\
&\hspace{2cm}+\sum_{v\in Z_2^{\,\Re }} \Re \left( \beta_{w_1}^{-1} \omega_{w_1v} - \beta_{w_2}^{-1} \omega_{w_2v} \right)\outmapT{v}{\sigma}{\m{N}}- \sum_{v\in Z_3^{\,\Re} } \Re\left(\beta_{w_2}^{-1} \omega_{w_2v} \right) \outmapT{v}{\sigma}{\m{N}}\\
 &=\Re \Big(\beta_{w_1}^{-1}\theta_{w_1}-\beta_{w_2}^{-1}\theta_{w_2} +\sum_{v\in Y_{w_1} \setminus Y_{w_2} } \beta_{w_1}^{-1} \omega_{w_1v} \outmapT{v}{\sigma}{\m{N}}\\
&\hspace{2cm}+\sum_{v\in Y_{w_1} \cap Y_{w_2} } \hspace{-3mm} \left( \beta_{w_1}^{-1} \omega_{w_1v} - \beta_{w_2}^{-1} \omega_{w_2v} \right)\outmapT{v}{\sigma}{\m{N}}- \hspace{-3mm}\sum_{v\in Y_{w_2} \setminus Y_{w_1} } \beta_{w_2}^{-1} \omega_{w_2v} \outmapT{v}{\sigma}{\m{N}} \Big)\\
&=\Re\left(\beta_{w_1}^{-1 }\OnoA{\m{N}_{w_1}}{\sigma}-\beta_{w_2}^{-1} \OnoA{\m{N}_{w_2}}{\sigma} \right)= \Re\left(\beta_{w_1}^{-1}\gamma_{w_1}-\beta_{w_2}^{-1}\gamma_{w_1}\right) \bm{1},
\end{aligned}
\end{equation}
showing that $\OnoA{\m{N}'}{\sigma}$ is constant, which stands in contradiction to $P_{\m{N}'}\neq \varnothing$. An entirely analogous argument leads to a contradiction in the case $S^{\, \Re}= \varnothing$ and $S^{\, \Im}\neq \varnothing$, establishing that \eqref{eq:pole-clustering-8} must hold. Now, as $\beta_{w_1}/\beta_{w_2}=\omega_{w_1v}/\omega_{w_2v}\in \R$, for all $w_1,w_2\in U$ and $v\in Y$, the $\beta_{w}$ all have the same complex argument, and so there must exist a $\vartheta\in\R$ such that $\beta_we^{-i\vartheta}\in\R$, for all $w\in U$, completing the proof of Claim 1.

Now, it follows from the decomposition \eqref{eq:pole-clustering(-2)} that the constant output scalar of every $\m{N}_{w}$ is $\theta_w$, which together with \eqref{eq:pole-clustering-8} implies that 
\begin{equation*}
\beta_{w_1}^{-1 }\OnoA{\m{N}_{w_1}}{\sigma}-\beta_{w_2}^{-1} \OnoA{\m{N}_{w_2}}{\sigma}= \beta_{w_1}^{-1}\theta_{w_1}-  \beta_{w_2}^{-1}\theta_{w_2}.
\end{equation*}
This and \eqref{eq:pole-clustering-7} together give
\begin{equation*}
\beta_{w_1}^{-1}(\theta_{w_1}-\gamma_{w_1})=\beta_{w_2}^{-1}(\theta_{w_2}-\gamma_{w_2}),
\end{equation*}
for all $w_1,w_2\in U$, implying the existence of a $c\in\C$ such that $\beta_{w}^{-1}(\theta_{w}-\gamma_{w})=c$, for all $w\in U$. Hence, recalling \eqref{eq:pole-clustering-6}, we have that the set of poles of
\begin{equation*}
\sum_{w\in U} \lambda^{(1)}_w\,  \sigma\big(\beta_{w} e^{-i\vartheta} \, \cdot  +\,\theta_w \big)= \sum_{w\in U} \lambda^{(1)}_w\,  \sigma\big(\beta_{w} e^{-i\vartheta} \, \cdot  +\, c\beta_w +\gamma_w \big)=f_U(e^{-i\vartheta}\,\cdot\, + c)
\end{equation*}
is bounded, and so, as $\beta_we^{-i\vartheta}\in\R$, for all $w\in U$, the SAC for $\sigma$ implies that $\sum_{w\in U} \lambda^{(1)}_w\,  \sigma\big(\beta_{w} e^{-i\vartheta} \, \cdot  +\,\theta_w \big)$ must be constant. Now, Lemma 
\ref{lem:LD->reduc} establishes the existence of a nonempty $U'\subset U$ and real numbers $\zeta$ and $\{\alpha_w\}_{w\in U'}$ such that $\left(\zeta,\{({\alpha}_w,\beta_we^{-i\vartheta},\theta_w)\}_{w\in U'}\right)$ is a symmetry of $\sigma$. On the other hand, Claim 1 implies that the nodes $U'$ have a common parent set $Y$ in $\m{N}$ and that there exist nonzero real numbers $\{\kappa_v\}_{v\in Y }$ such that $\{\omega_{wv}\}_{v\in Y}=\beta_we^{-i\vartheta} \{\kappa_v\}_{v\in Y}$, for all $w\in U'$, therefore implying that $\m{N}$ is $(\sigma,U')$--reducible. This, however, contradicts the assumption that $\m{N}$ is regular and thereby establishes that $p\in \m{C}^1(P_{\m{N}})$ in the case $p\notin  \bigcup_{w\in \Vout^{>1} } E_w$.

\textit{The case $p\in  \bigcup_{w\in \Vout^{>1} } E_w$. }
Define the sets ${P}_u^\circ\coleqq P_{u}\bigm\backslash   \bigcup_{w\in \Vout^{>1} } E_w$, for $u\in \Vout^{>1}$.  Then every element of ${P}_u^\circ$ is a cluster point of $P_{\m{N}}$ (by the case already established), for every $u\in \Vout^{>1}$, and thus $p$ itself will be a cluster point of $P_{\m{N}}$, provided we can establish the existence of a $u^*\in  \Vout^{>1}$ such that $p\in \m{C}^1( {P}_{u^*}^\circ )$. This will be an immediate consequence of the following claim.

\vspace{2mm}
\textit{Claim 2: We have $\bigcup_{w\in \Vout^{>1} } E_w= \bigcup_{u\in \Vout^{>1} }\m{C}^1( {P}_u^\circ )$.
}
\vspace{2mm}

\noindent\textit{Proof of Claim 2.} For every $u\in \Vout^{>1}$, we have ${P}_u^\circ\subset  P_{u}$, and so
\begin{equation*}
\m{C}^1(P_u^\circ)\subset \m{C}^1({P_{u}})= E_u\subset \bigcup_{w\in \Vout^{>1} } E_w,
\end{equation*}
implying that $\bigcup_{u\in \Vout^{>1} } \m{C}^1\left( {P}_u^\circ \right)\subset \bigcup_{w\in \Vout^{>1} } E_w$. For the reverse inclusion, we suppose by way of contradiction that there exists a point $y\in \bigcup_{w\in \Vout^{>1} } E_w  \bigm\backslash \bigcup_{u\in \Vout^{>1} }\m{C}^1\left( {P}_u^\circ \right)$. Now, for every $w\in \Vout^{>1} $, we have that 
\begin{equation*}
L_{\m{C}}(E_w,y)\leq L_{\m{C}}(E_w)\leq L_{\m{C}}(\overline{P_{w}})=L(\m{N}_w)<\infty,
\end{equation*}
 by statement (iii) for $\m{N}_w$, and so 
 \begin{equation}\label{eq:pole-clustering-0}
 k\coleqq \max_{\substack{w\in \Vout^{>1} \\ y\in E_w}} L_{\m{C}}(E_w,y)<\infty.
 \end{equation}
 Let $w^*\in \Vout^{>1}$ be such that $y\in E_{w^*}$ and $L_{\m{C}}(E_{w^*},y)=k$. Next, as $y$ is not an element of 
$
 \bigcup_{u\in \Vout^{>1} } \m{C}^1\left( {P}_u^\circ \right)
$, it is not a cluster point of ${P}_{w^*}^\circ$, and so
there exists an $\varepsilon>0$ such that ${P}_{w^*}^\circ \cap D^\circ(y,\varepsilon)=\varnothing$. Then, by definition of $P_{w^*}^\circ$, we have
\begin{equation*}
P_{w^*}\cap D^\circ(y,\delta)\bigm\backslash    \bigcup_{w\in \Vout^{>1} } E_w\cap D^\circ(y,\delta)\subset  \Big( P_{w^*}\bigm\backslash    \bigcup_{w\in \Vout^{>1} } E_w \Big)\cap D^\circ(y,\varepsilon) = {P}_{w^*}^\circ\cap  D^\circ(y,\varepsilon) =\varnothing,
\end{equation*}
for every $\delta\in(0,\varepsilon)$, and thus, using item (vi) of Lemma \ref{lem:cluster-prop}, we get 
\begin{equation}\label{eq:pole-clustering-1}
\begin{aligned}
\max_{w\in \Vout^{>1} } \Big\{ L_{\m{C}}(E_w \cap D^\circ(y,\delta)) \Big\}&= L_{\m{C}}\Big( \bigcup_{w\in \Vout^{>1} } E_w \cap D^\circ(y,\delta)  \Big) \geq L_{\m{C}}(P_{w^*}\cap D^\circ(y,\delta) ).
\end{aligned}
\end{equation}
On the other hand, as $E_{w^*}=\m{C}^1 ( P_{w^*})$, we have
$
\m{C}^1 ( P_{w^*} \cap D^\circ(y,\delta))= E_{w^*} \cap D(y,\delta)\neq\varnothing
$, and so
\begin{equation}\label{eq:pole-clustering-2}
L_{\m{C}}(P_{w^*} \cap D^\circ(y,\delta))= L_{\m{C}} (E_{w^*} \cap D(y,\delta)) +1 \geq L_{\m{C}}(E_{w^*},y)+1=k+1.
\end{equation}
Now, \eqref{eq:pole-clustering-1} and \eqref{eq:pole-clustering-2} together yield
\begin{equation*}
\max_{w\in \Vout^{>1} } \Big\{ L_{\m{C}}(E_w ,y ) \Big\}=\lim_{\delta \to 0} \max_{w\in \Vout^{>1} } \Big\{ L_{\m{C}}(E_w \cap D^\circ(y,\delta)) \Big\}\geq k+ 1,
\end{equation*}
and so there must exist a $w'\in \Vout^{>1}$ such that $L_{\m{C}}(E_{w'} ,y )\geq k+1$. Thus, by item (iii) of Lemma \ref{lem:cluster-prop} and the fact that $E_{w'}$ is closed (which follows from item (i) of the same lemma and $E_{w'}=\m{C}^1(P_{w'}) $), we must have $y\in E_{w'}$. But now
\begin{equation*}
\max_{\substack{w\in \Vout^{>1} \\ y\in E_w}} L_{\m{C}}(E_w,y)\geq L_{\m{C}}(E_{w'},y)\geq k+1,
\end{equation*}
which contradicts \eqref{eq:pole-clustering-0} and thus concludes the proof of Claim 2.

We have thus established that $\bigcup_{w\in  \Vout^{>1}} \overline{P_{w}}\subset \m{C}^1(P_{\m{N}})$, completing the proof of \eqref{eq:pole-clustering(-1)} and thereby proving statement (ii) for $\m{N}$.

In order to establish statement (iii), we use \eqref{eq:pole-clustering(-1)} together with item (vi) of Lemma \ref{lem:cluster-prop} and the induction hypothesis to argue as follows:
\begin{equation*}
L_{\m{C}}\left( \m{C}^1(P_{\m{N}}) \right)=L_{\m{C}}\Big( \bigcup_{w\in  \Vout^{>1}} \overline{P_{w}}\;  \Big)=\max_{w\in  \Vout^{>1}}  L_{\m{C}}\left(\overline{P_{w}} \right)=\max_{w\in  \Vout^{>1}}  L(\m{N}_w)=L(\m{N})-1.
\end{equation*}
This, in particular, implies that $\m{C}^1(P_{\m{N}})$ is nonempty (as $L(\m{N})-1\geq 1$), and hence
\begin{equation*}
L_{\m{C}}\left(\overline{P_{\m{N}}} \right)=L_{\m{C}}\left(\m{C}^1( \overline{P_{\m{N}}}) \right)+1=L_{\m{C}}\left(\m{C}^1( {P_{\m{N}}}) \right)+1= L(\m{N}),
\end{equation*}
where we also used $\m{C}^1(\overline{P_{\m{N}}})=\m{C}^1(P_{\m{N}})$. This concludes the proof of the proposition.
\end{proof}

\section{Input anchoring and the proof of Theorem \ref{thm:intro-Zab-NNC}}\label{sec:inp-anch}

\subsection{Input anchoring}

In this section, we introduce the procedure of input anchoring, which will allow us to extend the null-net property for meromorphic nonlinearities on a singleton input set to input sets of arbitrary size.
This procedure was first introduced in \cite{Vlacic2019} for networks satisfying the so-called no-clones condition, which constitutes a special case of irreducibility for nonlinearities with no affine symmetries other than the trivial ones. We now generalize this method to arbitrary nonlinearities $\sigma$ satisfying the SAC. This involves finding a precise ``topological description'' of the set of affine symmetries of $\sigma$ (in the sense of Lemma \ref{lem:orderly} below), as well as applying the Baire category theorem.

Before further discussing input anchoring, we address the case of regular GFNNs having input nodes without any outgoing edges (which is allowed by Definition \ref{def:NonDeg}). Concretely, suppose that $\m{M}=(V,E,\Vin, \Vout, \Omega, \Theta,\Lambda)$ is a non-trivial regular GFNN with one-dimensional output such that $\OnoA{\m{M}}{\rho}=0$.
Then, writing $\Vin^{0}$ for the set of input nodes of $\m{M}$ without any outgoing edges, we have $\Vin^{0}\subsetneq \Vin$, as $\m{M}$ is non-trivial. Therefore, we can define a non-trivial regular GFNN $\m{M}'= (V' ,E,V'_{\mrm{in}}, \Vout, \Omega, \Theta,\Lambda)$ with one-dimensional output, obtained from $\m{M}$ by deleting the nodes $\Vin^{0}$. This network also satisfies $\OnoA{\m{M}'}{\rho}=0$, as well as $V'= \anc(\Vout)$, which can be viewed as a stronger version of Property (i) of Definition \ref{def:NonDeg}. Thus, we can henceforth work w.l.o.g. with networks satisfying the following strong regularity condition.

\begin{definition}[Strong non-degeneracy and strong regularity]
Let $\m{M}=(V,E,\Vin, \Vout, \Omega, \Theta,\Lambda)$ be a GFNN. We say that $\m{M}$ is \emph{strongly non-degenerate} if it is non-degenerate and $V= \anc(\Vout)$. We call  $\m{M}$ \emph{strongly regular} if it is strongly non-degenerate and irreducible.
\end{definition}

Now, let $\m{M}=(V^\m{M},E^\m{M},\Vin^{\m{M}},\Vout^\m{M},\allowbreak\Omega^{\m{M}},\allowbreak\Theta^{\m{M}},\Lambda^{\m{M}})$ be a strongly regular GFNN with one-dimensional output identically equal to zero. Enumerate the input nodes according to $\Vin^\m{M}=\{v_1^0,\dots, \allowbreak v_{D_0}^0\}$, and suppose that $D_0\geq 2$.
Let $a\in\R$ and let $\rho$ be a nonlinearity. We seek to construct a non-trivial GFNN ${\m{M}}_a=(V^{{\m{M}_a}},E^{{\m{M}_a}},\Vin^{{\m{M}_a}},\Vout^{{\m{M}_a}},\allowbreak \Omega^{{\m{M}_a}}, \Theta^{{\m{M}_a}},\Lambda^{\m{M}_a})$ with one-dimensional output, input set $\Vin^{{\m{M}_a}}= \{v_1^0,\dots, v_{D_0-1}^0\}$ satisfying $\Vout^{{\m{M}_a}}=\Vout^{\m{M}}\cap V^{{\m{M}_a}}$, and the following two properties:
\begin{itemize}[--]
\item[(IA-1)] For all $w\in \Vout^{{\m{M}_a}}$,
\begin{equation*}
\outmapT{w}{\rho}{\m{M}_a}\!\left(t_1,t_2,\dots,t_{D_0-1}\right)=\outmapT{w}{\rho}{\m{M}}\!\left(t_1,t_2,\dots,t_{D_0-1},a\right),
\end{equation*}
 for all $(t_1,t_2,\dots,t_{D_0-1})\in\R^{D_0-1}$ (after identifying $\R^{\Vin}$ with $\R^{D_0}$).
\item[(IA-2)]\label{it:InpFixMainItem2} For all $w\in \Vout^{\m{M}}\setminus \Vout^{{\m{M}_a}}$, the function $\R^{D_0-1}\to \R$ given by
\begin{equation*}
(t_1,t_2,\dots,t_{D_0-1})\mapsto \outmapT{w}{\rho}{\m{M}}\!\left(t_1,t_2,\dots,t_{D_0-1},a\right)
\end{equation*}
is constant, and we denote its value by $\outmapT{w}{\rho}{\m{M}}\!\left(a\right)$.
\end{itemize}
As $\Vin^{{\m{M}_a}}= V^{{\m{M}}}_{\mrm{in}}\setminus\{v_{D_0}^0\}$, the network $\m{M}_a$ will, indeed, have fewer input nodes than $\m{M}$.

Suppose now that $\m{M}_a$ is such a network. Then, as $\m{M}$ is assumed to have identically zero output, we have
\begin{equation*}
\lambda^{(1),\,\m{M}}+\sum_{w\in \Vout^\m{M}}\lambda_{w}^{(1),\,\m{M}} \outmapT{w}{\rho}{\m{M}}(t_1,t_2,\dots,t_{D_0-1},a)=0,
\end{equation*}
for all $(t_1,t_2,\dots,t_{D_0-1})\in\R^{D_0-1}$, where $\lambda_w^{(1),\,\m{M}}\neq 0$, for all $w\in \Vout^{\m{M}}$, by non-degeneracy of $\m{M}$.
This can be rewritten as 
\begin{equation*}
\left(\lambda^{(1),\,\m{M}}+\sum_{w\in \Vout^{\m{M}}\setminus \Vout^{\m{M}_a}}\lambda_{w}^{(1),\,\m{M}}\outmapT{w}{\rho}{\m{M}}\!(a)\right)\bm{1}+\sum_{w\in \Vout^{\m{M}_a}}\lambda_{w}^{(1),\,\m{M}} \outmapT{w}{\rho}{\m{M}_a} =0,
\end{equation*}
and thus the output scalars $\Lambda^{\m{M}_a}$ can be chosen so that the output of $\m{M}_a$ is identically zero.

In the following definition, we provide the desired network $\m{M}_a$, and we refer the reader to Figure \ref{fig:anchoring-g-intro} in Section \ref{sec:ident-for-tanh} for an illustration of this construction.

\begin{definition}\label{def:InputFix}
Let $\m{M}=(V^{\m{M}},E^{\m{M}},\Vin^{\m{M}},\Vout^\m{M},\Omega^{\m{M}},\allowbreak\Theta^{\m{M}},\Lambda^{\m{M}})$ be a strongly regular GFNN with one-dimensional output, and input nodes $\Vin^\m{M}=\{v_1^0,\dots, v_{D_0}^0\}$, $D_0\geq 2$. Let $a\in\R$, and let $\rho$ be a nonlinearity such that $\OnoA{\m{M}}{\rho}= 0$. 
The \emph{network obtained from $\m{M}$ by anchoring the input $v_{D_0}^0$ to $a$ with respect to $\rho$} is the GFNN
${\m{M}}_a=(V^{{\m{M}_a}},E^{{\m{M}_a}},\Vin^{{\m{M}_a}},\Vout^{{\m{M}_a}},\Omega^{{\m{M}_a}},\allowbreak \Theta^{{\m{M}_a}},\Lambda^{\m{M}_a})$ given by the following:
\begin{itemize}[--]
\item $V^{{\m{M}_a}}=\{v\in V^{\m{M}}:\{v_1^0,\dots, v_{D_0-1}^0\}\cap \anc_{\m{M}}(\{v\})\neq\varnothing\}$,
\item $E^{{\m{M}_a}}=\{(v,\wtd{v}): v,\wtd{v}\in V^{{\m{M}_a}} \}$,
\item $\Vin^{{\m{M}_a}}= \{v_1^0,\dots, v_{D_0-1}^0\}$ and $\Vout^{{\m{M}_a}}=\Vout^{\m{M}}\cap V^{{\m{M}_a}}$,
\item $\Omega^{{\m{M}_a}}=\{\omega_{\wtd{v}v}: (v,\wtd{v})\in E^{{\m{M}_a}}\}$.
\item For a node $v\in V^{\m{M}}\setminus V^{{\m{M}_a}}$,
we define recursively
\begin{equation}\label{eq:InpFixBiasMod}
a_v=\begin{cases}
a, & v=v_{D_0}^0\\
\rho\left({\textstyle \sum_{u\in\pre_{\m{M}}(v) }}\omega_{vu}a_u+\theta_v\right), & v\neq v_{D_0}^0
\end{cases}.
\end{equation}
(Note that this is well-defined, as $\pre_{\m{M}}(v)\subset V^{\m{M}} \setminus V^{{\m{M}_a}}$ whenever $v\in V^{\m{M}}\setminus V^{{\m{M}_a}}$.)
Now, for $u\in V^{{\m{M}_a}}$, let
\begin{equation}\label{eq:InpFixBiasModTilde}
 \wtd{\theta}_{u}=\theta_u+\sum_{v\in\pre_{\m{M}}(u)\setminus V^{{\m{M}_a}}}\omega_{uv}a_v,
\end{equation}
and set $\Theta^{{\m{M}_a}}=\{\wtd{\theta}_u:u\in V^{{\m{M}_a}}\}$.
\item Set ${\lambda}^{(1),\,\m{M}_a}=\lambda^{(1),\,\m{M}}+\sum_{w\in \Vout^{\m{M}}\setminus \Vout^{\m{M}_a}}\lambda_{w}^{(1),\,\m{M}} \,a_w$ and $\Lambda^{\m{M}_a}=\{{\lambda}^{(1),\,\m{M}_a}\}\cup\{\lambda_w^{(1),\,\m{M}}:w\in \Vout^{\m{M}_a}\}$.
\end{itemize}
\end{definition}

The network ${\m{M}_a}$ satisfies  (IA-1) and (IA-2) by construction, and therefore $\OnoA{\m{M}_a}{\rho}= 0$ by the choice of the output scalars of $\m{M}_a$.
Moreover, ${\m{M}_a}$ is strongly non-degenerate. To see this, take an arbitrary $v\in V^{\m{M}_a}$. Then, by strong non-degeneracy of $\m{M}$, there exists a $w\in V^{\m{M}}_{\mrm{out}}$ such that $v\in \anc_{\m{M}}(w)$. As $w$ is connected directly with a node in $V^{\m{M}_a}$, it follows that $w\in V^{{\m{M}_a}}$, and so $w\in \Vout^{\m{M}_a}$.
Therefore, $v\in \anc_{\m{M}_a}(w)$, and, as $v$ was arbitrary, we obtain $V^{\m{M}_a}\subset \bigcup_{w\in \Vout^{\m{M}_a} }\anc_{\m{M}_a}(w)=\anc_{\m{M}_a}(\Vout^{\m{M}_a} )$. On the other hand, Property (ii) of Definition \ref{def:NonDeg} follows from $\Vout^{\m{M}_a}\subset \Vout^{\m{M}}$ and the fact that $\m{M}_a$ inherits the output scalars from $\m{M}$. This establishes that ${\m{M}_a}$ is strongly non-degenerate. Finally, if $\m{M}$ is layered, then so is ${\m{M}_a}$.

However, ${\m{M}_a}$ is not, in general, guaranteed to be irreducible. Consider, for instance, the network $\m{M}$ in Figure \ref{fig:anchoring-g-intro}. As the biases of the nodes $u_1,u_2,w_2,w_3$ are changed, the network $\m{M}_a$ may be $(\rho,\{u_1,u_2\})$--reducible or $(\rho,\{w_2,w_3\})$--reducible, or both.  This is unfortunate, as our program for proving Theorem \ref{thm:intro-Zab-NNC} envisages maintaining regularity when constructing networks with zero output. However, this nuisance can be circumvented, as the following lemma says that, for real meromorphic nonlinearities satisfying the SAC, either there exists some value of $a\in\R$ such that the network ${\m{M}_a}$ is, indeed, irreducible, or else it is possible to select a strongly regular subnetwork $\m{N}$ of $\m{M}$ with input $\{v_{D_0}^0\}$ and identically zero output. This will be sufficient for our purposes.

\begin{prop}[Input anchoring]\label{prop:inp-anch}
Let $\m{M}=(V^{\m{M}},E^{\m{M}},\Vin^{\m{M}},\Vout^\m{M},\Omega^{\m{M}},\Theta^{\m{M}},\Lambda^{\m{M}})$ be a strongly regular GFNN with one-dimensional output and input nodes $\Vin^\m{M}=\{v_1^0,\dots, v_{D_0}^0\}$, $D_0\geq 2$. Let $\sigma$ be a nonlinearity such that $\sigma(\R)\subset\R$, and suppose that $\sigma$ is meromorphic on $\C$ and satisfies the SAC. Finally, suppose that $\OnoA{\m{M}}{\sigma}= 0$, and let ${\m{M}}_a$ denote the network obtained by anchoring the input $v_{D_0}^0$ to some $a\in\R$ with respect to $\sigma$, according to Definition \ref{def:InputFix}. Then one of the following two statements must be true:
\begin{enumerate}[(i)]
\item There exists an $a\in\R$ such that ${\m{M}}_a$ is strongly regular.
\item There exists a strongly regular subnetwork $\m{N}=(V^{\m{N}},E^{\m{N}},\{v_{D_0}^0\},\Vout^\m{N},\Omega^{\m{N}},\Theta^{\m{N}},\Lambda^{\m{N}})$ of $\m{M}$
with one-dimensional output such that $\OnoA{\m{N}}{\sigma}= 0$.
\end{enumerate}
\end{prop}

The proof of Proposition \ref{prop:inp-anch} requires the following auxiliary result, whose proof can be found in the Appendix.
\begin{lemma}\label{lem:orderly}
Let $\sigma$ be a meromorphic nonlinearity on $\C$ satisfying the SAC. Furthermore, let $\{\beta_s\}_{s\in\m{I}}$ be a nonempty finite set of nonzero real numbers. Then
\begin{equation*}
\Gamma\coleqq\left\{(\gamma_s)_{s\in\m{I}}\in\R^{\m{I}}:\quad \text{\parbox{8cm}{ $\left(\zeta,\{(\alpha_s,\beta_s,\gamma_s)\}_{s\in\m{I}}\right)$ is an affine symmetry of $\sigma$ for some $\zeta \in \R$ and nonzero real numbers $\{\alpha_s\}_{s\in\m{I}}$ }}\quad \right\}
\end{equation*}
is a (possibly empty) countable union of parallel lines in $\R^{\m{I}}$. More specifically, there exists a countable set $\Gamma' \subset \R^{\m I}$ such that $\Gamma=\bigcup_{\gamma'\in\Gamma'}\{(\gamma'_s+ t\beta_s)_{s\in\m{I}}:t\in\R\}$.
\end{lemma}

\begin{proof}[Proof of Proposition \ref{prop:inp-anch}]
For a subset $U$ of nodes of $\m{M}$ define
\begin{equation*}
E_{U}=\{a\in\R: U\subset V^{\m{M}_a}\text{ and $\m{M}_a$ is $(\sigma,U)$--reducible}\}.
\end{equation*}
Suppose that statement {(i)} is false, so that, for every $a\in\R$, there exists a $U\subset V^{\m{M}}$ such that $a\in E_{U}$. We can then write $\R$ as a finite union
\begin{equation*}
\R=\bigcup_{U\subset V^{\m{M}}}E_{U},
\end{equation*}
and, as $\R$ is a complete metric space and the union over the subsets of $V^{\m{M}}$ is finite, it follows by the Baire category theorem \cite[Thm. 5.6]{Rudin1987} that there exists a $U\subset V^{\m{M}}$ such that $E_{U}$ is not meagre in $\R$, i.e., it is not a countable union of nowhere dense sets. Fix such a set $U$, let $P$ be the common parent set in $\m{M}_a$ of the nodes in $U$, let $\{\kappa_v\}_{v\in P}$ and $\{\beta_u\}_{u\in U}$ be such that $\{\omega_{uv}\}_{v\in P}=\beta_u\{\kappa_v\}_{v\in P}$, for all $u\in U$, and set $P'=\bigcup_{u\in U}(\pre_{\m{M}}(u)\setminus P)$. Note that, for $v\in P'$, the  map $\OTnoA{v}{\sigma}{\m{M}}$ depends on $v_{D_0}^0$, but not on the remaining input nodes $\{v_1^0,\dots,v_{D_0-1}^0\}$ of $\m{M}$, so we can write $\OTnoA{v}{\sigma}{\m{M}}(a)$ for the value of $\OTnoA{v}{\sigma}{\m{M}}$ at an arbitrary point $(t_1,\dots,t_{D_0-1},a)\in\R^{\Vin}$. Now, the bias of every $u\in U$ in $\m{M}_a$ is given by
\begin{equation*}
{\xi}_u(a)\coleqq \theta_u+\sum_{v\in P'\cap\pre_{\m{M}}(u)}\omega_{uv}\OTnoA{v}{\sigma}{\m{M}}(a).
\end{equation*}
As $\sigma$ is a meromorphic function satisfying the SAC and $\sigma(\R)\subset \R$, we know by Lemma \ref{lem:orderly} that the set
\begin{equation*}
\Gamma\coleqq\left\{(\gamma_u)_{u\in U}\in\R^{U}:\quad \text{\parbox{8.5cm}{ $\left(\zeta,\{(\alpha_u,\beta_u,\gamma_u)\}_{u\in U}\right)$ is an affine symmetry of $\sigma$ for some $\zeta \in \R$ and nonzero real numbers $\{\alpha_u\}_{u\in U}$ }}\quad \right\}
\end{equation*}
is a countable union of parallel lines in $\R^{U}$, i.e., there exists a countable set $\Gamma' \subset \R^{U}$ such that $\Gamma=\bigcup_{\gamma'\in\Gamma'}\Gamma_{\gamma'}$, where $\Gamma_{\gamma'}\coleqq\{(\gamma'_u+ t\beta_u)_{u\in U}:t\in\R\}$ are w.l.o.g. pairwise disjoint.  
Note that, by definition of reducibility, we have $(\xi_{u}(a))_{u\in U}\in\Gamma$, for all $a\in E_{U}$, and thus we can partition $E_U$ according to $E_{U}=\bigcup_{\gamma'\in \Gamma'}E_{U}^{\gamma'}$, where
\begin{equation*}
E_{U}^{\gamma'}\coleqq \{a\in E_{U}:(\xi_{u}(a))_{u\in U}\in \Gamma_{\gamma'}\},\quad \text{for }\gamma'\in\Gamma'.
\end{equation*}
Now, as $E_{U}$ is not a countable union of nowhere dense sets, and $\Gamma'$ is countable, there must exist a $\gamma'\in\Gamma'$ such that $E_{U}^{\gamma'}$ is dense in an open subset of $\R$. Next, consider $\vartheta\in\R^{U}$ such that $\sum_{u\in U}\beta_u \vartheta_u=0$. Then
\begin{equation}\label{eq:InpAnchLem-1}
\sum_{u\in U}(\xi_u(a)-\gamma'_u)\,\vartheta_u=0,
\end{equation}
for all $a\in E^{\gamma'}_{U}$, by definition of $\Gamma_{\gamma'}$. As $\sigma(\R)\subset \R$, the functions $a\mapsto \xi_u(a)$, for $u\in U$, are holomorphic in a neighborhood of $\R$. Hence, as $E_{U}^{\gamma'}$ has a cluster point in $\R$, it follows by the identity theorem \cite[Thm. 10.18]{Rudin1987} that \eqref{eq:InpAnchLem-1} holds for all $a\in\R$. Now, as $\vartheta$ was arbitrary, we see that, for every $a\in\R$, there exists a $\xi(a)\in\R$ such that 
\begin{equation*}
(\xi_u(a)-\gamma'_u)_{u\in U}=\xi(a)\cdot (\beta_u)_{u\in U}.
\end{equation*}
Then
\begin{equation*}
\sum_{v\in P'\cap\pre_{\m{M}}(u)}\omega_{uv}\OTnoA{v}{\m{M}}{\sigma}(a)= -\theta_u+\gamma'_u+\beta_u\xi(a),\quad a\in\R,
\end{equation*}
for all $u\in U$, and thus, for all $u_1,u_2\in U$, we have that
\begin{equation}\label{eq:InpAnchLem-2}
\begin{aligned}
&\sum_{v\in P'\cap\pre_{\m{M}}(u_1)}\hspace*{-5mm}\beta_{u_1}^{-1}\omega_{u_1v}\OTnoA{v}{\m{M}}{\sigma}(a)\;-\hspace*{-3mm}\sum_{v\in P'\cap\pre_{\m{M}}(u_2)}\hspace*{-5mm} \beta_{u_2}^{-1}\omega_{u_2v}\OTnoA{v}{\m{M}}{\sigma}(a)=\beta_{u_1}^{-1}( \gamma_{u_1}'-\theta_{u_1})-\beta_{u_2}^{-1}( \gamma_{u_2}'-\theta_{u_2}),
\end{aligned}
\end{equation}
is constant as a function of $a\in\R$. We now use this identity to construct a subnetwork $\m{N}$ of $\m{M}$ with one identically zero output and input $\{v_{D_0}^0\}$, thereby establishing statement (ii) of the proposition. This will be done analogously to the construction of the network $\m{N}'$ in the proof of Proposition \ref{prop:pole-structure}. Concretely, we proceed by showing that there exist $u_1,u_2\in U$, $u_1\neq u_2$, such that either
\begin{equation*}
 \pre_{\m{M}}(u_1)\setminus P\neq \pre_{\m{M}}(u_2)\setminus P
\end{equation*}
or
\begin{equation*}
\begin{aligned}
 \wtd{P}\coleqq \pre_{\m{M}}(u_1)\setminus P&= \pre_{\m{M}}(u_2)\setminus P\quad \text{and}\\
 \beta_{u_1}^{-1}\{\omega_{u_1v}\}_{v\in \wtd P}&\neq \beta_{u_2}^{-1}\{\omega_{u_2v}\}_{v\in \wtd P}.
 \end{aligned}
 \end{equation*} 
  Suppose by way of contradiction that this is not the case. First note that $\# (U)\geq 2$, as $\sigma$ is non-constant. Next, recalling that $P'=\bigcup_{u\in U}(\pre_{\m{M}}(u)\setminus P)$, we have $\pre_{\m{M}}(u)\setminus P=P'$, for all $u\in U$, and there exists a set of nonzero real numbers $\{\wtd{\kappa}_{v}\}_{v\in P'}$ 
 such that $\{\omega_{uv}\}_{v\in P'}={\beta}_u\{\wtd{\kappa}_v\}_{v\in P'}$, for all $u\in U$. Recalling that also $\{\omega_{uv}\}_{v\in P}=\beta_u\{{\kappa}_v\}_{v\in P}$, we obtain $\{\omega_{uv}\}_{v\in P\cup P'}=\beta_u\{\kappa_v'\}_{v\in P\cup P'}$, where
\begin{equation*}
\kappa'_v=\begin{cases}
\kappa_v, & v\in P\\
\wtd{\kappa}_v, & v\in P'
\end{cases}\;.
\end{equation*}
But this implies that $\m{M}$ is $(\rho,U)$--reducible, contradicting the assumption that $\m{M}$ is irreducible.  We can therefore find $u_1,u_2\in U$, $u_1\neq u_2$, such that either $\pre_{\m{M}}(u_1)\setminus P\neq \pre_{\m{M}}(u_2)\setminus P$, or $\wtd{P}\coleqq \pre_{\m{M}}(u_1)\setminus P= \pre_{\m{M}}(u_2)\setminus P$ and $\beta_{u_1}^{-1}\{\omega_{u_1v}\}_{v\in \wtd P}\neq \beta_{u_2}^{-1}\{\omega_{u_2v}\}_{v\in \wtd P}$. It hence follows that there exists a $v\in P'$ such that one of the following statements holds:
\begin{equation}
\begin{aligned}
(v,u_1)\in E^{\m{M}}& \text{ and } (v,u_2)\notin E^{\m{M}},\\
(v,u_1)\notin E^{\m{M}}& \text{ and } (v,u_2)\in E^{\m{M}},\text{ or } \label{eq:inp-anch-1}\\
(v,u_1), (v,u_2)\in E^{\m{M}}, & \text { and }\beta_{u_1}^{-1}\omega_{u_1v}-\beta_{u_2}^{-1}\omega_{u_2v}\neq 0.
\end{aligned}
\end{equation}
Hence $S\coleqq\{v\in P':\text{ one of \eqref{eq:inp-anch-1} holds}\}$ is nonempty, and we can set
\begin{equation*}
{\lambda}_{v}^{(1),\,\m{N}}=
\begin{cases}
\beta_{u_1}^{-1}\omega_{u_1v}, &\text{if }(v,u_1)\in E^{\m{M}}, (v,u_2)\notin E^{\m{M}}\\
-\beta_{u_2}^{-1}\omega_{u_2v}, &\text{if }(v,u_1)\notin E^{\m{M}}, (v,u_2)\in  E^{\m{M}} \\
\beta_{u_1}^{-1}\omega_{u_1v}-\beta_{u_2}^{-1}\omega_{u_2v}, & \text{if } (v,u_1), (v,u_2)\in E^{\m{M}} 
\end{cases},\quad \text{ for }v\in S,
\end{equation*}
and 
\begin{equation}\label{eq:lambda-N-inp-anch}
{\Lambda}^\m{N}=\{{\lambda}^{(1),\,\m{N}}\coleqq -\beta_{u_1}^{-1}(\gamma_{u_1}'-\theta_{u_1})+\beta_{u_2}^{-1}(\gamma_{u_2}'-\theta_{u_2})\}\cup\{\lambda_{v}^{(1),\,\m{N}}:v\in S\}.
\end{equation}
We now take $\m{N}=(V^{\m{N}},E^{\m{N}},\{v_{D_0}^0\},S,\Omega^{\m{N}},\Theta^{\m{N}},\Lambda^{\m{N}})$ to be the subnetwork of $\m{M}$ with one-dimen\-sional output, generated by $S$, and with $\Lambda^{\m{N}}$ as given in \eqref{eq:lambda-N-inp-anch}.
Then $\OnoA{\m{N}}{\sigma}=0$ by \eqref{eq:InpAnchLem-2}, and $\m{N}$ is strongly regular, as $\m{M}$ is. This establishes statement (ii) of the proposition and hence completes its proof.
\end{proof}

\subsection{Proof of Theorem \ref{thm:intro-Zab-NNC}}

We are now ready to combine the results of Sections \ref{sec:Clustering} and \ref{sec:inp-anch} to prove Theorem \ref{thm:intro-Zab-NNC}.

\begin{proof}[Proof of Theorem \ref{thm:intro-Zab-NNC}]
We argue by contradiction, so suppose that the statement is false. Specifically, fix a non-trivial regular GFNN $\m{A}$ with one-dimensional identically zero output and input set $\Vin$ of minimal cardinality. Then, as $\Vin$ is of minimal cardinality, $\m{A}$ must be strongly regular. We further claim that $\#(\Vin)=1$. To see this, suppose by way of contradiction that $\#(\Vin)\geq 2$, and apply Proposition \ref{prop:inp-anch} to $\m{A}$. Note that both circumstances of Proposition \ref{prop:inp-anch} yield a strongly regular network $\m{A}'$ with one-dimensional identically zero output, and input set $\Vin'$ strictly contained in $\Vin$. As $\#(\Vin')<\#(\Vin)$, we have a contradiction to the minimality of $\#(\Vin)$, and hence must have $\#(\Vin)=1$.

Now, as ${\OnoA{\m{A}}{\sigma}}\vert_{\R} =0$, it follows by the identity theorem that ${\OnoA{\m{A}}{\sigma}}$ continues in a unique fashion to the zero function on its natural domain $\dom_{\OnoA{\m{A}}{\sigma}}=\C$. On the other hand, as $\m{A}$ is non-trivial, Proposition \ref{prop:pole-structure} implies that the natural domain $\dom_{\OnoA{\m{A}}{\sigma}}$ of the analytic continuation of ${\OnoA{\m{A}}{\sigma}}$ is equal to $\C\setminus \overline {P_{\m{A}}}$, where $P_{\m{A}}$ is the set of poles of ${\OnoA{\m{A}}{\sigma}}$ satisfying $L_{\C}(\overline {P_{\m{A}}})=L(\m{A})\geq 1$. This, in particular, implies that $P_{\m{A}}$ must be nonempty, which stands in contradiction to $\dom_{\OnoA{\m{A}}{\sigma}}=\C$, completing the proof.

\end{proof}

\section{The alignment conditions for $\LtCla{a}{b}$-nonlinearities}\label{sec:LtCla-a-b}

\subsection{Basic properties of $\LtCla{a}{b}$-nonlinearities}

In this section, we derive various straightforward results about lattices in $\C$ and the functions in $\LtCla{a}{b}$ and use these findings to establish both the SAC and the CAC for $\LtCla{a}{b}$-nonlinearities. We begin with a lemma listing several elementary 
properties of $\LtCla{a}{b}$-nonlinearities. In the following we write $d(z,F)=\inf\{|z-w|:w\in F\}$ for the Euclidean distance between the point $z\in\C$ and the set $F\subset\C$.

\begin{lemma}\label{lem:Zab-props}
Let $a,b>0$, let $\{c_k\}_{k\in\Z}$ be a sequence of complex numbers, and suppose $a'\in(0,a)$ is such that $\sup_{k\in\Z}|c_k|e^{-\pi a' |k|/b}<\infty$. Then the series in \eqref{eq:Zab-series} converges uniformly on compact subsets of $\C\setminus P_\sigma$, where $P_\sigma=\{ak_a+ib(k_b+1/2):k_a,k_b\in\Z,\, c_{k_a}\neq 0\}$. Moreover, the function $\sigma$ given by \eqref{eq:Zab-series} has the following properties:
\begin{enumerate}[(i)]
\item $\sigma$ is an $ib$-periodic meromorphic function on $\C$,
\item the set of poles of $\sigma$ is $P_{\sigma}\subset \big(\frac{ib}{2}+a\Z\,\times\,ib \Z\big)$,
\item every pole of $\sigma$ is of order 1,
\item there exist constants $M>0$ and $\eta\in (0,\pi)$ such that $|\sigma(z)|\leq \frac{M }{1\wedge d(z,P_\sigma)}\,e^{\eta |z|/b}$, for all $z\in \dom_\sigma\coleqq \C\setminus P_\sigma$.
\end{enumerate}
\end{lemma}
\noindent The proof of Lemma \ref{lem:Zab-props} can be found in the Appendix. 

\subsection{Asymptotic density and the CAC for $\LtCla{a}{b}$-nonlinearities}

The first main result of this section is the following proposition that immediately implies the CAC for functions in $\LtCla{a}{b}$.

\begin{prop}\label{prop:laced-breakdown}
Let $a,b>0$, $\sigma\in \LtCla{a}{b}$, and let $\{(\alpha_s,\beta_s,\gamma_s)\}_{s\in\m{I}}$ be a nonempty finite set of triples of complex numbers such that $\alpha_s,\beta_s\in\C\setminus\{0\}$, for all $s\in \m{I}$. Furthermore, let $\{\epsilon_s\}_{s\in\m{I}}$ be ABCs, and
suppose that the function
\begin{equation}\label{eq:lac-break-stat}
z\mapsto f(z)\coleqq \sum_{s\in\m{I}}\alpha_s\,\sigma\left(\beta_s z+\gamma_s+\epsilon_s(1/z)\right)
\end{equation}
is analytic on $\C\setminus D(0,R)$, for some $R>0$. Then the set $\m{I}$ can be partitioned into sets $\m{I}_1,\dots,\m{I}_n$ such that, for every $j\in\{1,\dots,n\}$,
\begin{enumerate}[(i)]
\item there exists an ABC $\xi_j$ so that $\epsilon_s=\beta_s \xi_j$, for all $s\in\m{I}_j$, and
\item the function $f_j\coleqq \sum_{s\in\m{I}_j}\alpha_s\,\sigma(\beta_s\cdot\,+\,\gamma_s)$ is entire.
\end{enumerate}
\end{prop}

The proof of Proposition \ref{prop:laced-breakdown} uses several ancillary results about asymptotic densities of arithmetic sequences and lattices in the sense of  Definition \ref{def:dens}. Concretely, we will need the following three lemmas, whose proofs can be found in the Appendix, as well as a special case of Weyl's equidistribution theorem, which was also employed in the proof of the ``Deconstruction Lemma'' in \cite{Fefferman1994}.

\begin{lemma}\label{lem:d>0->Q}
Let $\Pi= a\Z\times ib\Z$ be a lattice in $\C$, where $a,b>0$. Let $\beta\in\C\setminus \{0\}$ and $\gamma\in\C$, and set $P=\beta^{-1}(\Pi-\gamma)$. Suppose $\ell$ is a line in $\C$ such that $\Delta(\ell,P)>0$. Then $\ell\cap P$ is an arithmetic sequence, and there exists an $\varepsilon_0>0$ such that $(\ell+D(0,\varepsilon_0))\cap P=\ell\cap P$.
\end{lemma} 

\begin{lemma}\label{lem:2-lattice-fit}
Let $\Pi= a\Z\times ib\Z$ be a lattice in $\C$, where $a,b>0$. Let $\beta_1,\beta_2\in\C\setminus \{0\}$ and $\gamma_1,\gamma_2\in\C$, and set $P_j=\beta_j^{-1}(\Pi-\gamma_j)$, for $j\in\{1,2\}$. Suppose that $\ell$ is a line in $\C$ such that $\ell\cap P_j=\{x_j+ky_j:k\in\Z \}$, $j\in\{1,2\}$, are arithmetic sequences so that $y_1/y_2$ is real and rational. Then the sets $P_{\ell,c}\coleqq\{p\in P_1\cup P_2: d(p,\ell)\leq c\}$ are uniformly discrete, for all $c>0$.
\end{lemma}

\begin{lemma}\label{lem:ent-vs-line}
Let $a,b>0$ and let $\sigma\in \LtCla{a}{b}$. Furthermore, let $\{(\alpha_s,\beta_s,\gamma_s)\}_{s\in\m{I}}$ be a finite set of triples of complex numbers such that $\alpha_s,\beta_s\neq 0$, for all $s\in \m{I}$, and set $f\coleqq \sum_{s\in\m{I}}\alpha_s\,\sigma(\beta_s\cdot\,+\, \gamma_s)$. Then, either
\begin{enumerate}[(i)]
\item $f$ is entire, or
\item $f$ has a nonempty set of poles $P_f$, and there exists a line $\ell$ in $\C$ such that $\Delta(\ell,P_f)>0$.
\end{enumerate}
\end{lemma}

\begin{prop}[Weyl, {\cite[Cor. 2.A.12]{Fefferman1994}}]\label{prop:Weyl-equi}
Let $x_1,x_2\in\C$, $y_1,y_2\in\C\setminus\{0\}$, and define the arithmetic sequences $\Pi_j=\{x_j+ky_j:k\in\Z\}$, for $j\in\{1,2\}$. If $y_1/y_2$ is real and irrational, then $\Delta(\Pi_1,\Pi_2)= 0$.
\end{prop}

\noindent We are now ready to prove Proposition \ref{prop:laced-breakdown}.

\begin{proof}[Proof of Proposition \ref{prop:laced-breakdown}]
Let $\delta\in(0,1/R)$ be sufficiently small for the functions $\epsilon_s$ to be analytic on an open neighborhood of $ D(0,\delta)$, for all $s\in \m{I}$.
Then, for every $s\in\m{I}$, 
\begin{equation*}
z\mapsto \wtd\sigma_s(z)\coleqq \sigma\left(\beta_s z+\gamma_s+\epsilon_s(1/z)\right)
\end{equation*}
 is a meromorphic function on $\dom_\delta\coleqq \C\setminus D(0,1/\delta)$. Let $\wtd{P}_s\subset \dom_\delta$ denote its set of poles. Next, for $s\in\m{I}$, set $P_s=\beta_s^{-1}(P_\sigma-\gamma_s)$, where $P_\sigma$ is the set of poles of $\sigma$. We now show the following:

\textit{Claim: There exist $\delta'\in(0,\delta)$ and $A> 1/(2\delta')$ such that, for all $s\in\m{I}$, the function $g_s:D^\circ (0,\delta')\to \C$ given by $g_s(z)=\frac{\beta_s z}{\beta_s+z\epsilon_s(z)}$ is biholomorphic onto its image $\mrm{Img}(g_s)\supset D^\circ (0,1/A)$, and, for every $p'\in \wtd{P}_s\setminus D(0,2A)$, we have 
\begin{equation}\label{eq:lacing-lem-1}
\begin{aligned}
p&\coleqq 1/g_s(1/p') \; \in P_s\setminus D(0,A),\quad \text{and }\\
p'&=p-h_s(1/p),
\end{aligned}
\end{equation}
where $h_s\coleqq\beta_{s}^{-1} (\epsilon_s\circ g_s^{-1}):D^\circ(0,1/A)\to\C$.
}\\
\noindent\textit{Proof of Claim.} First note that, for every $s\in \m{I}$, the function $z\mapsto g_s(z)=\frac{\beta_s z}{\beta_s+z\epsilon_s(z)}$ is holomorphic on a neighborhood of $0$. Moreover, we have $g_s(0)=0$ and $g_s'(0)=1$, and thus by the complex open mapping theorem \cite[Thm. 10.32]{Rudin1987}, there exists $\delta_s\in(0,\delta)$ such that $g_s:D^\circ (0,\delta_s)\to \C$ is biholomorphic onto its image. Let $\delta'=\min_{s\in\m{I}}\delta_s$ and  $A>0$ be such that
\begin{equation*}
\begin{aligned}
&\max_{|z|\leq 1/(2A)}\left|\epsilon_s(z)\right|\leq 1,  && \text{for all }s\in\m{I},\\
&A> \frac{1}{2\delta'} \, \vee \, \max_{s\in\m{I}} |\beta_s|^{-1}, &&\text{for all }s\in\m{I}, \text{ and} \\
&D^\circ(0,1/A)\subset\bigcap_{s\in\m{I}}g_s\!\left(D^\circ (0,\delta')\right). \\
\end{aligned}
\end{equation*}
The last of these conditions implies that the image $\mrm{Img}(g_s)$ contains $D^\circ (0,1/A)$, for every $s\in\m{I}$. We proceed to show \eqref{eq:lacing-lem-1}.
To this end, fix an $s\in\m{I}$ and take a $p'\in \wtd{P}_s\setminus D(0,2A)$. Then,$|p'|>2A$ and $\psi\coleqq \beta_s p'+\gamma_s+\epsilon_s(1/p')\in P_\sigma$. Now,
\begin{equation*}
p=1/g_s(1/p')=\frac{\beta_s+(1/p')\cdot\epsilon_s(1/p')}{\beta_s/p'}=\beta_s^{-1}(\psi-\gamma_s)\in \beta_s^{-1}(P_\sigma - \gamma_s)= P_s,
\end{equation*}
and, as $|1/p'|<1/(2A)< \delta'$, we have
\begin{equation*}
|p|=|p'+\beta_s^{-1}\epsilon_s(1/p')|\geq |p'|-|\beta_s^{-1}||\epsilon_s(1/p')|>2A-A\cdot 1=A. 
\end{equation*}
This establishes $p\in P_s\setminus D(0,A)$. Next, as $|1/p'|< \delta'$ and $g_s:D^\circ (0,\delta')\to \C$ is a bijection with its image containing $D^\circ(0,1/A)$, we must have $1/p'=g_s^{-1}(1/p)$. Therefore,
\begin{equation*}
p'=\beta_s^{-1}(\psi-\gamma_s)-\beta_s^{-1}\epsilon_s(1/p')=p-\beta_{s}^{-1}\,\epsilon_s\!\left( g_s^{-1}(1/p)\right)=p-h_s(1/p),
\end{equation*}
as desired. This concludes the proof of the claim.

Now, define an undirected graph $\m G=(\m I, \m E)$ by setting
\begin{equation*}
\m{E}=\{(s_1,s_2)\in \m I \times \m I: s_1\neq s_2, \text{ and }\exists \text{ line }\ell \text{ in }\C\text{ s.t. } \Delta(\ell,\wtd{P}_{s_1}\cap \wtd{P}_{s_2})>0\},
\end{equation*}
and let $\m{I}_1,\dots,\m{I}_n$ be the subsets of $\m{I}$ corresponding to different connected components of $\m{G}$.
Next, fix a connected component $\m{I}_j$ of $\m{G}$. We proceed to establish the existence of an ABC $\xi_j$ such that $\epsilon_s=\beta_s\xi_j$, for all $s\in\m{I}_j$. If $\m{I}_j=\{s^*\}$ is a singleton set, we can then simply set $\xi_j=\beta_{s^*}^{-1}\epsilon_{s^*}$, so suppose that $\#(\m{I}_j)\geq 2$.

 Fix $s_1,s_2\in\m{I}$ such that $(s_1,s_2)\in\m{E}$, and let $\ell$ be a line in $\C$ such that $\Delta(\ell,\wtd{P}_{s_1}\cap\wtd{P}_{s_2})>0$. We proceed by showing that $P_{\ell,1}\coleqq\{p\in P_{s_1}\cup P_{s_2}: d(p,\ell)\leq 1\}$ is uniformly discrete. To this end, take an arbitrary $\varepsilon>0$, let $A_\varepsilon>A$ be such that 
 \begin{equation}\label{eq:dens-calc-pre-1}
 |h_{s_m}(1/z)|\leq \varepsilon,\quad \text{ for $|z|> A_\varepsilon$ and $m\in\{1,2\}$},
 \end{equation}
  and then let $B_\varepsilon>2A$ be such that $|1/g_{s_m}(1/z)|>A_\varepsilon$, for $|z|>B_\varepsilon$ and $m\in\{1,2\}$. We now have that, for each $m\in\{1,2\}$, whenever $p'\in \wtd{P}_{s_m}$ is such that $|p'|>B_\varepsilon$, then 
\begin{equation}\label{eq:dens-calc-pre-2}
p'=p-h_{s_m}(1/p),\quad\text{for some $p\in P_{s_m}$ s.t. $|p|>A_\varepsilon$}.
\end{equation}
 Indeed, for $p'\in \wtd{P}_{s_m}$ satisfying $|p'|>B_\varepsilon$, we have $p'\in \wtd{P}_{s_m}\setminus D(0,2A)$, and it hence follows by the Claim that $p'=p-h_{s_m}(1/p)$, where $p=1/g_{s_m}(1/p')\in P_{s_m}$. Then $|p|>A_{\varepsilon}$ by our choice of $B_\varepsilon$, establishing \eqref{eq:dens-calc-pre-2}.
 We use this to get the following estimate for both $m=1$ and $m=2$:
\begin{equation}\label{eq:massive-dens-est-1}
\begin{aligned}
 \Delta_{2\varepsilon}(\ell,P_{s_m})&=\limsup_{N\to\infty}\frac{1}{2N}\, \#\left\{p\in {P}_{s_m}\cap D(0,N):d(p,\ell)\leq 2\varepsilon \right\}\\
&\geq \limsup_{N\to\infty}\frac{1}{2N}\, \#\left\{\text{$p\in P_{s_m}\cap D(0,N):|p|> A_\varepsilon,\, d( p-h_{s_m}(1/p),\ell)\leq\varepsilon$} \right\}\\
&\geq \limsup_{N\to\infty}\frac{1}{2N}\, \#\left\{ p\in P_{s_m}:\quad \text{\parbox{6.2cm}{  $p-h_{s_m}(1/p)\in D(0,N-\varepsilon)$, $|p|> A_\varepsilon$, \\   $d( p-h_{s_m}(1/p),\ell)\leq \varepsilon$} }\, \right\}\\
&\geq \limsup_{N\to\infty}\frac{1}{2(N-\varepsilon)}\! \left(\#\left\{p'\in \wtd{P}_{s_m}\! \cap D(0,N-\varepsilon): d( p',\ell)\leq \varepsilon \right\} - \#\big(\wtd{P}_{s_m}\!\cap D(0,B_\varepsilon)\big)\!  \right)\\
&=\Delta_{\varepsilon}(\ell,\wtd{P}_{s_m})\geq \Delta(\ell,\wtd{P}_{s_m})\geq \Delta(\ell,\wtd{P}_{s_1}\cap\wtd{P}_{s_2}),
\end{aligned}
\end{equation}
where the first two inequalities follow by \eqref{eq:dens-calc-pre-1}, and the third is a consequence of \eqref{eq:dens-calc-pre-2}.
As $\varepsilon$ was arbitrary, we obtain $\Delta(\ell,P_{s_m})\geq\Delta(\ell,\wtd{P}_{s_1}\cap\wtd{P}_{s_2})>0$, and so it follows by Lemma \ref{lem:d>0->Q} that $\ell\cap P_{s_m}=\{x_m+ky_m:k\in\Z\}$ is an arithmetic sequence, and there exists an $\varepsilon_0>0$ such that we have the following implication
\begin{equation}\label{eq:snap-ell}
p\in P_{s_m},\; d(p,\ell)\leq 2\varepsilon \quad\implies\quad p\in\ell,
\end{equation}
for all $\varepsilon\in(0,\varepsilon_0)$ and both $m\in\{1,2\}$.
Fix such an $\varepsilon$ and define $A_\varepsilon$ and $B_\varepsilon$ as above so that we have \eqref{eq:dens-calc-pre-1} and \eqref{eq:dens-calc-pre-2}. We claim that, whenever $p'\in \wtd{P}_{s_1}\cap \wtd{P}_{s_2}$ is such that $|p'|>B_\varepsilon$ and $d(p',\ell)\leq\varepsilon$, then 
\begin{equation}\label{eq:dens-calc-pre-3}
p'=p_2-h_{s_2}(1/p_2),
\text{ where $p_2\in \ell \cap P_{s_2}\setminus D(0,A_\varepsilon)$ and $d(p_2,\ell\cap P_{s_1})\leq 2\varepsilon$}.
\end{equation}
Indeed, if $p'\in \wtd{P}_{s_1}\cap \wtd{P}_{s_2}$ is such that $|p'|>B_\varepsilon$, then by \eqref{eq:dens-calc-pre-2} we have
\begin{equation}\label{eq:dens-calc-pre-3.5}
p'=p_1-h_{s_1}(1/p_1)=p_2-h_{s_2}(1/p_2),
\end{equation}
where $p_m=1/g_{s_m}(1/p')\in P_{s_m}$ is such that $|p_m|>A_\varepsilon$, for $m\in\{1,2\}$. If $p'$ additionally satisfies $d(p',\ell)\leq \varepsilon$, then
\begin{equation*}
d(p_m,\ell)\leq d(p',\ell)+|h_{s_m}(1/p_m)|\leq \varepsilon+\varepsilon=2\varepsilon,\quad \text{for }m\in\{1,2\},
\end{equation*}
by \eqref{eq:dens-calc-pre-3.5} and \eqref{eq:dens-calc-pre-1}, and so \eqref{eq:snap-ell} establishes $p_1,p_2\in\ell$, further implying $p_1\in \ell\cap P_{s_1}$, $p_2\in\ell\cap P_{s_2}$, and
\begin{equation*}
d(p_2,\ell\cap P_{s_1})\leq |p_1-p_2|=|h_{s_1}(1/p_1)-h_{s_2}(1/p_2)|\leq \varepsilon+\varepsilon=2\varepsilon,
\end{equation*}
where we used $p_1\in \ell\cap P_{s_1}$ and  \eqref{eq:dens-calc-pre-1}. This establishes \eqref{eq:dens-calc-pre-3}.
We now argue
\begin{equation*}
\begin{aligned}
 &\Delta_{2\varepsilon}(\ell\cap P_{s_1},\ell\cap P_{s_2})=\limsup_{N\to\infty}\frac{1}{2N}\, \#\left\{p_2\in \ell \cap {P}_{s_2}\cap D(0,N):  d(p_2\,,\ell\cap P_{s_1})\leq 2\varepsilon \right\}\\
 \geq\, & \limsup_{N\to\infty}\frac{1}{2N}\, \#\left\{ p_2\in \ell \cap P_{s_2}:\;  p_2-h_{s_2}(1/p_2)\in D(0,N-\varepsilon),\; |p_2|> A_\varepsilon,\;  d( p_2,\ell\cap P_{s_1})\leq 2\varepsilon \, \right\}\\
\geq\, & \limsup_{N\to\infty}\frac{1}{2(N-\varepsilon)}\, \left(\#\left\{ p'\in\wtd{P}_{s_2}\cap\wtd{P}_{s_1}\cap D(0,N-\varepsilon): d(p',\ell)\leq \varepsilon \right\}-\# \big(\wtd{P}_{s_2}\!\cap D(0,B_\varepsilon)\big)\! \right)\\
=\,& \Delta_{\varepsilon}(\ell,\wtd{P}_{s_1}\cap\wtd{P}_{s_2})\geq \Delta(\ell,\wtd{P}_{s_1}\cap\wtd{P}_{s_2}),
\end{aligned}
\end{equation*}
where the first inequality follows by \eqref{eq:dens-calc-pre-1} and the second by \eqref{eq:dens-calc-pre-3}.
Taking the infimum over $\varepsilon$ yields $\Delta(\ell \cap P_{s_1},\ell\cap P_{s_2})\geq  \Delta(\ell,\wtd{P}_{s_1}\cap\wtd{P}_{s_2})>0$, and hence Proposition \ref{prop:Weyl-equi} implies $y_1/y_2\in\Q$. It now follows directly from Lemma \ref{lem:2-lattice-fit} that $P_{\ell,1}\coleqq\{p\in P_{s_1}\cup P_{s_2}: d(p,\ell)\leq 1\}$ is uniformly discrete.

We are now ready to show that $\beta_{s_1}^{-1}\epsilon_{s_1}=\beta_{s_2}^{-1}\epsilon_{s_2}$. To this end, let $\{q_k\}_{k\in\N}$ be a sequence in $\wtd{P}_{s_1}\cap\wtd{P}_{s_2}$ such that $|q_k|\to\infty$ and $d(q_k,\ell)\to 0$ as $k\to\infty$, and let $\{p_k'\}_{k\in\N}\subset P_{s_1}$ and $\{p_k\}_{k\in\N}\subset P_{s_2}$ be the corresponding sequences such that $|p_k'|\to\infty$ and $|p_k|\to\infty$ as $k\to \infty$, and
\begin{equation*}
p_k'-h_{s_1}(1/p_k')=p_k- h_{s_2}(1/p_k)=q_k,\quad\text{for all }k\in\N.
\end{equation*}
Then $p_k,p_k'\in P_{\ell,1}$ for sufficiently large $k$, and, since 
\begin{equation*}
p_k'-p_k=h_{s_1}(1/p_k')- h_{s_2}(1/p_k)\to h_{s_1}(0)-h_{s_2}(0)=0-0= 0\quad\text{as }k\to\infty,
\end{equation*}
and $P_{\ell,1}$ is uniformly discrete, we must have $p_k'=p_k$, for all sufficiently large $k$. Therefore $(h_{s_1}-h_{s_2})(1/p_k)=0$, for all sufficiently large $k$. Since $h_{s_1}-h_{s_2}$ is holomorphic in $D^\circ(0,1/A)$ and $1/p_k\to 0$ as $k\to\infty$, it follows by the identity theorem that $ h_{s_1}-h_{s_2}=0$ on $D^\circ(0,1/A)$. Now, choose an arbitrary $x\in D^{\circ}(0,\delta')\cap g_{s_1}^{-1}\big(D^\circ(0,1/A)\big)$, and set $x'=g_{s_2}^{-1}(g_{s_1}(x))$. Then
\begin{equation*}
\beta_{s_2}^{-1}\epsilon_{s_2}(x')=\beta_{s_2}^{-1}(\epsilon_{s_2}\circ g_{s_2}^{-1}\circ g_{s_1}(x))=h_{s_2}\circ g_{s_1}(x)=h_{s_1}\circ g_{s_1}(x)=\beta_{s_1}^{-1}\epsilon_{s_1}(x),
\end{equation*}
and thus
\begin{equation*}
\begin{aligned}
&\frac{x-x'}{(1+x\,\beta_{s_2}^{-1}\epsilon_{s_2}(x'))(1+x'\,\beta_{s_2}^{-1}\epsilon_{s_2}(x'))}= \frac{x}{1+x\,\beta_{s_2}^{-1}\epsilon_{s_2}(x')}-\frac{x'}{1+x'\,\beta_{s_2}^{-1}\epsilon_{s_2}(x')} \\
=\,&\frac{x}{1+x\,\beta_{s_1}^{-1}\epsilon_{s_1}(x)}-\frac{x'}{1+x'\,\beta_{s_2}^{-1}\epsilon_{s_2}(x')}=g_{s_1}(x)-g_{s_2}(x')=0.
\end{aligned}
\end{equation*}
Hence, $x=x'$ and $\beta_{s_1}^{-1}\epsilon_{s_1}(x)=\beta_{s_2}^{-1}\epsilon_{s_2}(x)$, and, since $x$ was arbitrary, we again deduce by the identity theorem that $\beta_{s_1}^{-1}\epsilon_{s_1}=\beta_{s_2}^{-1}\epsilon_{s_2}$ on $D^\circ(0,\delta)$.

Now, choose an arbitrary $s^*\in\m{I}_j$ and define the ABC $\xi_j=\beta^{-1}_{s^*}\epsilon_{s^*}$. Then, for every $s\in \m{I}_j$, as $\m{I}_j$ is a connected component of $\m{G}$, we can find a finite sequence $s_1=s,s_2,\dots,s_{m-1},s_m=s^*$ in $\m{I}_j$ such that $(s_{k},s_{k+1})\in\m{E}$, for $k\in \{1,\dots, m-1\}$. Consequently,
\begin{equation*}
\beta^{-1}_{s}\epsilon_{s}=\beta^{-1}_{s_1}\epsilon_{s_1}=\dots=\beta^{-1}_{s_m}\epsilon_{s_m}=\beta^{-1}_{s^*}\epsilon_{s^*}=\xi_j,
\end{equation*}
and thus $\epsilon_s=\beta_{s} \xi_j$, for all $s\in \m{I}_j$. As the connected component $\m{I}_j$ of $\m{G}$ was arbitrary, we have established item (i).

It remains to show that the functions $f_j= \sum_{s\in\m{I}_j}\alpha_s\,\sigma(\beta_s\cdot\,+\,\gamma_s)$ are entire. To this end, fix a $j\in\{1,\dots,n\}$, and suppose by way of contradiction that the set $P_{f_j}$ of poles of $f_j$ is nonempty. Then, by Lemma \ref{lem:ent-vs-line}, there must exist a line $\ell$ in $\C$ such that $\Delta(\ell,P_{f_j})>0$. Next, define the function
\begin{equation*}
z\mapsto \widetilde{f}_{j}(z)\coleqq\sum_{s\in\m{I}_j}\alpha_s\,\sigma\big(\beta_s z+\gamma_s+\underbrace{\epsilon_s(1/z)}_{=\,\beta_s\xi_j(1/z)}\big)=f_j\left(z+\xi_j(1/z)\right)
\end{equation*}
on $\dom_\delta$, and let $P_{\wtd{f}_j}$ denote its set of poles. Now, for every $p\in P_{f_j}$ with sufficiently large $|p|$, there exists a unique $p'\in \dom_{\delta}$ such that 
\begin{equation*}
\frac{1/p'}{1+(1/p')\cdot \xi_j(1/p')}=\frac{1}{p},
\end{equation*}
and $1/p'\to 0$ as $|p|\to \infty$.
Then $p'\in P_{\widetilde{f}_j}$ and $|p-p'|=|\xi_j(1/p')|\to 0$ as $|p|\to\infty$. Performing density estimates analogous to \eqref{eq:massive-dens-est-1}, we find that $\Delta(\ell,P_{\widetilde{f}_j})\geq \Delta(\ell,P_{{f}_j})>0$. Finally, we let $P_{f}\subset \dom_\delta$ be the set of poles of $f$, and argue
\begin{equation*}
\begin{aligned}
\Delta(\ell,P_{f})&\geq \Delta\Big(\ell,P_{\widetilde{f}_j}\bigm\backslash \bigcup_{s\in\m{I}\setminus\m{I}_j}P_{\wtd{f}_j}\cap \wtd{P}_s\Big)\\
&\geq \Delta\Big(\ell,P_{\widetilde{f}_j}\bigm\backslash\bigcup_{s\in\m{I}\setminus\m{I}_j} \bigcup_{s'\in\m{I}_j} \wtd{P}_{s'}\cap \wtd{P}_s\Big)\\
&\geq  \Delta(\ell,P_{\widetilde{f}_j})-\sum_{s\in\m{I}\setminus\m{I}_j}\sum_{s'\in\m{I}_j}\underbrace{\Delta(\ell, \wtd{P}_{s'}\cap \wtd{P}_s)}_{=\,0}\\
&=\Delta(\ell,P_{\widetilde{f}_j})>0,
\end{aligned}
\end{equation*}
where we used that $\Delta(\ell, \wtd{P}_{s'}\cap \wtd{P}_s)=0$, for $s$ and $s'$ in different connected components of $\m{I}$, by definition of the graph $\m{G}$.
This in particular implies that $P_{f}\neq\varnothing$, which stands in contradiction to the assumption that $f$ is analytic on $\dom_\delta$, and hence establishes that $f_j$ must be entire. Since $j\in\{1,\dots,n\}$ was arbitrary, the proof of the proposition is complete.
\end{proof}

\subsection{The SAC for $\LtCla{a}{b}$-nonlinearities}

The second main result of this section establishes the SAC for $\LtCla{a}{b}$-nonlinearities:

\begin{prop}\label{prop:Zab-SAC}
Let $a,b>0$ and let $\sigma\in\LtCla{a}{b}$. Then $\sigma$ satisfies the SAC.
\end{prop}

\noindent The proof of Proposition \ref{prop:Zab-SAC} relies on Carlson's theorem, as well as Lemma \ref{lem:ent-vs-line} that we already used to establish the CAC.

\begin{prop}[Carlson {\cite[Sec. 5.81]{Titchmarsh1939}}]\label{prop:Carlson}
Assume that $f$ is an entire function such that
\begin{enumerate}[(i)]
\item there exist $M>0$ and $\eta\in(0,\pi)$ so that $|f(z)|\leq Me^{\eta|z|}$, for all $z\in\C$, and
\item $f(n)=0$, for all $n\in\N$.
\end{enumerate}
Then $f$ is identically $0$.
\end{prop}

\begin{proof}[Proof of Proposition \ref{prop:Zab-SAC}]
Let $R>0$ and a finite set $\{(\alpha_s,\beta_s,\gamma_s)\}_{s\in\m{I}}\in\C^{\m{I}}\times\R^{\m{I}}\times \R^{\m{I}}$ be such that $f\coleqq \sum_{s\in\m{I}}\alpha_s\,\sigma(\beta_s\cdot\,+\gamma_s)$ is analytic on $\C\setminus D(0,R)$, and assume w.l.o.g. that $\alpha_s\neq 0$, $\beta_s\neq 0$, for all $s\in\m{I}$. We use induction on $\#(\m{I})$ to show that $f$ is constant. If $\#(\m{I})=0$, i.e., $\m{I}=\varnothing$, then $f$ is given by the empty sum, and so $f\equiv 0$ is constant. Suppose now that $\#(\m{I})\geq 1$, and assume that the implication in the definition of the SAC holds for all $\{(\alpha_s',\beta_s',\gamma_s')\}_{s\in\m{I}'}\in\C^{\m{I}'}\times\R^{\m{I}'}\times \R^{\m{I}'}$ with $\#(\m{I}')<\#(\m{I})$. First, note that, as the set of poles of $f$ is bounded, its density along any line in $\C$ is zero, and so it follows by Lemma \ref{lem:ent-vs-line} that $f$ must be entire. Now, let $\beta_{\mrm{max}}=\max\{|\beta_s|:s\in\m{I}\}$, $\beta_{\mrm{min}}=\min\{|\beta_s|:s\in\m{I}\}$, and set $\m{I}_1=\{s\in\m{I}:|\beta_s|=\beta_{\mrm{max}}\}$.
 Then, as the functions $\sigma_s\coleqq \alpha_s \sigma(\beta_s\cdot\,+\gamma_s)$ do not have poles along $\frac{ib}{2\beta_{\mrm{max}}}+\R$, for $s\in\m{I}\setminus\m{I}_1$, the function
\begin{equation*}
f_1\coleqq \sum_{s\in\m{I}_1}\alpha_s\,\sigma(\beta_s\cdot\,+\gamma_s)=f-\sum_{s\in\m{I}\setminus\m{I}_1}\alpha_s\,\sigma(\beta_s\cdot\,+\gamma_s)
\end{equation*}
does not have poles along $\frac{ib}{2\beta_{\mrm{max}}}+\R$ either. Therefore, as $f_1$ is $\frac{ib}{\beta_{\mrm{max}}}\,$--periodic and its poles are contained in $\bigcup_{n\in\Z}\big[\R+\frac{ib}{\beta_{\mrm{max}}}\left(n+\frac{1}{2}\right)\big]$, it follows that $f_1$ is entire. Next, by item (iv) of Lemma \ref{lem:Zab-props}, there exist $M>0$ and $\eta\in(0,\pi)$  such that $|\sigma(z)|\leq \frac{M }{1\wedge d(z,P_\sigma)}\,e^{\eta |z|/b}$, for all $z\in \dom_\sigma\coleqq \C\setminus P_\sigma$, where $P_\sigma$ is the set of poles of $\sigma$. Now, let $P_s=\beta_s^{-1}(P_\sigma-\gamma_s)$ be the set of poles of $\sigma_s$, for $s\in\m{I}_1$, and set $\wtd{P}=\bigcup_{s\in\m{I}_1}P_{s}$. Then, for $z\in\C\setminus\wtd{P}$,
\begin{equation*}
\begin{aligned}
|f_1(z)|&\leq \sum_{s\in\m{I}_1} |\alpha_s| |\sigma(\beta_sz+\gamma_s)|\leq \sum_{s\in\m{I}_1} |\alpha_s|\frac{M }{1\wedge d(\beta_sz+\gamma_s, P_\sigma)}\,e^{\eta |\beta_sz+\gamma_s|/b}\\
&\leq  \sum_{s\in\m{I}_1} |\alpha_s|\frac{M }{1\wedge \beta_{\mrm{min}} d(z, \beta_s^{-1}(P_\sigma-\gamma_s))}\,e^{\eta \beta_{\mrm{max}} |z|/ b}e^{\eta |\gamma_s|/b}\\
&\leq \sum_{s\in\m{I}_1} |\alpha_s|\frac{M (1\wedge \beta_{\mrm{min}})^{-1}}{1\wedge d(z, P_s)}\,e^{\eta \beta_{\mrm{max}} |z|/ b} e^{\eta |\gamma_s|/b}\\
&\leq\frac{1}{1\wedge d(z, \wtd{P})} \underbrace{M (1\wedge \beta_{\mrm{min}})^{-1}  \Bigg( \sum_{s\in\m{I}_1} |\alpha_s|e^{\eta |\gamma_s|/b}\Bigg) }_{M'\coleqq} \,e^{\eta \beta_{\mrm{max}} |z|/ b}.\\
\end{aligned}
\end{equation*}
Now, as 
\begin{equation*}
\wtd{P}=\bigcup_{s\in\m{I}_1}P_s\subset \bigcup_{s\in\m{I}_1} \Big(\textstyle\frac{ib}{2\beta_{\mrm{max}}}-\frac{\gamma_s}{\beta_s}+\frac{a}{\beta_{\mrm{max}}}\Z\times\frac{ib}{\beta_{\mrm{max}}}\Z\Big) = \{\frac{ib}{2\beta_{\mrm{max}}}-\frac{\gamma_s}{\beta_s}:s\in\m{I}_1\} +\frac{a}{\beta_{\mrm{max}}}\Z\times\frac{ib}{\beta_{\mrm{max}}}\Z\,,
\end{equation*}
 we have that $\wtd{P}$ is uniformly discrete, i.e.,
\begin{equation*}
\mu\coleqq \inf\{|p_1-p_2|:p_1,p_2\in\wtd{P},p_1\neq p_2\}>0.
\end{equation*}
Therefore, for $z\in \C$ such that $d(z, \wtd{P})\geq \mu/2$, we have $|f_1(z)|\leq \frac{M'}{1\wedge (\mu/2)}\,e^{\eta \beta_{\mrm{max}} |z|/ b} $. Suppose now that $z\in\C$ satisfies $d(z, \wtd{P})< \mu/2$. Then there exists a $p\in \wtd{P}$ such that $z\in D(p,\mu/2)$, and, by definition of $\mu$, $D(p,\mu/2)\cap \wtd{P}=\{p\}$. Let $z'\in\C$ be such that $|z-z'|=\mu/2$, but otherwise arbitrary. Now, as $f_1$ is analytic on an open neighborhood of the disk $D(p,\mu/2)$, it follows by the maximum modulus principle \cite[Thm. 10.24]{Rudin1987} that
\begin{equation*}
|f_1(z)|\leq |f_1(z')|\leq \frac{M'}{1\wedge (\mu/2)}\,e^{\eta \beta_{\mrm{max}} |z'|/ b} \leq \underbrace{\frac{M'e^{\eta \mu \beta_{\mrm{max}}/(2b)}}{1\wedge (\mu/2)}}_{M''\coleqq }\,e^{\eta \beta_{\mrm{max}} |z|/ b} .
\end{equation*}
Hence $|f_1(z)|\leq M''e^{\eta \beta_{\mrm{max}} |z|/ b} $, for all $z\in\C$. Now, $f_1\left(\frac{ib}{\beta_{\mrm{max}}}\,\cdot \right)-f_1(0)$ satisfies the assumptions of Proposition \ref{prop:Carlson}, and therefore must be identically zero. This establishes that $f_1$ is constant. But now
\begin{equation*}
f_2\coleqq \sum_{s\in\m{I}\setminus\m{I}_1}\alpha_s\,\sigma(\beta_s\cdot\,+\gamma_s)=f-f_1
\end{equation*}
is entire, and therefore constant, by the induction hypothesis, and so $f=f_1+f_2$ is constant. This completes the induction step and concludes the proof of the proposition.
\end{proof}

\section*{Acknowledgment}

The authors would like to thank Prof. Charles Fefferman for his insightful comments on an earlier version of the manuscript, which have lead to a significantly improved exposition in Section \ref{sec:ident-for-tanh} and a simplification of the proof of Proposition \ref{prop:pole-structure}.

\bibliographystyle{IEEEtran} 
\bibliography{ref}

\section*{Appendix: proofs of auxiliary results}

\setcounter{subsection}{0}

\subsection{Proof of Proposition \ref{prop:intro-exotic-symm}}
 
\begin{proof}
Let $\{r_k\}_{k\in\Z}\subset\R$ be a solution of the linear recurrence $\sum_{l=0}^n \alpha_{l}r_{k-l}=0$, $k\in\Z$, such that $r_0=\alpha_n$, $r_1=\alpha_{n-1}$, \dots, $r_{n-1}=\alpha_1$. Then $|r_k|$ grows at most exponentially, i.e., there exists a $b>0$ such that $\sup_{k\in\Z}|r_k|e^{-\pi |k|/b}<\infty$. We can thus define $\sigma\in\LtCla{1}{b}$ by
\begin{equation}\label{eq:Zab-series-exotic}
\sigma=\sum_{k\in\Z }r_k\,   \big[\sgn(k)+\, \tanh\!\big(\pi b^{-1}(\,\cdot - k)\big)\big].
\end{equation}
Then, by Lemma \ref{lem:Zab-props}, $\sigma$ is a meromorphic function with only simple poles contained in $\Pi\coleqq  \frac{ib}{2}+\Z\,\times\,ib \Z$. Hence,
$
f\coleqq \sum_{l=0}^n \alpha_{l}\,\sigma(\cdot -l)
$
 is also meromorphic with only simple poles contained in $\Pi$. However, their residues are
\begin{equation*}
\begin{aligned}
\mrm{Res}\big(f,k+(m+\textstyle\frac{1}{2})ib \big)&=\sum_{l=0}^n \alpha_l \, \mrm{Res}\big(\sigma(\cdot-l), k+(m+{\textstyle\frac{1}{2}} )ib \big)=\sum_{l=0}^n \alpha_l \, \mrm{Res}\big(\sigma, k-l +(m+\textstyle\frac{1}{2})ib \big) \\
&=\sum_{l=0}^n \alpha_{l}r_{k-l } \, \mrm{Res}\big(\tanh(\pi b^{-1}\,\cdot ),(m+\textstyle\frac{1}{2})ib\big)=0,
\end{aligned}
\end{equation*}
 for all $k,m\in\Z$, and therefore $f$ is, in fact, entire.

Now, as $\sigma\in\LtCla{1}{b}$, it follows by Proposition \ref{prop:Zab-SAC} that $\sigma$ satisfies the SAC, and so $f$ must be constant. Let $\zeta\in\R$ be such that $f=\zeta\,\bm{1}$. In order to establish that $\left(\zeta,\{(\alpha_k,1,k)\}_{k=0}^n\right)$ is an affine symmetry of $\sigma$, it remains to show that if $\m{I}\subset\{0,1,\dots, n\}$ is a nonempty set such that  $f_{\m{I}}\coleqq \sum_{k\in\m{I}} \alpha_{k}\,\sigma(\cdot -k)$ is constant, then $\m{I}=\{0,1,\dots, n\}$. Let $\m{I}$ be such a set. Then, as $\alpha_lr_{n-l}=\alpha_l^2>0$, for $l\in\{1,\dots,n\}$, using $\sum_{l=0}^n\alpha_l r_{n-l}=0$, we obtain 
\begin{equation*}
\alpha_0 r_n=-\sum_{l=1}^n \alpha_l r_{n-l}=-\sum_{l=1}^n\alpha_l^2<0.
\end{equation*}
Next, as $f_{\m{I}}$ is constant on $\R$, its analytic continuation does not have poles at any of the points $\Z+\frac{ib}{2}$, and so $\sum_{l\in\m{I}} \alpha_{l}r_{n-l}=0$. Therefore, we must have $0\in\m{I}$, as otherwise we would have $\sum_{l\in\m{I}} \alpha_{l}r_{n-l}>0$, which constitutes a contradiction. Now,
\begin{equation*}
\sum_{l\in \{1,\dots,n\}\setminus \m{I}} \alpha_l^2\;=\sum_{l=0}^n \alpha_l r_{n-l}- \sum_{l\in\m{I}} \alpha_{l}r_{n-l}=0,
\end{equation*}
which implies $\m{I}\supset\{1,\dots, n\}$. Therefore, $\m{I}=\{0,1,\dots, n\}$, as desired.
\end{proof}

\subsection{Proofs of auxiliary results in Section \ref{sec:formal-NNT}}

\begin{proof}[Proof of Proposition \ref{prop:modif-basic-prop}]
\textit{(i)} Let $A$, $B$, and $C$ be sets of nodes such that $\m{N}_2$ is a $(\rho\,;A,B,C)$--modification of $\m{N}_1$ with respect to an affine symmetry $\left(\zeta,\{(\alpha_u,\beta_u,\theta_u)\}_{u\in A\cup B}\cup \{(\alpha'_p,\beta'_p,\gamma'_p)\}_{p=1}^n \right)$ of $\rho$, and adopt the remaining notation of Definition \ref{def:modif}. Suppose that $\m{N}_1$ is layered. Enumerate $C=\{u_1',\dots, u_n'\}$. Then, as $\pre_{\m{N}_2}(u_p')=P$, for $p\in\{1,\dots, n\}$, we have 
\begin{equation*}
\lvl_{\m{N}_2}(u_p')=\lvl_{\m{N}_2}(u)=\lvl_{\m{N}_1}(v)+1=\lvl_{\m{N}_2}(v)+1,
\end{equation*}
for $p\in\{1,\dots, n\}$, $u\in A\cup B$, and $v\in P$. Therefore, for $w\in W$ and $u\in B\cup C$,
\begin{equation*}
\lvl_{\m{N}_2}(w)=\max\left(\{\lvl_{\m{N}_2}(\wtd{u}):\wtd{u}\in B\}\cup \{\lvl(u_p'):p\in\{1,\dots, n\}\}\right)+1=\lvl_{\m{N}_2}(u) +1,
\end{equation*}
hence $\m{N}_2$ is layered.

\noindent \textit{(ii)} We show that $\m{N}_1$ is a $(\rho\,;C,B,A)$--modification of $\m{N}_2$. To this end, first note that we have
\begin{equation*}
\left[\sum_{p=1}^{n}\alpha_p' \,\rho(\beta_p' t +\gamma_p')+\sum_{u\in B}\alpha_u\,\rho(\beta_u t+\theta_{u})\right]+\sum_{u\in A}\alpha_u\,\rho(\beta_u t+\theta_{u})=\zeta\,\bm{1}(t),\quad t\in\R,
\end{equation*}
so Condition (i) of Definition \ref{def:modif} is satisfied for the putative $(\rho\,;C,B,A)$--modification. Moreover, Conditions (ii)--(iv) are satisfied with the same sets $\{\kappa_v\}_{v\in P}$, $\{\nu_w\}_{w\in W}$, and $\{\mu_r\}_{r=1}^D$, so $\m{N}_2$ admits a $(\rho\,;C,B,A)$--modification. It is now a routine check to verify that the $(\rho\,;C,B,A)$--modification of $\m{N}_2$ is, in fact, $\m{N}_1$, as desired.

\noindent \textit{(iii)} Consider a node $w\in V^2$. If $w\in V^2$ and $\anc_{\m{N}_2}(\{w\})\cap C=\varnothing$, then $\OTnoA{w}{\rho}{\m{N}_2}=\OTnoA{w}{\rho}{\m{N}_1}$, simply as the part of $\m{N}_1$ contributing to the map of $w$ is unaffected by the $\rho$-modification. Suppose now that $w\in W$ and write $B=B_1\cup B_2\cup B_3$, where
\begin{equation*}
\begin{aligned}
B_1&=\{u\in B: (u,w)\notin E\},\\
B_2&=\{u\in B\setminus B_1: \omega_{wu}-\nu_w\alpha_u= 0\},\text{ and}\\
B_3&=\{u\in B\setminus B_1: \omega_{wu}-\nu_w\alpha_u\neq 0\}.
\end{aligned}
\end{equation*}
Then $\pre_{\m{N}_1}(w)=A\cup B_2\cup B_3$, $\pre_{\m{N}_2}(w)= B_1\cup B_3 \cup C$, and, denoting $K_P(t)=\sum_{v\in P}\kappa_v\OTnoA{v}{\rho}{\m{N}}(t)$, we can compute
\begin{align*}
&\sum_{u\in \pre_{\m{N}_2}\!(w)\cap(B\cup C) }\omega_{wu}^{\, \m{N}_2}\;\OTnoA{u}{\rho}{\m{N}_2}(t)\;+\theta_w^{\,\m{N}_2}\\[5mm]
=&\, \sum_{u\in B_1\cup B_3\cup C }\omega_{wu}^{\, \m{N}_2}\;\OTnoA{u}{\rho}{\m{N}_2}(t)\;+\theta_w^{\,\m{N}_2} \quad +\quad \sum_{\mathclap{u\in B_2 }}(\omega_{wu}-\alpha_u\nu_w) \rho\big(\beta_u K_P(t)+\theta_u \big)\\[5mm]
=&\,\sum_{\mathclap{u_p\in C }}-\alpha_p'\nu_w \rho\big(\beta_p' K_P(t)+\gamma_p' \big)\quad+\quad \sum_{\mathclap{u\in B_1 }}(-\alpha_u\nu_w) \rho\big(\beta_u K_P(t)+\theta_u \big)   \\
&\hspace{6cm}+\quad \sum_{\mathclap{u\in B_2\cup B_3 }}(\omega_{wu}-\alpha_u\nu_w) \rho\big(\beta_u K_P(t)+\theta_u \big) \quad +\theta_w+\zeta\nu_w\\
=&\,\nu_w \Big(\zeta\,\bm{1}-\sum_{p=1}^{n}\alpha_p' \,\rho(\beta_p'\cdot +\gamma_p')-\sum_{u\in B}\alpha_u\,\rho(\beta_u\cdot+\theta_{u})\Big)\left(K_P(t)\right)+ \sum_{\mathclap{u\in B_2\cup B_3 }}\omega_{wu} \rho\big(\beta_u K_P(t)+\theta_u \big)+\theta_w \\[3mm]
=&\,\nu_w\sum_{u\in A}\alpha_u\,\rho(\beta_u K_P(t)+\theta_{u})\quad+ \sum_{{u\in B\cap \pre_{\m{N}_1}\!(w) }}\omega_{wu}\, \rho\big(\beta_u K_P(t)+\theta_u \big) \quad+\;\theta_w\\
=&\,\sum_{u\in A}\omega_{wu}\,\rho(\beta_u K_P(t)+\theta_{u})\quad+ \sum_{{u\in B\cap \pre_{\m{N}_1}\!(w) }}\omega_{wu}\, \rho\big(\beta_u K_P(t)+\theta_u \big) \quad +\; \theta_w, \qquad\text{for } t\in\R,
\end{align*}
and, as $\pre_{\m{N}_2}\!(w)\setminus(B\cup C)=\pre_{\m{N}_1}\!(w)\setminus(A\cup B)$, we consequently have
\begin{align*}
&\OTnoA{w}{\rho}{\m{N}_2}(t)\\
=\,& \rho\Bigg(\sum_{u\in \pre_{\m{N}_2}\!(w) }\omega_{wu}^{\,\m{N}_2}\,\OTnoA{u}{\rho}{\m{N}_2}(t)\;+\theta_w^{\,\m{N}_2}\Bigg)\\
=\,&\rho\Bigg(\qquad\quad\sum_{\mathclap{u\in \pre_{\m{N}_2}\!(w)\cap(B\cup C)} }\quad\omega_{wu}^{\,\m{N}_2}\,\OTnoA{u}{\rho}{\m{N}_2}(t)\;+\theta_w^{\,\m{N}_2}\quad+\quad\sum_{\mathclap{u\in \pre_{\m{N}_2}\!(w)\setminus(B\cup C) }}\quad\omega_{wu}\OTnoA{u}{\rho}{\m{N}_2}(t) \Bigg)\\
=\,&\rho\Bigg( \sum_{u\in A}\omega_{wu}\,\rho(\beta_u K_P(t)+\theta_{u})+ \sum_{\mathclap{u\in B\cap \pre_{\m{N}_1}\!(w) }}\omega_{wu} \rho\big(\beta_u K_P(t)+\theta_u \big) +\theta_w+\;\;\sum_{\mathclap{u\in \pre_{\m{N}_1}\!(w)\setminus(A\cup B) }}\;\;\omega_{wu}\OTnoA{u}{\rho}{\m{N}_2}(t)  \Bigg)\\
=\,&\OTnoA{w}{\rho}{\m{N}_1}(t), \qquad\text{for } t\in\R.
\end{align*}

\noindent Thus, $\OTnoA{w}{\rho}{\m{N}_2}=\OTnoA{w}{\rho}{\m{N}_1}$, for all $w\in W$, and it then immediately follows that $\OTnoA{w}{\rho}{\m{N}_2}=\OTnoA{w}{\rho}{\m{N}_1}$, for all $w\in V^2\setminus C$. Now, if $A\cap \Vout^1=\varnothing$, then $C\cap \Vout^2=\varnothing$ and $\OnoA{\m{N}_1}{\rho}=\OnoA{\m{N}_2}{\rho}$, as desired. If, on the contrary, $A\subset \Vout^1$, then $C\subset \Vout^2$, and the $r$-th coordinate of $\OnoA{\m{N}_2}{\rho}$ is given by
\begin{equation*}
\begin{aligned}
\big(\OnoA{\m{N}_2}{\rho}\big)_r\!&=\lambda^{(r)}+\zeta\,\mu_r+\sum_{u\in C}-\alpha_p'\mu_r\outmap{u}{\rho}+\sum_{u\in B_{\mrm{out}}^{(r)}}(\lambda^{(r)}_{u}-\alpha_u\mu_r)\outmap{u}{\rho}+\sum_{u\in \Vout^2\setminus (B\cup C)}\lambda^{(r)}_{u}\outmap{u}{\rho}\\
&=\lambda^{(r)}+\mu_r\Big(\zeta-\sum_{u\in C}\alpha_p'\outmap{u}{\rho}-\sum_{u\in B}\alpha_u\outmap{u}{\rho}\Big)+\sum_{u\in \Vout^2\setminus C}\lambda^{(r)}_{u}\outmap{u}{\rho}\\
&=\lambda^{(r)}+\mu_r \sum_{u\in A}\alpha_u\outmap{u}{\rho}+\sum_{u\in \Vout^1\setminus A}\lambda^{(r)}_{u}\outmap{u}{\rho}\\
&=\lambda^{(r)}+ \sum_{u\in A}\lambda_u^{(r)}\outmap{u}{\rho}+\sum_{u\in \Vout^1\setminus A}\lambda^{(r)}_{u}\outmap{u}{\rho}\\
&=\big(\OnoA{\m{N}_1}{\rho}\big)_r,
\end{aligned}
\end{equation*}
for $r\in\{1,\dots, D\}$, concluding the proof.
\end{proof}

\begin{proof}[Proof of Lemma \ref{lem:modif-irred}]
Let $\m{N}'$ be the $(\rho\,;A,B,C)$--modification of $\m{N}$ with respect to the affine symmetry
\begin{equation}\label{eq:modif-irred-sym-1}
\sum_{u\in A\cup B}\alpha_u\,\rho(\beta_u t+\theta_{u})+\sum_{p=1}^n\alpha_p' \,\rho(\beta_p' t +\gamma_p')=\zeta\,\bm{1}(t),\quad t\in\R,
\end{equation}
and adopt the remaining notation of Definition \ref{def:modif}. By definition of $\rho$-modification, there exist nonzero real numbers $\{\kappa_v\}_{v\in P}$ such that $\{\omega_{u_p'v}\}_{v\in P}=\beta_p'\{\kappa_v\}_{v\in P}$, for $p\in\{1,\dots, n\}$, and $\{\omega_{uv}\}_{v\in P}=\beta_u\{\kappa_v\}_{v\in P}$, for all $u\in A\cup B$.

Suppose by way of contradiction that $\m{N}'$ is $(\rho,U)$--reducible, for some set of nodes $U$ with common parent set $P_U$. Then, as $\m{N}$ itself is irreducible, we must have
\begin{equation}\label{eq:UcapC}
 \text{either $U\cap C\neq\varnothing $ or $P_U\cap C\neq \varnothing$.}
\end{equation}
 Suppose first that $C_U\coleqq U\cap C\neq \varnothing $, and let $D=U\setminus C$. It follows by definition of reducibility that the parent set of all nodes in $U=C_U\cup D$ is $P_U$, and, in particular, as $U\cap C\neq\varnothing$, we have $P_U=P$. Moreover, there exist nonzero real numbers $\{\wtd{\beta}_u\}_{u\in C_U\cup D}$ and $\{\wtd{\kappa}_v\}_{v\in P}$ such that $\{\omega_{uv}'\}_{v\in P}=\wtd{\beta}_u\{\wtd{\kappa}_v\}_{v\in P}$, for all $u\in C_U\cup D$, as well as nonzero real numbers $\{\wtd{\alpha}_u\}_{u\in C_U\cup D}$ and $\wtd{\zeta}\in \R$ such that
\begin{equation}\label{eq:modif-irred-sym-2}
\sum_{u_p'\in C_U}\wtd{\alpha}_{u_p'}\,\rho(\wtd{\beta}_{u_p'} t+\gamma_p')+ \sum_{u\in D}\wtd{\alpha}_u\,\rho(\wtd{\beta}_u t+\theta_{u})=\wtd{\zeta}\,\bm{1}(t),\quad t\in\R.
\end{equation}
Specifically, we have 
\begin{equation*}
\begin{aligned}
\beta_p'\{\kappa_v\}_{v\in P}&=\{\omega_{u_p'v}'\}_{v\in P}=\wtd{\beta}_{u_p'}\{\wtd{\kappa}_v\}_{v\in P}, &&\text{ for }u_p'\in C_U,\text{ and}\\
\beta_u\{\kappa_v\}_{v\in P}&=\{\omega_{uv}'\}_{v\in P}=\wtd{\beta}_u\{\wtd{\kappa}_v\}_{v\in P},&&\text{ for }u\in B\cap D.
\end{aligned}
\end{equation*}
 Fix an arbitrary $u_{p^*}'\in C_U$ and let $\tau=\beta_{p^*}'/\wtd{\beta}_{u_{p^*}'}$. Now, by replacing $\wtd{\beta}_u$ by $\tau\wtd{\beta}_u$, for $u\in C_U\cup D$, and $\{\wtd{\kappa}_v\}_{v\in P}$ by $\{\tau^{-1}\wtd{\kappa}_v\}_{v\in P}$, we may assume w.l.o.g. that 
\begin{equation*}
\begin{aligned}
\wtd{\beta}_{u_p'}&=\beta_p',\quad\{\wtd{\kappa}_v\}_{v\in P}=\{{\kappa}_v\}_{v\in P} , &&\text{ for }u_p'\in C_U,\text{ and}\\
\wtd{\beta}_u&=\beta_u, &&\text{ for }u\in B\cap D.
\end{aligned}
\end{equation*}
Similarly, by replacing the $\wtd{\alpha}_{u}$ with $\wtd{\alpha}_{u}/\wtd{\alpha}_{u_{p^*}'}$, for $u\in C_U\cup D$, and $\wtd{\zeta}$ by  $\wtd{\zeta}/\wtd{\alpha}_{u_{p^*}'}$, we may assume w.l.o.g. that $\wtd{\alpha}_{u_{p^*}'}=1$.
With this, \eqref{eq:modif-irred-sym-2} reads
\begin{equation}\label{eq:modif-irred-sym-3}
\rho(\beta_{u_{p^*}'} t+\gamma_{p^*}')+\hspace*{-3mm}\sum_{u_p'\in C_U\setminus\{u_{p^*}'\}} \hspace*{-2mm}{\textstyle \wtd{\alpha}_{u_{p}'}} \rho(\beta_{u_{p}'} t+\gamma_{p}') + \sum_{u\in B}{\textstyle \wtd{\alpha}_u}\,\rho(\beta_u t+\theta_{u}) +\! \sum_{u\in D\setminus B}{\textstyle \wtd{\alpha}_u}\,\rho(\wtd{\beta}_u t+\theta_{u})={\textstyle\wtd{\zeta}}\bm{1}(t),
\end{equation}
for $t\in\R$. Combining \eqref{eq:modif-irred-sym-1} and \eqref{eq:modif-irred-sym-3} now yields
\begin{equation}\label{eq:modif-irred-sym-4}
\begin{aligned}
&\sum_{u\in A}\alpha_u\rho(\beta_u t+\theta_{u})+\sum_{u\in B}\left(\alpha_u-{\textstyle \alpha_{p^*}' \wtd{\alpha}_u}\right)\,\rho(\beta_u t+\theta_{u})+ \sum_{u\in D\setminus B}\!\left(-{\textstyle \alpha_{p^*}' \wtd{\alpha}_u}\right)\,\rho(\wtd{\beta}_u t+\theta_{u}) \\
&\qquad+ \sum_{u_p'\in C_U\setminus\{u_{p^*}'\}}\hspace*{-2mm}\left(\alpha_p' -{\textstyle {\alpha_{p^*}' \wtd{\alpha}_{u_{p}'}}}\right) \rho(\beta_{u_{p}'} t+\gamma_{p}') 
+ \hspace*{-2mm} \sum_{u_p'\in C\setminus C_U} \alpha_p' \,\rho(\beta_p' t +\gamma_p') =\left(\zeta- {\textstyle \alpha_{p^*}' \wtd{\zeta}} \right)\bm{1}(t),
\end{aligned}
\end{equation}
for $t\in\R$. Now let $\wtd{C}=(C\setminus C_U)\cup\{u_p'\in C_U\setminus\{u_{p^*}'\}: \alpha_p' - {\alpha_{p^*}' \wtd{\alpha}_{u_{p}'}}\neq 0 \}$, and note that $\wtd{C}\neq \varnothing$. Indeed, suppose by way of contradiction that $\wtd{C}= \varnothing $. Then $C_U=C$ and $\alpha_p' - {\alpha_{p^*}' \wtd{\alpha}_{u_{p}'}}= 0$, for all $p\in\{1,\dots, n\}\setminus\{p^*\}$, and so \eqref{eq:modif-irred-sym-4} reduces to
\begin{equation*}
\sum_{u\in A}\alpha_u\rho(\beta_u t+\theta_{u})+\sum_{u\in B}\left(\alpha_u-{\textstyle \alpha_{p^*}' \wtd{\alpha}_u}\right)\,\rho(\beta_u t+\theta_{u})+\hspace*{-2mm} \sum_{u\in D\setminus B}\!\left(-{\textstyle {\alpha_{p^*}' \wtd{\alpha}_u}}\right)\,\rho(\wtd{\beta}_u t+\theta_{u})  =\left(\zeta- {\textstyle \alpha_{p^*}' \wtd{\zeta}} \right)\!\bm{1}(t).
\end{equation*}
Lemma \ref{lem:LD->reduc} now implies that $\m{N}$ is $(\rho,D')$--reducible, for some $D'\subset A\cup B\cup D$, which contradicts the assumption that $\m{N}$ is irreducible and thus establishes $\wtd{C}\neq \varnothing$. It now follows from \eqref{eq:modif-irred-sym-4} that $\m{N}$ admits a $(\rho\,;A,\wtd{B}\cup (D\setminus B),\wtd{C})$--modification, where $\wtd{B}=\{u\in B:\alpha_u- {\alpha_{p^*}' \wtd{\alpha}_u}\neq 0\}$. But, as $u_{p^*}'\notin \wtd{C}$, we have $\wtd{C}\subsetneq C$, contradicting the assumption that $C$ is a subset of $C_0$ of least possible cardinality such that $\m{N}$ admits a corresponding $\rho$-modification. This establishes that $U\cap C=\varnothing$.

Recalling \eqref{eq:UcapC}, we deduce that we must have $P_U\cap C\neq \varnothing$, which further implies $U\subset  W$ and $C\subset P_U$. Next, by definition of $\rho$-modification, there exist nonzero real numbers $\{\nu_w\}_{w\in U}$ such that $\{\omega_{wu}\}_{u\in A}=\nu_w\{\alpha_u\}_{u\in A}$, for all $w\in U$. We now write $B=B_1\cup B_2\cup B_3$, where
\begin{equation*}
\begin{aligned}
B_1&=\{u\in B: (u,w)\notin E,\text{ for all }w\in U\},\\
B_2&=\{u\in B\setminus B_1: \omega_{wu}-\nu_w\alpha_u= 0,\text{ for all }w\in U\},\text{ and}\\
B_3&=\{u\in B\setminus B_1: \omega_{wu}-\nu_w\alpha_u\neq 0,\text{ for some }w\in U\}.
\end{aligned}
\end{equation*}
Now, by definition of reducibility, there exists a set of nodes $D\subset V\setminus (A\cup B)$ such that $P_{U}^*\coleqq A\cup B_2\cup B_3\cup D=\pre_{\m{N}}(w)$ and $\pre_{\m{N}'}(w)=P_U=B_1\cup B_3\cup C\cup D$, for $w\in U$, with the pertinent weights and biases of $\m{N}'$ given by
\begin{equation*}
\begin{aligned}
\omega'_{wu}&=-\nu_w\alpha_u, && u\in B_1, & && \qquad\omega'_{wu}&=\omega_{wu}-\nu_w\alpha_u, && u\in B_3,\\
\omega'_{wu_p'}&=-\nu_w\alpha_p',  && u_{p}'\in C, & && \qquad \omega'_{wu}&=\omega_{wu},  && u\in D,
\end{aligned}
\end{equation*}
and $\theta_w'=\theta_w+\zeta\nu_w$, for all $w\in U$.
Furthermore, there exist nonzero real numbers $\{\wtd{\beta}_w\}_{w\in U}$ and  $\{\wtd{\kappa}_u\}_{u\in P_U}$ such that $\{\omega_{wu}'\}_{u\in P_U}=\wtd{\beta}_w\{\wtd{\kappa}_u\}_{u\in P_U}$, for $w\in U$, as well as nonzero real numbers $\{\wtd{\alpha}_w\}_{w\in  U}$ and $\wtd{\zeta}\in \R$ so that
\begin{equation}\label{eq:modif-irred-sym-5}
\sum_{w\in U}\wtd{\alpha}_{w}\,\rho(\wtd{\beta}_{w} t+\theta_w' )=\wtd{\zeta}\,\bm{1}(t),\quad t\in\R.
\end{equation} 
Now, as $C\subset P_U$, we have
\begin{equation*}
-{\nu_{w}^{-1}}{\wtd{\beta}_{w}}\, \{\kappa_{u_p'}\}_{u_p'\in C}= -\wtd{\nu}_{u_p'}^{-1}\, \{\omega_{wu_p'}'\}_{u_p'\in C}=\{\alpha_p'\}_{u_p'\in C},\quad \text{for all }w\in U,
\end{equation*}
and thus 
\begin{equation*}
\left(\nu_{w_1}^{-1}{\wtd{\beta}_{w_1}}-\nu_{w_2}^{-1}{\wtd{\beta}_{w_2}} \right) \{\kappa_{u_p'}\}_{u_p'\in C}=\bm{0},\quad \text{for all }w_1,w_2\in U,
\end{equation*}
which implies the existence of a $\tau\in \R\setminus\{0\}$ such that $\nu_w=\tau\wtd{\beta}_w$, for all $w\in U$.
Next, for $w\in U$, define the function
\begin{equation*}
\begin{aligned}
K_w&\coleqq   \sum_{u\in B_3} \wtd{\beta}_w^{-1} \omega_{wu} \, \rho(\beta_u \cdot\, +\theta_u) \\
&=  \sum_{u\in B_3} \wtd{\beta}_w^{-1} \left(\omega_{wu}' +  \nu_w \alpha_u \right)  \, \rho(\beta_u \cdot\, +\theta_u) \\
&=\sum_{u\in B_3} \left(\wtd{\kappa}_{u} +\tau\alpha_u \right) \, \rho(\beta_u \cdot \,+\theta_u) .
\end{aligned}
\end{equation*}
We observe from the last expression that $K_{w_1}=K_{w_2}$, for all $w_1,w_2\in U$. Therefore,
\begin{equation*}
\sum_{u\in B_3}(\wtd{\beta}_{w_1}^{-1} \omega_{w_1u}-\wtd{\beta}_{w_2}^{-1} \omega_{w_2u} )\; \rho(\beta_u \cdot\, +\theta_u)=K_{w_2}-K_{w_1}=0,
\end{equation*}
for all $w_1,w_2\in U$, $w_1\neq w_2$, and we hence must have $\wtd{\beta}_{w_1}^{-1} \omega_{w_1u}-\wtd{\beta}_{w_2}^{-1} \omega_{w_2u}=0$,  for all $w_1,w_2\in U$, $w_1\neq w_2$, and $u\in B_3$, as otherwise Lemma \ref{lem:LD->reduc}  would  imply $(\rho,B')$--reducibility of $\m{N}$, for some $B'\subset B_3$. Therefore, for every $u\in B_3$, there exists a $\tau_u\in\R\setminus \{0\}$ such that $\omega_{wu}=\tau_u\wtd{\beta}_w$, for $w\in U$. Summarizing, we have
\begin{equation*}
\omega_{wu}=
\begin{cases}
\nu_w\alpha_u=\tau\alpha_u\,\wtd{\beta}_w, & u\in A\cup B_2\\
\tau_u\,\wtd{\beta}_w, &u\in B_3\\
\omega'_{wu}=\wtd{\kappa}_u\, \wtd{\beta}_w, & u\in D
\end{cases},\quad\text{for }w\in U. 
\end{equation*}
We have hence established the existence of nonzero real numbers $\{\kappa_u^*\}_{u\in P_{U}^*}$ such that $\{\omega_{wu}\}_{u\in P_{U}^*}=\wtd{\beta}_{w}\{\kappa_u^*\}_{u\in P_{U}^*}$, which together with \eqref{eq:modif-irred-sym-5} implies that $\m{N}$ is $(\rho,U)$--reducible. This again contradicts the assumption that $\m{N}$ is irreducible, and concludes the proof of the lemma.
\end{proof}

\begin{proof}[Proof of Lemma \ref{lem:modif-nondeg}]
Let $\mathscr{B}$ be the set of all $B'\subset B$ such that $\m{N}$ admits a $(\rho\,;A\cup B',B\setminus B',C)$--modification. We have $\mathscr{B}\neq\varnothing$, as $\varnothing\in\mathscr{B}$ by assumption. Let $B^*$ be a maximal element of $\mathscr{B}$ with respect to set inclusion, and let $\m{N}'$ be the $(\rho\,;A\cup B^*,B\setminus B^*,C)$--modification of $\m{N}$ with respect to the affine symmetry
\begin{equation}\label{eq:modif-nondeg-1}
\left(\zeta,\{(\alpha_u,\beta_u,\theta_u)\}_{u\in (A\cup B^*)\cup(B\setminus B^*)}\cup \{(\alpha'_p,\beta'_p,\gamma'_p)\}_{p=1}^n \right)
\end{equation}
 of $\rho$, and let $\{\kappa_v\}_{v\in P}$, $\{\nu_w\}_{w\in W}$, and $\{\mu_r\}_{r=1}^D$ be as in Definition \ref{def:modif}. We now show that $\m{N}'$ is non-degenerate. Assume by way of contradiction that $\m{N}'$ is degenerate and let $u^*\in B\setminus B^*$ be such that
\begin{itemize}[--]
\item $\{w\in V: (u^*,w)\in E\}= W$ and $\omega_{wu^*}-\alpha_{u^*}\nu_w=0$, for all $w\in W$, and 
\item either
\begin{itemize}[]
\item[(a)] $u^*\notin \Vout$ and $A\cap \Vout=\varnothing$, or
\item[(b)] $u^*\in \Vout$, $A\subset \Vout$, and $\lambda_{u^*}^{(r)}-\alpha_{u^*}\mu_r=0$, for all $r\in\{1,\dots,D\}$. 
\end{itemize}
\end{itemize}
We claim that then $\m{N}$ admits a $(\rho\,;A\cup B^*\cup\{u^*\},B\setminus (B^*\cup\{u^*\}),C)$--modification. Indeed, as $ (A\cup B^*)\cup(B\setminus B^*)=(A\cup B^*\cup\{u^*\})\cup\left(B\setminus (B^*\cup\{u^*\})\right)$,
Condition (i) of Definition \ref{def:modif} is satisfied by the same affine symmetry \eqref{eq:modif-nondeg-1}.
Moreover, $\omega_{wu^*}=\alpha_{u^*}\nu_w$, for all $w\in W$, and, in the circumstance (b) above, $\lambda_{u^*}^{(r)}=\alpha_{u^*}\mu_r$, for all $r\in\{1,\dots,D\}$, and so Conditions (ii)--(iv) of Definition \ref{def:modif} are satisfied with the same sets $\{\kappa_v\}_{v\in P}$, $\{\nu_w\}_{w\in W}$, and $\{\mu_r\}_{r=1}^D$. Therefore, $B^*\cup\{u^*\}\in\mathscr{B}$, contradicting the maximality of $B^*$ and thus completing the proof of the lemma.
\end{proof}

\subsection{Proofs of auxiliary results in Section \ref{sec:Clustering}}

\begin{proof}[Proof of Lemma \ref{lem:cluster-prop}]
Items (i) and (ii) are elementary facts from the analysis of metric spaces.

\noindent {(iii)} Suppose to the contrary that $z$ is neither an element of $E$ nor a cluster point of $E$. Then there exists an $\varepsilon>0$ such that $D^{\circ}(z,\varepsilon)\cap E=\varnothing$. This then implies $L_{\m{C}}(E,z)=0$, contradicting the assumption $L_{\m{C}}(E,z)\geq 1$.

\noindent {(iv)} We use induction on $k$. For the base case, we have $\m{C}^{0}(E)=E\subset F=\m{C}^{0}(F)$ by definition. For the induction step, suppose that $k\geq 1$ and $\m{C}^{k-1}(E)\subset \m{C}^{k-1}(F)$. Then, as every cluster point of $\m{C}^{k-1}(E)$ is a cluster point of $\m{C}^{k-1}(F)$, we obtain $\m{C}^{k}(E)\subset \m{C}^{k}(F)$, as desired.

\noindent {(v)} We again proceed by induction on $k$, starting with the base case $k=1$ (the case $k=0$ is clear). First, as every cluster point of $E$ is a cluster point of $E\cup F$, we have $\m{C}^1(E\cup F)\supset \m{C}^{1}(E)$. Similarly, $\m{C}^1(E\cup F)\supset \m{C}^{1}(F)$, and so $\m{C}^1(E\cup F)\supset \m{C}^{1}(E)\cup \m{C}^1(F)$. For the reverse inclusion, suppose that $z$ is neither a cluster point of $E$ nor $F$, i.e., there exists an $\varepsilon>0$ such that 
\begin{equation*}
(D^{\circ}(z,\varepsilon)\setminus\{z\})\cap E=\varnothing\quad\text{ and }(D^{\circ}(z,\varepsilon)\setminus\{z\})\cap F=\varnothing.
\end{equation*}
Then  $(D^{\circ}(z,\varepsilon)\setminus\{z\})\cap (E\cup F)=\varnothing$, and so $z$ is not a cluster point of $E\cup F$. Therefore, every cluster point of $E\cup F$ must be a cluster point of at least one of $E$ or $F$, establishing $\m{C}^1(E\cup F)=\m{C}^1(E)\cup \m{C}^1(F)$. For the induction step, assume $k\geq 2$ and $\m{C}^{k-1}(E\cup F)=\m{C}^{k-1}(E)\cup \m{C}^{k-1}(F)$. Then, using the identity for the already established base case, we have
\begin{equation*}
\begin{aligned}
\m{C}^{k}(E\cup F)=\m{C}^1\big(\m{C}^{k-1}(E\cup F)\big)&=\m{C}^1\big(\m{C}^{k-1}(E)\cup \m{C}^{k-1}(F)\big)\\
&=\m{C}^1\big(\m{C}^{k-1}(E)\big)\cup\m{C}^1\big( \m{C}^{k-1}(F)\big)=\m{C}^k(E)\cup \m{C}^k(F),
\end{aligned}
\end{equation*}
as desired.

\noindent{(vi)} Letting $k=L_{\m{C}}(E\cup F)$, we have $\varnothing=\m{C}^k(E\cup F)=\m{C}^k(E)\cup \m{C}^k(F)$, and thus both $\m{C}^k(E)$ and $\m{C}^k(F)$ must be empty. Then $L_{\m{C}}(E)\leq k$ and $L_{\m{C}}(F)\leq k$, and thus 
\begin{equation}\label{eq:cluster-lemma-1}
\max\{ L_{\m{C}}(E), L_{\m{C}}(F)\}\leq k=L_{\m{C}}(E\cup F).
\end{equation}
 Next, let $k'=\max\{ L_{\m{C}}(E), L_{\m{C}}(F)\}$. Then $L_{\m{C}}(E)\leq k'$ and $L_{\m{C}}(F)\leq k'$, and so both $\m{C}^{k'}(E)$ and $\m{C}^{k'}(F)$ are empty. Thus, $\m{C}^{k'}(E\cup F)=\m{C}^{k'}(E)\cup \m{C}^{k'}(F)=\varnothing$, and so
\begin{equation*}
L_{\m{C}}(E\cup F)\leq k'=\max\{ L_{\m{C}}(E), L_{\m{C}}(F)\},
\end{equation*}
which together with \eqref{eq:cluster-lemma-1} implies the desired identity.
\end{proof}

\begin{proof}[Proof of Lemma \ref{lem:nat-dom}]
The function $g\circ f$ can clearly be analytically continued to $\m{D}$, so it remains to show that $\m{D}$ has countable complement in $\C$. To this end, let $E_f= \C\setminus\dom_{f}$ and $E=\C\setminus\dom$. We claim that if $z^*$ is a cluster point of $E\cap \dom_f$, then $z^*\in E_f$. Suppose by way of contradiction that this is not the case, and let $(z_n)_{n\in\N}$ be a sequence of distinct elements of $E\cap \dom_f$ such that $z_n\to z^*$, for some $z^*\in \dom_f$. Now, as $f$ is holomorphic, it is, in particular, continuous on $\dom_f$, and therefore $f(z_n)\to f(z^*)$ as $n\to\infty$. On the other hand, we have $f(z_n)\in P$, by definition of  $E$, and as $P$ is discrete, we deduce that there exists a $p^*\in P$ such that $f(z_n)=p^*$ for all sufficiently large $n\in\N$. Now, as $E_f$ is closed and countable by assumption, we have that $\dom_f$ is connected, and therefore it follows by the identity theorem that $f(z)=p^*$, for all $z\in \dom_f$. But this contradicts the assumption that $f$ is non-constant, and thus completes the proof that any cluster point of $E\cap \dom_f$ is contained in $E_f$.

Now define the compact sets 
$
E^N:=\{z\in E:|z|\leq N,\; d(z,E_f)\geq 1/N\},\text{ for }N\in\N
$.
We see that $E^N$ is finite, for each $N\in\N$, for otherwise there would exist a sequence $(z_n)_{n\in\N}$ of distinct elements of $E^N$ converging to a point $z^*\in\C$. But then, by the claim above, we would have $z^*\in E_f$, contradicting $d(z_n,E_f)\geq 1/N$, for all $n\in \N$. We deduce that $E=E_f\cup \bigcup_{N\in\N}E^N$ is a closed countable set, as desired.
\end{proof}

\begin{proof}[Proof of Lemma \ref{lem:LD->reduc}]
Let $\mathscr{I}$ be the set of all $\m{I}'\subset \m{J}$ such that $j^*\in \m{I}'$, and there exist real numbers $\{\wtd{\alpha}_s\}_{s\in \m{I}'}$ such that $\wtd{\alpha}_{j^*}\neq 0$ and $\sum_{s\in\m{I}'}\wtd{\alpha}_s\rho(\beta_s\cdot\,+\gamma_s)$ is constant. Note that $\m{J}\in\mathscr{I}$ by assumption. Let $\m{I}$ be a minimal element of $\mathscr{I}$ with respect to set inclusion. We then have $\sum_{s\in\m{I}}\wtd{\alpha}_s\rho(\beta_s\cdot\,+\gamma_s)=\zeta\,\bm{1}$, for some $\zeta\in\R$, so in order to show that $\left(\zeta,\{(\wtd{\alpha}_s,\beta_s,\gamma_s)\}_{s\in\m{I}}\right)$ is an affine symmetry of $\rho$, it suffices to establish that there does not exist an $\m{I}'\subsetneq \m{I}$ such that $\{\rho(\beta_s\cdot\,+\gamma_s):s\in\m{I}'\}\cup\{\bm{1}\}$ is linearly dependent. Suppose by way of contradiction that such an $\m{I}'$ exists. Assume for now that $j^*\in \m{I}'$, and let $\alpha'_{s}\in\R$, for $s\in\m{I}'$, be such that $\sum_{s\in \m{I}'}\alpha'_s\rho(\beta_s\cdot\,+\gamma_s)$ is constant. Then we must have $\alpha'_{j^*}=0$, for otherwise we would have $\m{I}'\in\mathscr{I}$, contradicting the minimality of $\m{I}$. Therefore, $\sum_{s\in \m{I}'\setminus \{j^*\}}\alpha'_s\rho(\beta_s\cdot\,+\gamma_s)$ is constant, so we may w.l.o.g. assume $j^*\notin\m{I}'$ by replacing $\m{I}'$ with $ \m{I}'\setminus \{j^*\}$ if necessary. Now, there exist $s^*\in\m{I}'$, $\xi\in\R$, and $\delta_s\in\R$, for $s\in\m{I}'\setminus\{s^*\}$, such that $\rho(\beta_{s^*}\cdot\,+\gamma_{s^*})=\xi\,\bm{1} +\sum_{s\in \m{I}'\setminus\{s^*\}}\delta_s\rho(\beta_s\cdot\,+\gamma_s)$. Thus,
\begin{equation*}
\wtd{\alpha}_{j^*}\rho(\beta_{j^*}\cdot\,+\,\gamma_{j^*}) + \hspace*{-3mm} \sum_{s\in\m{I}\setminus(\{j^*\}\cup  \m{I}') } \hspace*{-3mm} \wtd{\alpha}_s\rho(\beta_s\cdot\,+\,\gamma_s) + \sum_{s\in\m{I}'\setminus\{s^*\}}(\wtd{\alpha}_s+\wtd{\alpha}_{s^*}\delta_s)\rho(\beta_s\cdot\,+\,\gamma_s)=(\zeta-\wtd{\alpha}_{s^*}\xi) \,\bm{1},
\end{equation*}
and therefore $\m{I}\setminus\{s^*\}\in\mathscr{I}$, which again contradicts the minimality of $\m{I}$ and concludes the proof.
\end{proof}

\subsection{Proof of Lemma \ref{lem:orderly}}

\begin{proof}
Let $\{\beta_s\}_{s\in\m{I}}$ and $\Gamma$ be as in the statement of the lemma, and fix a  $\gamma=(\gamma_s)_{s\in\m{I}}\in \Gamma$. 
Now, let $P_\sigma$ be the set of poles of $\sigma$, and let $P_s=\beta_s^{-1}(P_\sigma-\gamma_s)$ be the set of poles of $\sigma(\beta_s\,\cdot\,+\gamma_s)$, for $s\in\m{M}$. We define an undirected graph $\m{G}=(\m{I},\m{E})$ by setting
\begin{equation*}
\m{E}=\left\{(s_1,s_2)\in\m{I}\times\m{I}:s_1\neq s_2,\, P_{s_1}\cap P_{s_2}\neq\varnothing\right\},
\end{equation*}
and claim that $\m{G}$ is connected. Suppose by way of contradiction that $\m{G}$ is disconnected, and let $\m{I}=\m{I}_1\cup \m{I}_2$ be a partition of $\m{I}$ into nonempty subsets that are not mutually connected.
By definition of $\Gamma$, there exist $\zeta\in \R$ and nonzero real numbers $\{\alpha_s\}_{s\in\m{I}}$ such that $\left(\zeta,\{(\alpha_s,\beta_s,\gamma_s)\}_{s\in\m{I}}\right)$ is an affine symmetry of $\sigma$. Now, for $j\in\{1,2\}$, let $f_j=\sum_{s\in\m{I}_j}\alpha_s\sigma(\beta_s\,\cdot\,+\gamma_s)$, and note that $f_j$ is meromorphic and its poles are contained in $A_j\coleqq \bigcup_{s\in\m{I}_j} P_{s}$. Moreover, $A_1\cap A_2=\varnothing$ by the choice of $\m{I}_1$ and $\m{I}_2$. Thus, as $f\coleqq f_1+f_2=\zeta\,\bm{1}$ is constant, it follows that $f_1$ must be entire, for otherwise $f$ would have poles. It hence follows by the SAC for $\sigma$ that $f_1$ must, in fact, be constant. But this violates condition (ii) of Definition \ref{def:sym}, so we have reached the desired contradiction, establishing that $\m{G}$ is connected.

Next, for every $t\in\R$, we have that
\begin{equation*}
\sum_{s\in\m{I}}\alpha_s\sigma\big(\beta_s\,\cdot\,+(\gamma_s-t\beta_s)\big)=\sum_{s\in\m{I}}\alpha_s\sigma\big(\beta_s(\,\cdot\,-t)+\gamma_s\big)=f(\,\cdot-t)
\end{equation*}
is constant, and so $(\gamma_s -  t\beta_s)_{s\in\m{I}}\in\Gamma$. Therefore, fixing an arbitrary $s_0\in\m{I}$, we can write
$
\gamma= (\gamma'_s + \beta_{s_0}^{-1} \gamma_{s_0} \cdot  \beta_s)_{s\in\m{I}}
$,
where
$
\gamma'\coleqq (\gamma_s - \beta_{s_0}^{-1}\gamma_{s_0} \cdot \beta_s)_{s\in\m{I}}\in\Gamma
$.
Now, as $\m{G}$ is connected, for every $s\in\m{I}$ there exists a path of vertices $t_0^s=s_0, \,t_1^s,\, t_2^s,\,\dots,\, t_{n_s}^s=s$ of $\m{G}$ leading from $s_0$ to $s$. Then, by definition of $\m{E}$ and $P_s$, there exist poles $p_1^s,\dots,p_{n_s}^s,\tilde{p}_0^s,\dots,\tilde{p}_{n_s-1}^s\in P_\sigma$ such that 
\begin{equation}\label{eq:poles-graph-lem}
\beta_{t_k^s}^{-1}(p_k^s-\gamma_{t_k^s}')=\beta_{t_{k-1}^s}^{-1}(\tilde{p}_{k-1}^s-\gamma_{t_{k-1}^s}'),
\end{equation}
 for all $s\in \m{I}$ and $k\in\{1,\dots,n_s\}$. Further, observing that $\gamma'_{s_0}=0$ and summing \eqref{eq:poles-graph-lem} over $k\in \{1,\dots,n_s\}$, we have
\begin{equation*}
{\beta_s^{-1}}{\gamma_{s}'}=\Re\left({\beta_s^{-1}}{\gamma_{s}'}-{\beta_{s_0}^{-1}}{\gamma_{s_0}'}\right)=\sum_{k=1}^{n_s}\left({\beta_{t_k^s}^{-1}}\Re\left( {p^s_{k}}\right)\,-\,{\beta_{t_{k-1}^s}^{-1}}\Re\left({\tilde{p}^s_{k-1}}\right)\right)\in Q,\quad \text{for }s\in \m{I},
\end{equation*}
where $Q\coleqq \sum_{s_1\in\m{I}}{\beta_{s_1}^{-1}}\Re(P_\sigma)-\sum_{s_2\in\m{I}}\beta_{s_2}^{-1}\Re(P_\sigma)$ is countable. Therefore,
$\gamma'\in\Gamma'\coleqq \Gamma\cap\, \bigtimes_{s\in\m{I}} \beta_{s} Q$, and so
$
\gamma= (\gamma'_s + \beta_{s_0}^{-1} \gamma_{s_0} \cdot  \beta_s)_{s\in\m{I}}\in\{(\gamma'_s+t\beta_s)_{s\in\m{I}}:t\in\R\}
$.
Finally, as $\gamma$ was arbitrary, we deduce that $\Gamma=\bigcup_{\gamma'\in\Gamma'}\{(\gamma'_s+ t\beta_s)_{s\in\m{I}}:t\in\R\}$, and, as $\Gamma'$ is countable, this concludes the proof.
\end{proof}

\subsection{Proofs of auxiliary results in Section \ref{sec:LtCla-a-b}}

\begin{proof}[Proof of Lemma \ref{lem:Zab-props}]
First, note that the function $z\mapsto ({1+e^{-2 \pi z/b}})^{-1}$ is $ib$-periodic, its poles are simple and located at $P_b\coleqq \{ib(k+1/2):k\in\Z\}$, and it has a finite limit as $\Re(z)\to+\infty$. Therefore,
\begin{equation*}
M_1\coleqq \frac{1}{2} \sup_{\substack{z\in \C \\ \Re(z)\geq 0 }} \left| \frac{1\wedge d\left(z,P_b\right) }{1+e^{-2 \pi z/b}}\right|<\infty.
\end{equation*}
Next, note that $1+\tanh\left(\pi b^{-1}z\right)={2e^{\pi z/b}}(e^{\pi z/b}+e^{-\pi z/b})^{-1}$, and so
\begin{equation*}
\left|1+\tanh\left(\pi b^{-1}z\right)\right|\big(1\wedge d\left(z,P_b\right) \big) \leq  2\, \left| \frac{1\wedge d\left(z,P_b\right) }{1+e^{-2 \pi z/b}} \right| \leq M_1,\quad \text{for }\Re(z)\geq 0,
\end{equation*}
and
\begin{equation*}
\left|1+\tanh\left(\pi b^{-1}z\right)\right|\big(1\wedge d\left(z,P_b\right) \big) \leq  2\,  |e^{2 \pi z / b}|  \left| \frac{1\wedge d\left(-z,P_b\right) }{1+e^{2 \pi z/b}} \right| \leq M_1 e^{2 \pi \Re(z) / b} ,\quad \text{for }\Re(z)< 0.
\end{equation*}
We thus deduce that
\begin{equation*}
\left|1+\tanh\left(\pi b^{-1}z\right)\right|\big(1\wedge d\left(z,P_b\right) \big)\leq M_1 (1\wedge e^{2 \pi \Re(z)/b}),\quad \text{for all }z\in\C.
\end{equation*}
Next, let $M_2>0$ be such that $|c_k|\leq M_2 e^{\pi a' |k|/b}$, for all $k\in\Z$.  Then, for $z\in\C\setminus P_\sigma$, we have
\begin{align}
&\sum_{k\geq 1 }|c_k| \big|1+\, \tanh\big(\pi b^{-1}(z - ka)\big)\big|\leq \sum_{k\geq 1 }|c_k| \frac{M_1}{1\wedge d\left(z,P_b+ka\right) }(1\wedge e^{2 \pi \Re(z-ka)/b}) \notag\\
\leq \,&\sum_{k\geq 1 }M_2 e^{\pi a' k/b} \frac{M_1}{1\wedge d\left(z,P_\sigma\right) }(1\wedge e^{2 \pi \Re(z-ka)/b}), \notag\\
\leq \, &\frac{M_1M_2}{1\wedge d\left(z,P_\sigma\right) }\left( \sum_{k\geq \Re(z)/a }e^{\pi a' k/b}e^{2 \pi \Re(z-ka)/b}\quad+\sum_{1\leq k< \Re(z)/a }e^{\pi a' k/b} \right) \notag\\
\leq \, &\frac{M_1M_2}{1\wedge d\left(z,P_\sigma\right) }\left( e^{2 \pi \Re(z)/b}\sum_{k\geq \Re(z)/a }e^{-\pi k(2 a-a')/b}\quad+\frac{e^{\frac{ \pi a'}{b}\frac{\Re(z)}{a}}-1}{e^{\frac{ \pi a'}{b}}-1}\right) \notag\\
\leq \, &\frac{M_1M_2}{1\wedge d\left(z,P_\sigma\right) }\left( e^{2 \pi \Re(z)/b}\frac{e^{-\pi \frac{\Re(z)}{a} (2a-a')/b}}{1-e^{-\pi (2a-a')/b}}\quad+\frac{e^{\pi a'\frac{\Re(z)}{a}/b}-1}{e^{\pi a'/b}-1}\right) \notag\\
\leq \, &\frac{1}{1\wedge d\left(z,P_\sigma\right) }\underbrace{M_1M_2\left(\frac{1}{1-e^{-\pi (2a-a')/b}}\quad+\frac{1}{e^{\pi a'/b}-1}\right)}_{M_3\coleqq } e^{\eta\Re(z)/b}<\infty, \label{eq:Zab-props-1}
\end{align}
where $\eta=\frac{ a'\pi}{a}<\pi$. Next, as
$
-1+\tanh(\pi b^{-1}z)=-(1+\tanh(-\pi b^{-1}z))
$, a derivation completely analogous to the above yields
\begin{equation}\label{eq:Zab-props-2}
\sum_{k\leq 0 }|c_k| \big|-1+\, \tanh\left(\pi b^{-1}(z - ka)\right)\big|\leq \frac{M_3}{1 \wedge d\left(z,P_\sigma\right) }e^{-\eta\Re(z)/b}<\infty,
\end{equation}
for $z\in\C\setminus P_\sigma$. Thus, the series \eqref{eq:Zab-series} converges absolutely uniformly on compact subsets of $\dom_{\sigma}$, and  $\sigma$ is therefore holomorphic on $\dom_{\sigma}$. As the summand functions $\sgn(k)+\, \tanh\left(\pi b^{-1}(\,\cdot - ka)\right)$ are meromorphic with simple poles at $ka+P_b$, items (i), (ii), and (iii) of the lemma follow immediately. Finally, item (iv) follows from \eqref{eq:Zab-props-1}, \eqref{eq:Zab-props-2}, and the fact that $e^{\pm\eta\Re(z)/b}\leq e^{\eta |z|/b}$, for all $z\in\C$.
\end{proof}

\begin{proof}[Proof of Lemma \ref{lem:d>0->Q}]
The proof relies on the following special case of the Kronecker-Weyl equidistribution theorem:

\begin{prop}[Kronecker-Weyl {\cite{Weyl1916}, \cite{Beck2017}}]\label{prop:KW-equi}
Let $x_1,x_2\in[0,1)$ and $y_1,y_2\in\R\setminus\{0\}$ be such that $\|(y_1,y_2)\|_2=1$. Furthermore, for a Jordan measurable set $J\subset [0,1)\times[0,1)$, define
\begin{equation*}
S_J=\{t\in\R: (x_1+ty_1\; \mrm{mod}\; 1, \; x_2+ty_2 \; \mrm{mod}\; 1)\in J\}.
\end{equation*}
If $y_1/y_2$ is irrational, then
\begin{equation*}
\lim_{T\to\infty}\frac{1}{2T}\,\mu(S_J\cap[-T,T])= \mrm{A}(J),
\end{equation*}
where $\mu$ denotes the Lebesgue measure on $\R$ and $\mrm{A}$ stands for the Lebesgue measure on $[0,1)\times [0,1)$.
\end{prop}

Turning back to the proof of the lemma, in order to show that $\ell\cap P$ is an arithmetic sequence, it suffices to find a $z_0\in \Pi$ and a pair $(n_a,n_b)\in \Z\times\Z\setminus\{(0,0)\}$ such that 
\begin{equation*}
\ell'\cap \Pi= \{z_0+(n_a a+in_b\, b)k:k\in\Z\},
\end{equation*}
where $\ell'=\beta\ell +\gamma$. We begin by noting that
\begin{equation*}
\begin{aligned}
\Delta_\varepsilon(\ell,P)&=\limsup_{N\to\infty}\frac{1}{2N}\, \#\{p\in \beta^{-1}(\Pi-\gamma)\cap D(0,N):d(p,\ell)\leq \varepsilon \}\\
&\stackrel{\mathclap{p'=\beta p+\gamma}}{=}\quad \limsup_{N\to\infty}\frac{1}{2N}\, \#\{p'\in \Pi \cap D(\gamma,|\beta| N): d(p' ,\beta \ell+\gamma)\leq |\beta| \varepsilon \}\\
&=|\beta|\cdot \limsup_{N\to\infty}\frac{1}{2|\beta|N}\, \#\{p'\in \Pi \cap D(0,|\beta| N): d(p' ,\ell')\leq |\beta| \varepsilon \}\\
&=|\beta| \Delta_{|\beta|\varepsilon}(\ell',\Pi),
\end{aligned}
\end{equation*}
 for all $\varepsilon>0$, and therefore $\Delta(\ell',\Pi)=|\beta|^{-1}\Delta(\ell,P)>0$.
 
 Next, let $x\in\C$ and $y\in\C\setminus\{0\}$ be such that $\ell'=\{x+ty:t\in\R\}$. Assume w.l.o.g. that $|y|=1$, and write $x=x_1+ix_2$, $y=y_1+iy_2$, where $x_1,x_2,y_1,y_2\in\R$. We claim that $y_1/a$ and $y_2/b$ are rationally dependent, i.e., there exists a pair $(n_a,n_b)\in \Z\times\Z\setminus\{(0,0)\}$ such that $n_b\frac{y_1}{a}-n_a\frac{y_2}{b}=0$. Suppose by way of contradiction that this is not the case, i.e., $y_1\neq0$, $y_2\neq 0$, and $\frac{y_1}{a}/\frac{y_2}{b}$ is irrational. Fix an $\varepsilon\in\left(0,\frac{1}{4}\min\{a,b\}\right)$, and, for $N>0$, set
\begin{equation*}
P_\varepsilon^N=\left\{p\in \Pi \cap D(0,N): d(p , \ell') \leq \varepsilon \right\}.
\end{equation*}
Consider now an arbitrary $p \in P_\varepsilon^\infty$, and write $p=p_a a+ip_b b$, where $p_a,p_b\in\Z$. Select a $t_p\in\R$ such that $|x+t_py-p|\leq \varepsilon$. Then, for all $t\in I_p\coleqq [t_p-\varepsilon,t_p+\varepsilon]$, we have
\begin{equation*}
|x+ty-p|\leq |x+t_py-p|+|t-t_p||y|\leq 2\varepsilon,
\end{equation*}
and so
\begin{equation*}
\left\|\big({\textstyle\frac{x_1}{a}}+t{\textstyle\frac{y_1}{a}}, \; {\textstyle \frac{x_2}{b}}+t{\textstyle \frac{y_2}{b}}\big)-(p_a,p_b)\right\|_2\leq \min\{a,b\}^{-1}|x+ty-p|\leq 2\min\{a,b\}^{-1} \varepsilon.
\end{equation*}
Therefore, defining
\begin{equation*}
\begin{aligned}
J_\varepsilon&=\{(u_1 \;\mrm{mod}\; 1,\,u_2 \;\mrm{mod}\; 1):(u_1,u_2)\in\R^2,\,\|(u_1,u_2)\|_2\leq 2\min\{a,b\}^{-1} \varepsilon\}\quad\text{and}\\
S_{\varepsilon}&=\left\{t\in\R: \big( ({\textstyle\frac{x_1}{a}}+t{\textstyle\frac{y_1}{a}}) \; \mrm{mod}\; 1, \; ({\textstyle \frac{x_2}{b}}+t{\textstyle \frac{y_2}{b}} ) \; \mrm{mod}\; 1\big) \in J_\varepsilon\right\},
\end{aligned}
\end{equation*}
we have $ I_p \subset S_\varepsilon$. We next show that $I_p\cap I_{p'}=\varnothing$, for distinct elements $p$ and $p'$ of $P_\varepsilon^\infty$. Indeed, for such $p$, $p'$, and for $t\in I_p$, $t'\in I_{p'}$, we have
\begin{equation*}
\begin{aligned}
|t-t'|&=|(x+ty)-(x+t'y)|\\
&\geq |p-p'|-|x+ty-p|-|x+t'y-p'| \geq \min\{a,b\} -2\varepsilon-2\varepsilon>0,
\end{aligned}
\end{equation*}
where the last inequality is by our choice of $\varepsilon$. Therefore, $t\neq t'$, and, as $t\in I_p$ and $t'\in I_{p'}$ were arbitrary, we deduce  $I_p\cap I_{p'}=\varnothing$. Next, as $\Delta_\varepsilon(\ell',\Pi)\geq \Delta(\ell',\Pi)>0$, by definition of $\Delta_\epsilon$,
there exists a sequence of positive reals $\{N_k\}_{k\in\N}$ increasing to $\infty$ such that
\begin{equation*}
\# \left(P_\varepsilon^{N_k} \right)\geq 2N_k \cdot \frac{1}{2}\Delta(\ell',\Pi)=N_k\Delta(\ell',\Pi),
\end{equation*}
for all $k\in\N$. Then, for $k\in\N$ such that $N_k\geq |x|+2\varepsilon$ and $p\in P_\varepsilon^{N_k}\subset P_\varepsilon^\infty$, we have
\begin{equation*}
|t_p|+\varepsilon=|(x+t_py-p)+p-x|+\varepsilon \leq \varepsilon + N_k + |x|+\varepsilon \leq  2N_k,
\end{equation*}
and therefore $I_p\subset [-2N_k,2N_k]$. Thus, as the intervals $I_p$ are disjoint for distinct $p$, we obtain
\begin{equation*}
\begin{aligned}
\mu\left(S_\varepsilon\cap [-2N_k,2N_k]\right)& \geq \mu\Big(\bigcup_{p\in P_\varepsilon^{N_k}} I_p\Big)=\# \left(P_\varepsilon^{N_k} \right) \cdot 2\varepsilon\geq 2 N_k\Delta(\ell',\Pi)\, \varepsilon.
\end{aligned}
\end{equation*}
Now, as $J_\varepsilon$ is a union of 4 circular sectors, it is Jordan measurable and its area is given by $\mrm{A}(J_\varepsilon)=\pi (2\min\{a,b\}^{-1} \varepsilon)^2$. Proposition \ref{prop:Weyl-equi} therefore implies that
\begin{equation*}
\pi (2\min\{a,b\}^{-1} \varepsilon)^2=\mrm{A}(J_\varepsilon)=\lim_{k\to\infty}\frac{1}{4N_k} \,\mu\left(S_\varepsilon\cap [-2N_k,2N_k]\right)\geq \frac{1}{2}\Delta(\ell',\Pi)\, \varepsilon, 
\end{equation*}
and thus $\Delta(\ell',\Pi)\leq 8\pi\min\{a,b\}^{-2}\varepsilon$. But $\varepsilon\in\left(0,\frac{1}{4}\min\{a,b\}\right)$ was arbitrary, so we must have $\Delta(\ell',\Pi)=0$. This constitutes a contradiction, and so our assumption that $y_1/a$ and $y_2/b$ are rationally independent must be false.

We can thus find $(n_a,n_b)\in \Z\times\Z\setminus\{(0,0)\}$ such that $n_b\frac{y_1}{a}-n_a\frac{y_2}{b}=0$. Moreover, in the case when one of $y_1$ or $y_2$ is zero, we take  $(n_a,n_b)\in\{(0,1),(1,0)\}$, and if $y_1$ and $y_2$ are both nonzero, we assume w.l.o.g. that $n_a$ and $n_b$ are coprime. Then, letting $K=n_aa+in_bb$, we have
\begin{equation*}
\ell'=\bigcup_{k\in\Z}\{x+sy+kK:s\in[0,|K|] \} \subset \bigcup_{s\in[0,|K|]}x+sy+(n_aa)\Z\times(in_bb)\Z.
\end{equation*}
Let $d(F_1,F_2)=\inf\{|f_1-f_2|:f_j\in F_j,j\in\{1,2\}\}$ denote the Euclidean distance between two closed sets $F_1,F_2\subset \C$. As $\Delta(\ell',\Pi)>0$, we must have $d(\ell',\Pi)=0$, and therefore
\begin{equation*}
\inf_{s\in[0,|K|]}d(x+sy+(n_aa)\Z\times(in_bb)\Z,\Pi)\leq d(\ell',\Pi)=0.
\end{equation*}
Now, as $s\mapsto d(x+sy+(n_aa)\Z\times(in_bb)\Z,\Pi)$ is continuous, and $[0,|K|]$ is compact, there must exist an $s_0\in [0,|K|]$ such that $d(x+s_0y+(n_aa)\Z\times(in_bb)\Z,\Pi)=0$. Then, letting $z_0=x+s_0y$, we have that $z_0+(n_aa)\Z\times(in_bb)\Z$ is an affine sublattice of $\Pi$, and therefore, as $(n_a,n_b)\in\{(0,1),(1,0)\}$ or $n_a$ and $n_b$ are coprime, we have
\begin{equation*}
\ell'\cap \Pi= \left\{z_0+(n_a a+in_b\, b)k:k\in\Z\right\},
\end{equation*}
as desired.

It remains to show that there exists an $\varepsilon_0>0$ such that $(\ell+D(0,\varepsilon_0))\cap P=\ell\cap P$. To this end, let $Y=\{z_0+tK:t\in[0,1]\}$, and note that, as $\ell'$ and $\Pi$ are $K$--periodic, we get
\begin{equation*}
\begin{aligned}
\inf\{d(p,\ell):p\in P\setminus\ell\} &=\inf\{d(p,\beta^{-1} (\ell' -\gamma)):p\in P\setminus\ell\}\\
&= |\beta|^{-1} \inf\{d(p',\ell'):p'\in\Pi\setminus\ell'\}\\
&= |\beta|^{-1} \inf\{d(p',kK+Y):p'\in\Pi\setminus\ell',k\in\Z\}\\
&=  |\beta|^{-1} \inf\{d(p'-kK,Y):p'\in\Pi\setminus\ell',k\in\Z\}\\
&=  |\beta|^{-1} \inf\{d(p',Y):p'\in\Pi\setminus\ell'\}.\\
\end{aligned}
\end{equation*}
Note that the last expression is strictly positive, as $\Pi\setminus \ell'$ is closed, $Y$ is compact, and $(\Pi\setminus \ell')\cap Y=\varnothing$. Let $\varepsilon_0>0$ be such that $\varepsilon_0<|\beta|^{-1} \inf\{d(p',Y):p\in\Pi\setminus\ell'\}$. Now, if $p^*\in P$ is such that $d(p^*,\ell)\leq \varepsilon_0$, then 
\begin{equation*}
d(p^*,\ell)<\inf\{d(p,\ell):p\in P\setminus\ell\},
\end{equation*}
and so $p^*\in\ell$. As $p^*$ was arbitrary, this shows that $(\ell+D(0,\varepsilon_0))\cap P=\ell\cap P$ and thus completes the proof.
\end{proof}

\begin{proof}[Proof of Lemma \ref{lem:2-lattice-fit}]
Let $c>0$ be arbitrary, and let $m_1,m_2\in\Z$ be such that $B\coleqq m_1y_1=m_2y_2$. Then the sets $P_1$ and $P_2$ are both $B$-periodic, so $P_{\ell,c}$ is $B$-periodic. Now, for a point $z\in\C$, let $\pi_\ell(z)$ be the orthogonal projection of $p$ onto $\ell$.
Then, writing $Y=\{p\in P_{\ell,c}:\pi_\ell (p)\in \{x_1+\, tB:t\in[0,1]\}\}$, we have
\begin{equation*}
\begin{aligned}
&\inf\{|p-p'|:p,p'\in P_{\ell,c}, \;p\neq p'\}\\
=&\inf\{|p+kB-(p'+k'B)|:p,p'\in Y,\,k,k'\in\Z,\, \;p+kB\neq p'+k'B\}\\
=&\inf\{|p-p'|:p,p'\in Y\cup (Y+B),\, p\neq p'\}=:\eta(c).
\end{aligned}
\end{equation*}
As the infimum in the last quantity is taken over a finite set of positive numbers, we have $\eta(c)>0$. Therefore, $P_{\ell,c}$ is uniformly discrete, for all $c>0$, as desired.
\end{proof}

\begin{proof} [Proof of Lemma \ref{lem:ent-vs-line}] 
Let $\m{P}\coleqq\{(\alpha_s,\beta_s,\gamma_s)\}_{s\in\m{I}}$ be as in the statement of the lemma.
We define an equivalence relation $\sim_\Q$ on $\m{P}$ by setting $(\alpha_{s_1},\beta_{s_1},\gamma_{s_1})\sim_\Q (\alpha_{s_2},\beta_{s_2},\gamma_{s_2})$ if and only if $\beta_{s_1}/\beta_{s_2}\in\Q$. Let $\Psi(\m{P})$ denote the set of equivalence classes of $\m{P}$ with respect to $\sim_\Q$.
We proceed with the proof of the lemma by induction on $n\coleqq \#(\Psi(\m{P}))$. If $n=0$, i.e., $\m{P}=\varnothing$, then $f$ is given by the empty sum, and so $f=0$ is trivially entire, as desired. Next, suppose that $n\geq 1$, and that the statement of the lemma holds for all functions parametrized by $\m{P}'=\{(\alpha_s',\beta_s',\gamma_s')\}_{s\in\m{I}'}$ with $\#(\Psi(\m{P}'))<n$.

Let $P_f$ be the (possibly empty) set of poles of $f$, and assume that $\Delta(\ell,P_f)=0$, for every line $\ell$ in $\C$. We show that then $f$ must be entire. To this end, fix an equivalence class $\m{P}_1\coleqq \{(\alpha_s,\beta_s,\gamma_s)\}_{s\in\m{I}_1}\in \Psi(\m{I})$, and note that, as $\beta_{s_1}/\beta_{s_2}\in\Q$, for all $s_1,s_2\in \m{I}_1$, there exists a $T\in \C$ such that $\beta_s T\in\Z$, for all $s\in\m{I}_1$. Next, define $g= f(\,\cdot+ibT)-f$ and let $P_g\subset (P_f-ibT)\cup P_f$ be its set of poles. 
Then, as $\sigma$ is $ib$-periodic, we have
\begin{equation*}
\begin{aligned}
g=f(\,\cdot+ibT)-f &= \sum_{s\in\m{I}}\left (\alpha_s\,\sigma\Big(\beta_s \cdot+ ib(\beta_sT) +\gamma_s\Big)- \alpha_s\,\sigma(\beta_s \cdot+\gamma_s)\right)\\
&= \sum_{s\in\m{I}\setminus \m{I}_{1}} \left (\alpha_s\,\sigma\Big(\beta_s \cdot+  ib(\beta_sT) + \gamma_s\Big)- \alpha_s\,\sigma(\beta_s \cdot+\gamma_s)\right),
\end{aligned}
\end{equation*}
and so $g$ takes the form
$
g=\sum_{s\in\m{I}'}\alpha_s'\,\sigma(\beta_s' \cdot+\gamma_s')
$,
for some $\m{P}'\coleqq \{(\alpha_s',\beta_s',\gamma_s')\}_{s\in\m{I}'}$ such that $\{\beta_s'\}_{s\in\m{I}'}\subset\{\beta_s\}_{s\in\m{I}\setminus\m{I}_1}$. But then $\#(\Psi(\m{P}'))= \#(\Psi(\m{P}\setminus\m{P}_1))< \#(\Psi(\m{P}))=n$, so it follows by the induction hypothesis that either $g$ is entire, or $\Delta(\ell,P_g)>0$ for some line $\ell$ in $\C$. On the other hand, as we assumed that $\Delta(\ell,P_f)=0$, for every line $\ell$ in $\C$, we have
\begin{equation*}
\begin{aligned}
\Delta(\ell,P_g)&\leq \Delta(\ell,(P_f-ibT)\cup P_f)\\
&\leq \Delta(\ell,P_f-ibT)+\Delta(\ell,P_f)=\Delta(\ell+ibT,P_f)+\Delta(\ell,P_f)=0,
\end{aligned}
\end{equation*}
for every line $\ell$ in $\C$, and therefore $g$ must be entire.
We are now ready to show that $f$ must also be entire. To see this, suppose by way of contradiction that there exists a $p^*\in P_f$. Then, as $f(\,\cdot+ibT)-f=g$ is analytic at $p^*$ and $p^*-ibT$, we also have $p^*+ibT,p^*-ibT\in P_f$, and thus, applying this argument repeatedly, we conclude that $\{p^*+ibTk:k\in\Z\}\subset P_f$. Therefore, letting $\ell'=\{p^*+ibTt:t\in\R\}$, we obtain $\Delta(\ell',P_f)\geq (b|T|)^{-1}>0$, contradicting our assumption that $\Delta(\ell,P_f)=0$, for every line $\ell$ in $\C$. This establishes that $f$ is entire and concludes the proof of the lemma.
\end{proof}

\end{document}